\documentclass[12pt,oneside,english]{amsart}
\usepackage[T1]{fontenc}
\usepackage[latin9]{inputenc}
\usepackage[pdftex,letterpaper]{geometry}
\usepackage{babel}
\usepackage{refstyle}
\usepackage{float}
\usepackage{mathtools}
\usepackage{enumitem}
\usepackage{amstext}
\usepackage{amsthm}
\usepackage{amssymb}
\usepackage{mathdots}
\usepackage{stackrel}
\usepackage[pdftex]{graphicx}
\graphicspath{{../pdf/}{C:/Users/Juan/Pictures/}}
\DeclareGraphicsExtensions{.pdf,.jpeg,.png}

\usepackage{esint}
\usepackage[authoryear]{natbib}

\makeatletter

\AtBeginDocument{\providecommand\figref[1]{\ref{fig:#1}}}
\AtBeginDocument{\providecommand\partref[1]{\ref{part:#1}}}
\AtBeginDocument{\providecommand\remref[1]{\ref{rem:#1}}}
\AtBeginDocument{\providecommand\lemref[1]{\ref{lem:#1}}}
\AtBeginDocument{\providecommand\propref[1]{\ref{prop:#1}}}
\AtBeginDocument{\providecommand\assuref[1]{\ref{assu:#1}}}
\AtBeginDocument{\providecommand\defref[1]{\ref{def:#1}}}
\AtBeginDocument{\providecommand\secref[1]{\ref{sec:#1}}}
\AtBeginDocument{\providecommand\eqref[1]{\ref{eq:#1}}}
\AtBeginDocument{\providecommand\algref[1]{\ref{alg:#1}}}
\AtBeginDocument{\providecommand\thmref[1]{\ref{thm:#1}}}
\AtBeginDocument{\providecommand\subsecref[1]{\ref{subsec:#1}}}
\AtBeginDocument{\providecommand\tabref[1]{\ref{tab:#1}}}
\providecommand{\tabularnewline}{\\}
\RS@ifundefined{subsecref}
  {\newref{subsec}{name = \RSsectxt}}
  {}
\RS@ifundefined{thmref}
  {\def\RSthmtxt{theorem~}\newref{thm}{name = \RSthmtxt}}
  {}
\RS@ifundefined{lemref}
  {\def\RSlemtxt{lemma~}\newref{lem}{name = \RSlemtxt}}
  {}

\numberwithin{equation}{section}
\numberwithin{figure}{section}
  \theoremstyle{remark}
  \newtheorem{rem}{\protect\remarkname}[section]
  \theoremstyle{plain}
  \newtheorem{lem}{\protect\lemmaname}[section]
  \theoremstyle{plain}
  \newtheorem{prop}{\protect\propositionname}[section]
  \theoremstyle{definition}
  \newtheorem{defn}{\protect\definitionname}[section]
  \theoremstyle{plain}
  \newtheorem{assumption}{\protect\assumptionname}
  \theoremstyle{plain}
  \newtheorem{lyxalgorithm}{\protect\algorithmname}
  \theoremstyle{plain}
  \newtheorem{thm}{\protect\theoremname}[section]

\AtBeginDocument{\providecommand\assuref[1]{\ref{assu:#1}}}
\AtBeginDocument{\providecommand\defref[1]{\ref{def:#1}}}

\usepackage{bbm}
\usepackage{times}
\usepackage{amsfonts}
\allowdisplaybreaks

\numberwithin{equation}{subsection}

\renewcommand{\[}{\begin{equation}}
\renewcommand{\]}{\end{equation}}

\newref{alg}{name=algorithm~,Name=Algorithm~,names=algorithms~,Names=Algorithms~}
\newref{rem}{name=remark~,Name=Remark~,names=remarks~,Names=Remarks~}
\newref{fig}{name=figure~,Name=Figure~,names=figures~,Names=Figures~}
\newref{subsec}{name=subsection~,Name=Subsection~,names=subsections~,Names=Subsections~}
\newref{def}{name=definition~,Name=Definition~,names=definitions~,Names=Definitions~}
\newref{assu}{name=assumption~,Name=Assumption~,names=assumptions~,Names=Assumptions~}
\newref{lem}{name=lemma~,Name=Lemma~,names=lemmas~,Names=Lemmas~}
\newref{cor}{name=corollary~,Name=Corollary~,names=corollaries~,Names=Corollaries~}
\newref{claim}{name=claim~,Name=Claim~,names=claims~,Names=Claims~}
\newref{prop}{name=proposition~,Name=Proposition~,names=propositions~,Names=Propositions~}
\newref{thm}{name=theorem~,Name=Theorem~,names=theorems~,Names=Theorems~}
\newref{sum}{name=summary~,Name=Summary~,names=summaries~,Names=Summaries~}

\usepackage{babel}
\providecommand{\assumptionname}{Assumption}
\providecommand{\definitionname}{Definition}
\providecommand{\lemmaname}{Lemma}
\providecommand{\remarkname}{Remark}
\providecommand{\theoremname}{Theorem}

\usepackage{thmtools}

\makeatother

  \providecommand{\algorithmname}{Algorithm}
  \providecommand{\assumptionname}{Assumption}
  \providecommand{\definitionname}{Definition}
  \providecommand{\lemmaname}{Lemma}
  \providecommand{\propositionname}{Proposition}
  \providecommand{\remarkname}{Remark}
\providecommand{\theoremname}{Theorem}

\begin{document}

\title{Intervention On Default Contagion Under Partial Information}

\author{Yang Xu}
\begin{abstract}
We model the default contagion process in a large heterogeneous financial network under the interventions of a regulator (a central bank) with only partial information which is a more realistic setting than most current literature. We provide the analytical results for the asymptotic optimal intervention policies and the asymptotic magnitude of default contagion in terms of the network characteristics. We extend the results of \citet{Amini2013} to incorporate interventions and the model of \citet{Amini2015,Amini2017} to heterogeneous networks with a given degree sequence and arbitrary initial equity levels. The insights from the results are that the optimal intervention policy is "monotonic"  in terms of the intervention cost, the closeness to invulnerability and connectivity. Moreover, we should keep intervening on a bank once we have intervened on it. Our simulation results show a good agreement with the theoretical results.
\end{abstract}

\keywords{random network process, configuration model, financial stability,
partial information, macroprudence}

\maketitle
\tableofcontents{}

Important formulations and programs
\begin{center}
\begin{tabular}{|c|c|}
\hline 
Formulations and programs & Page\tabularnewline
\hline 
\hline 
Asymptotic control problem (\ref{eq:ACP}) & \pageref{eq:ACP}\tabularnewline
\hline 
Optimal control problem (\ref{eq:NonAOCP}) & \pageref{eq:NonAOCP}\tabularnewline
\hline 
Optimization problem (\ref{eq:OptProgram}) & \pageref{eq:OptProgram}\tabularnewline
\hline 
\end{tabular}
\par\end{center}

Important index sets

$\Gamma\coloneqq\{(i,j,c,l):0\leq i,0\leq j,0\leq l<c\leq i\text{ or }c=i+1,l=i\},$

$\Phi\coloneqq\{(i,j,c,c-1):0\leq i,0\leq j,1\leq c\leq i\}.$

For some integer $M^{\epsilon}$, 

$\Gamma^{\epsilon}\coloneqq\{(i,j,c,l):i\lor j<M^{\epsilon},0\leq l<c\leq i\text{ or }c=i+1,l=i\},$

$\Phi^{\epsilon}\coloneqq\{(i,j,c,c-1):i\lor j<M^{\epsilon},0\leq c\leq i\}.$

\newpage{}

\part{Introduction}

\section{Introduction}

The systemic risk defined as the large scale defaults of financial
institutions in a financial network has been drawing more and more
interest of regulators and researchers in the recent decade, especially
after the Asian financial crsises in the late 1990's and the more
recent economic recession in 2008 and 2009. In a modern financial
system, financial institutions (hereafter, banks), connected through
lending and borrowing relations constitute a financial network. Due
to the intricate nature of the interconnectedness, the default of
some banks in the network may lead to losses of creditor banks through
interbank connections, which may in turn results in losses of their
creditors. This process goes on and creates risk at the system level. 

We model the contagion process in a financial network under a short
term illiquidity risk. In this setting no banks would want to lend
new loans to other banks meanwhile they are still obliged to pay back
their current loans which are due in the time frame of the model.
So we fix the in and out degrees of banks in the network, and set
up a probability space under which the financial network is generated
by a uniform matching of the in and out degrees (a configuration model).
A directed link in this model represents one unit of loan. In the
following we may use ``bank'' and ``node'' interchangeably. 

Before an external shock to the system, each node has a positive equity
level, which indicates the number of lost loans a node can withstand
due to the default of its debtors before it defaults. In other words,
it is the ``distance to default''. After an external shock, some
nodes in the system default initially and we set their equity level
to zero. We consider the default contagion in the following way. When
a node defaults, it defaults on all of its loan liabilities. We assume
a zero recovery rate of the loan, i.e. the creditor receives zero
value from the loan, which is the most realistic assumption for short
term default as suggested in \citet{Cont2010a} and \citet{Amini2013}.
But there is a time span between a node's default and the time its
creditor records the loan as a loss (by writing down the loan from
its balance sheet). We model this time span by independently, identically
exponentially distributed random variables. After the affected node
records the lost loan, it may request the regulator for interventions.
This corresponds to the central bank's provision of short term liquidity
to banks, including the traditional discount window, the Term Auction
Facility (TAF), etc, as well as the government's interventions by
bailing out the stressed banks. If the regulator decides to intervene
by infusing one unit of equity, the equity level of the affected node
will stay the same, otherwise its equity level will decrease by one.
Once the equity level reaches zero, the node defaults. We assume that
the once a bank has defaulted it cannot become liquid again within
the time horizon of the model because it is very unlikely for a bank
that has declared default to gain enough credits in a short term as
considered in the model. During the contagion process, the regulator
knows the default set (the set of defaulted nodes) and the set of
revealed out links from the default set, but other out links from
the default set are hidden until the affected nodes record the defaulted
loans.

\begin{figure}[H]
\centering{}\includegraphics[scale=0.7]{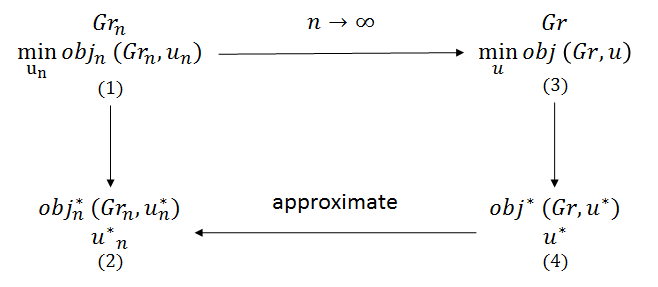}\caption{\label{fig:Methodology}The methodology: approximation of the finite
network with the infinite network}
\end{figure}

Our methodology is illustrated in \figref{Methodology}. The regulator's
goal is to minimize the number of final defaulted nodes with the minimum
number of interventions, so we obtain a stochastic control problem
$\min_{u_{n}}\text{obj}_{n}(Gr{}_{n},u_{n})$ where the objective
function depends on the graph $Gr_{n}$ and the intervention policy
$u_{n}$, shown in (1). We aim to solve it for the optimal intervention
policy $u_{n}^{*}$ and thus obtain the optimal objective function
value $\text{obj}_{n}(Gr_{n},u_{n}^{*})$, shown in (2). However,
solving the problem with the usual dynamic programming approach will
incur intractability problem because of the fast expansion of the
state space as in \citet{Amini2015,Amini2017}, not to mention the
heterogeneous network in our setting. We take an alternative approach
based on the fact that under some regularity conditions, the objective
function converges as $n\rightarrow\infty$. We solve the asymptotic
optimal control problem $\min_{u}\text{obj}(Gr,u)$ in (3) where $\text{obj}$,
$Gr$ and $u$ are the limit forms of $\text{obj}_{n}$, $Gr_{n}$
and $u_{n}$, respectively, and obtain the optimal policy $u^{*}$
and the objective function $\text{obj}(Gr,u^{*})$ in (4), then we
will be able to construct the optimal intervention policy $u_{n}^{*}$
for a finite $n$ through $u^{*}$ and approximate $\text{obj}_{n}(Gr_{n},u_{n}^{*})$
with $\text{obj}(Gr,u^{*})$. Our results of the numerical experiments
validate the approximation for networks with sizes close to the real
financial networks. 

\section{Relations to previous literature}

To understand the systemic risk the existing literature tries to relate
the systemic propagation of financial stress to various features of
the financial network and the banks within it, including the degree
of connectivity, the equity levels and so on. There are mainly two
types of literature: empirical and theoretical. The empirical studies
conduct statistical analysis on the interbank markets using data on
interbank lending and borrowing as far as they are available and provide
an overview of the structural characteristics of the interbank network
in different countries (\citet*{Furfine2003,Cont2010,Boss2004a}).
The theoretical studies model the financial network with network models
but differ in their assumptions of the network structure and approaches,
some focus on ``stylized'' networks whose structures are hypothetical
(\citet{Allen2000,May2010,Haldane2011,Gai2011}) while others reply
on simulations (\citet*{Nier2007a,Cont2010}). Among them, \citet{Amini2013}
and \citet{Eisenberg2001} propose random network models that allow
more realistic, heterogeneous structures. From the regulatory perspective,
there are mainly two shortcomings in the studies. While it is important
to understand the dependence of systemic risk on network features,
it is still unclear in the system risk literature on how the regulator
(such as the central bank) can do optimally to intervene on the contagion
of financial stress in the network. Moreover, the majority of studies
assume the regulator has complete information of the network which
does not align with the reality, as for example pointed out by \citet{Hurd2015}:
``\textit{Interbank exposure data are never publicly available, and
in many countries nonexistent even for central regulators''. }Although
there are some attempts for considering partial information in a financial
network (\citet{Song2015,Battiston2012,Amini2015}), they either focus
on simplified networks or lack analytical results of the regulator's
optimal strategies.

In particular, our work is an extension of the papers \citet{Amini2013},
\citet{Amini2015} and \citet{Amini2017}. \citet{Amini2013} model
a financial network as a weighted directed graph where the weight
represents the amount of loan liabilities between two banks. Then
they show that under some conditions a random weighted directed graph
has the same law as a unweighted directed graph with a default threshold
on each node and propose a sequential construction algorithm to model
the development of the default set. The default contagion model in
our work extends the construction algorithm to incorporate interventions.
Thus if there are no interventions, the asymptotic fraction of defaults
will be the same as in \citet{Amini2013}.

\citet{Amini2015,Amini2017} consider a stylized core-periphery financial
network where the core consists of identical nodes, i.e. each core
node having the same in and out degrees and initial equity level if
it is liquid or zero otherwise, and each periphery node having in
degree of one and one initial equity. We consider a heterogeneous
network where the nodes have arbitrary in and out degrees and initial
equity levels. We adopted the dynamics of their model which constructs
the default set with interventions through a configuration model due
to the fact that the configuration model can be adapted to modeling
the contagion process as described in \citet{VanderHofstad2014}. 

\citet{Amini2015} solve a stochastic control problem for a finite
network in which the regulator seeks to maximize under budget constraints
the value of the financial system defined as the total number of projects.
\citet{Amini2017} also deals with a finite network with the objective
to minimize under some budget the expected loss at the end of the
cascade. They take the pair (remaining equity level of banks, the
number of interventions on banks with the same remaining equity) as
the system state and apply the dynamic programming approach. The limitation
of the approach is that the state space expands very fast as the type
of banks and the total number of interventions increase, rendering
the problem intractable. In \citet{Amini2017} they present an example
of a network of $20$ core banks, each having initial equity level
$3$ and $20$ periphery banks and total $12$ interventions, the
state space is about $10^{9}$. Motivated by the intention to overcome
the limitations, we solve an optimal control problem to approximate
the solution of a stochastic control problem in which the regulator
seeks to minimize the number of final defaults with the minimum number
of interventions, and give analytical formulations for the optimal
intervention policies as well as the asymptotic number of interventions
and final defaulted banks. The asymptotic results provide a good approximation
to real financial networks, which are heterogeneous and have several
hundreds of banks thanks to the fast convergence behavior of our results.
We have run numerical experiments and validated the convergence behavior. 

\section{Contributions}

The main contribution of our work is that we derive rigorous asymptotic
results of the optimal strategy for the regulator and the fraction
of final defaults for a heterogeneous network with a given degree
sequence and arbitrary initial equity levels. The analytical expressions
are expressed in terms of measurable features of the network. For
convergence of the results we assume the network is sparse which are
supported by the empirical studies of real financial networks (e.g.
\citet{Furfine1999,Bech2010}) and used by some theoretical studies
(e.g. \citet{Amini2013}) .

The key insight of our findings of the optimal strategies is that
we should only consider intervention on a bank when it is very close
to default, so we never intervene on an invulnerable bank. The optimal
intervention policy depends strongly on the intervention cost. The
smaller intervention cost, the more interventions are implemented.
Moreover, the optimal intervention policy is ``monotonic'' with
respect to the measurable features of the network: We should not intervene
on the banks with out degrees in a certain range regardless of their
other features. For those banks worth interventions, the larger the
sum of initial equity and accumulative interventions received, the
earlier we should begin intervening if they are affected. Moreover,
the time to start intervention on a node is also ``monotonic'' in
its in and out degree. Once we begin intervening on a node we keep
intervening on it every time it is affected. By comparing the fractions
of final defaults under no interventions and the optimal intervention
policy, we are able to quantify the improvement made by interventions
in terms of the network features. This gives guidance as the maximum
impact the regulator can have to offset the contagion of defaults. 

The paper is organized as follows. We set up the model and introduce
the stochastic control problem (SCP) in \partref{Model-Description}.
In \partref{The-Asymptotic-Control} we formulate the asymptotic control
problem (ACP) that gives the limit for the objective function of the
SCP as the size of the network goes to infinity and present the necessary
conditions for the optimal intervention policy, which lead to the
main theorems. At last we show the results of the numerical experiments
to validate the approximation of ACP to SCP. 

\part{\label{part:Model-Description}Model Description and Dynamics}

\section{Basic setup\label{sec:Basic-setup}}

We model a financial network with prescribed degree sequence $(d^{-}(v),d^{+}(v))_{v\in[n]}$
as an unweighted directed network $([n],\mathcal{E}_{n})$, where
$[n]=\{1,\ldots,n\}$ is the set of nodes, $\mathcal{E}_{n}$ denotes
the set of links, and $m=\sum_{v\in[n]}d^{-}(v)=\sum_{v\in[n]}d^{+}(v)$.
We understand the prescribed degree sequence $(d^{-}(v),d^{+}(v))_{v\in[n]}$
as a set of random variables living in a probability space then we
set up the following model conditioning on it. A directed link $(v,w)\in\mathcal{E}_{n}$
represents $v$ borrows a unit of loan from $w$, i.e. $v$ is obliged
to repay $w$ one unit of loan. We allow multiple loans to exist between
two nodes. $\mathcal{E}_{n}$ needs to satisfy that
\begin{align*}
d^{-}(v) & =|\{(w,v):(w,v)\in\mathcal{E}_{n}\}|,\\
d^{+}(v) & =|\{(v,w):(v,w)\in\mathcal{E}_{n}\}|,
\end{align*}
where $\left|A\right|$ represents the cardinality of the set $A$.
Now we set up a probability space $(G_{n,m},\mathbb{P})$ where $G_{n,m}$
is the set of networks on $n$ nodes with at most $m$ directed links.
Recall $m$ is total in or out degree of the network. So the random
financial network lives in this probability space and under $\mathbb{P}$
the law of the random link set $\mathcal{E}_{n}$ is determined as
follows. We start with $n$ unconnected nodes and assign node $v$
with $d^{-}(v)$ in half links and $d^{+}(v)$ out half links. An
in half link represents an offer of a loan and an out half link a
demand for a loan. Then the $m$ in half links and $m$ out half links
are matched uniformly so that the borrowers and lenders are determined.
The resulting random network is called the \textit{configuration model}. 
\begin{rem}
\label{rem:ConfigModel}The uniform matching of the in and out half
links allows us to construct the random network sequentially: at every
step we can choose any out half link by any rule and choose the other
in half link uniformly over all available in half links to form a
directed link. This is because conditional on any set of observed
matched links, the hidden matched links also follow the uniform distribution.
Moreover, the conditional law of hidden matched links only depends
on the number of the observed matched links, not the matching history.
Additionally we can restrict the matching to out half links from the
defaulted nodes so that we can model the development of the set of
defaulted nodes with their revealed out links.

Then we endow a node $v\in[n]$ with its initial equity level $e_{0}^{v}\in\mathbb{N}_{0}$
which represents the number of lost loans $v$ can tolerate until
$v$ defaults, so it is the ``distance to default''. Next after
the system receives some external shock, some nodes default and the
system begins to evolve. Define time $0$ as the time right after
the shock. Let $(\mathcal{G}_{k})_{0\leq k\leq m}$ be the filtration
for the probability space $(G_{n,m},\mathbb{P})$ which models the
arrival of new information, i.e. the revealed link at each step. Because
this implies that the revealed node will have its remaining equity
decrease by one, $(\mathcal{G}_{k})_{0\leq k\leq m}$ also models
the default contagion at the same time. Note in the following the
network with the set of revealed links evolves in the space $G_{n,m}$
as the result of the contagion process.
\end{rem}

\subsection{Initial condition}

From the initial equity levels $(e_{0}^{v})_{v\in[n]}$ we can determine
the initial default set $\mathcal{D}_{0}=\{v:\ e_{0}^{v}=0\}$. Let
the set of hidden out links from the nodes in $\mathcal{D}_{0}$ be
$Q_{0}=\{(i,j)\in\mathcal{E}_{n}:i\in\mathcal{D}_{0}\}$. All the
hidden links in $Q_{0}$ are assigned a clock which rings after a
random time following an independent exponential distribution with
mean $1$, i.e. $\exp(1)$. Let the sigma-algebra representing the
information available initially be $\mathcal{G}_{0}=\sigma\{(e_{0}^{v})_{v\in[n]}\}$.
Let $c_{k}^{v}$ be the sum of initial equity and accumulative number
of interventions on node $v$ and $l_{k}^{v}$ be the number of revealed
in links of node $v$ at step $k$, so $c_{0}^{v}=e_{0}^{v}$, $l_{0}^{v}=0$. 

\subsection{Dynamics\label{subsec:Dynamics}}

Let $k$ be the $k$th event that a clock rings. If $Q_{k-1}$ is
nonempty, let $(V_{k},W_{k})$ be a pair of random variables representing
the hidden link from node $V_{k}$ to node $W_{k}$ whose clock rings
first at step $k$, which means that the node $W_{k}$ records the
loss of loan due to the default of $V_{k}$. We call that $(V_{k},W_{k})$
is \textit{revealed} and $W_{k}$ is \textit{selected}. Assume $(V_{k},W_{k})=(v,w)$.
We proceed with the following steps:
\begin{itemize}
\item Update $\mathcal{G}_{k}=\sigma\left(\mbox{\ensuremath{\mathcal{G}}}_{k-1}\cup\left\{ (v,w)\right\} \right)$. 
\item Update the number of revealed out links: $l_{k}^{w}=l_{k-1}^{w}+1$
and $l_{k}^{\eta}=l_{k-1}^{\eta}$ for $\eta\neq w$.
\item Determine the intervention $\mu_{k}\in\{0,1\}$ $\mathcal{G}_{k}$
measurable at step $k$ for the selected node $w$.
\item Update $c_{k}^{w}=c_{k-1}^{w}+\mu_{k}^{w}$, otherwise $c_{k}^{\eta}=c_{k-1}^{\eta}$
for $\eta\neq w$. 
\item Update the default set. Note that $c_{k}^{\eta}\leq l_{k}^{\eta}$
indicates that the node $\eta$ has defaulted by step $k$. If $c_{k}^{w}\leq l_{k}^{w}$
and $w\notin\mathcal{D}_{k-1}$, then $\mathcal{D}_{k}=\mathcal{D}_{k-1}\cup\{w\}$
and $Q_{k}=Q_{k-1}\backslash\{(v,w)\}\cup\{(w,\eta)\in\mathcal{E}_{n}\}$
and every newly added hidden link in $\{(w,\eta)\in\mathcal{E}_{n}\}$
is assigned a clock with law $\exp(1)$, independent of everything
else. If $c_{k}^{w}>l_{k}^{w}$, $\mathcal{D}_{k}=\mathcal{D}_{k-1}$
and $Q_{k}=Q_{k-1}\backslash\{(v,w)\}$.
\end{itemize}
If $Q_{k}$ is empty, the process ends and let the process end time
be $T_{n}=k$, otherwise repeat the process. Define $D_{n}$ as the
number of defaulted nodes by the process end time $T_{n}$.
\begin{lem}
\label{lem:SelectedNodeLaw}(Adapted from Lemma 3.2 in \citet{Amini2015})
For $0\leq k\leq T_{n}-1$, the selected node $W_{k+1}$ which is
at the end of the revealed link $(V_{k+1},W_{k+1})$ from the set
$Q_{k}$ has the probability conditional on the sigma-algebra $\mathcal{G}_{k}$
that
\begin{equation}
\mathbb{P}(W_{k+1}=w\mid\mathcal{G}_{k})=\frac{d^{-}(w)-l_{k}^{w}}{m-k}\text{ for }w\in[n].\label{eq:SelectedNodeLaw}
\end{equation}
\end{lem}
\begin{proof}
Because $0\leq k\leq T_{n}-1$, $Q_{k}\neq\emptyset$ by definition.
At step $k$ because the clocks on the hidden links follow independently
and identically distributed exponential distribution which has the
memoryless property, a link is chosen uniformly among all the hidden
links and revealed. By \remref{ConfigModel}, the conditional law
of the identity of the selected node is given by the uniform matching
to the available in half links when the network is constructed sequentially.
Note in (\ref{eq:SelectedNodeLaw}) the numerator $d^{-}(w)-l_{k}^{w}$
is the total number of available in half links of node $w$. Because
at every step an in half link is connected and there are $k$ steps,
so the denominator $m-k$ is the total number of available in half
links at step $k$. In sum, a node $w$ is selected at step $k$ with
a probability proportional to the number of its available in half
links (hidden in links).
\end{proof}
Define $(c_{k}^{v},l_{k}^{v},v\in[n])$ as the state of the system
at step $k$. Given $(c_{k}^{v},l_{k}^{v},v\in[n])$, the selected
node $W_{k+1}$ and the intervention $\mu_{k+1}$ at step $k+1$,
the state of the system at $k+1$ is determined. By \lemref{SelectedNodeLaw},
the law of $W_{k+1}$ depends on $(l_{k}^{v})_{v\in[n]}$. Moreover,
the intervention $\mu_{k}$ is adapted to $\mathcal{G}_{k}$ which
can be generated by the history of the process $(c_{k}^{v},l_{k}^{v},v\in[n])_{0\leq k\leq m}$.
The objective function we introduce later is expressed in terms of
current state variable, so the system is Markovian in the state variable.
\begin{rem}
From the description of the dynamics, we obtain a continuous time
model because of the time span between a node defaulting and its creditors
recording the loss of the loan. If we only look at the event every
time a link is revealed (corresponding to a clock ringing), we obtain
the embedded discrete time Markov chain. The state of the discrete
time process at step $k$ corresponds to the state of the continuous
time process after the $k$th clock rings. Note that although the
regulator can intervene at any time, it suffices to intervene only
at the event of a link being revealed. As a result the state of the
system in continuous time does not change between the events. Since
the objective function depends on the state of the system by the end
of the contagion process, not on time, it suffices to work with the
discrete time Markov chain.
\end{rem}
Let $IT_{k}$ be the accumulative number of interventions by step
$k$ and $D_{k}=|\mathcal{D}_{k}|$ be the number of defaults at step
$k$. Particularly define $IT_{n}\coloneqq IT_{T_{n}}$ and $D_{n}\coloneqq D_{T_{n}}$.
The regulator aims to minimize the number of defaulted nodes by $T_{n}$
with the minimum amount of interventions, so we define the objective
function as a linear combination of the (scaled) number of interventions
and defaults by the end of the process $T_{n}$ as
\begin{eqnarray}
J_{n} & = & \mathbb{E}(K\frac{IT_{n}}{n}+\frac{D_{n}}{n}\mid\mathcal{G}_{0}),
\end{eqnarray}
where $K>0$ is the relative ``cost'' of an intervention. Further
by the definition of $c_{T_{n}}^{v}$ and noting that a node defaults
at last if $c_{T_{n}}^{v}\leq l_{T_{n}}^{v}$, i.e. the number of
lost loans exceeds the total of the initial equity level and the number
of interventions received by $T_{n}$, we can express $IT_{n}$ and
$D_{n}$ as
\begin{eqnarray}
IT_{n} & = & \sum_{v\in[n]}(c_{T_{n}}^{v}-e_{0}^{v}),\nonumber \\
D_{n} & = & \sum_{v\in[n]}\mathbbm1_{(c_{T_{n}}^{v}\leq l_{T_{n}}^{v})}.\nonumber \\
\label{eq:IT_D_1stTime}
\end{eqnarray}
Now we define the stochastic optimal control problem as
\begin{equation}
\tag{SCP}\min_{\mu\in\mathbb{U}}J_{n},\label{eq:Obj_SCP}
\end{equation}
where $\mu=(\mu_{k})_{1\leq k\leq m}$, $\mu_{k}\in\{0,1\}$ and $\mathbb{U}$
contains all $(\mathcal{G}_{k})_{0\leq k\leq m}$ adapted process
$\mu$. 

\part{\label{part:The-Asymptotic-Control}The Asymptotic Control Problem}

\section{Assumptions and definitions}

By the dynamics of the model we may only intervene on the node selected
at each step. Moreover, a bank cannot become liquid again once it
has defaulted, thus we cannot save defaulted banks. This assumption
is reasonable in the setting of default contagion in a stressed network.
Nor do we intervene on invulnerable nodes, because they never default
but intervening on them will only prevent us from saving the banks
that are very close to default if the interventions are costly. 

To begin our discussion about the default contagion process with interventions,
we will show first that even if the regulator is able to intervene
on multiple nodes and apply more than one unit of credit every time,
it will not be better. 
\begin{prop}
\label{prop:OptimalControlCandidates}For the stochastic control problem
(\ref{eq:Obj_SCP}), we only consider intervening on a node that,
when selected, has only one unit of equity remaining.
\end{prop}
\begin{proof}
We give a proof in words similar to the proof of proposition 3.4 in
\citet{Amini2015} for a different objective function of optimizing
the value of the financial system at the end of the process under
some budget constraint. We observe that the objective function $J_{n}$
depends on the set of defaulted nodes only through its cardinality.
Any node will affect the states of other nodes only after it defaults
because the set of unrevealed out links of the defaulted nodes determining
the contagion process grows only after a node defaults. And it is
possible for a default to occur only when a node has one unit of equity
(distance to default equal to one) at the time of being selected.
Before that time, the equity only decreases by one every time it is
selected. Moreover, there is always a chance to intervene on a node
before it defaults. However, if we intervene on a node that is not
selected at the current step or has more than one units of remaining
equity when selected, it is possible that the node may not be selected
in the following steps before the process ends in which case we implemented
redundant interventions without reducing the number of defaults.

Then we provide a mathematical proof. Let $w_{k}$ be the node selected
at step $k$ and $w\coloneqq(w_{k})_{k\in[1,m]}$ be a realization
of the sequence of selected nodes throughout the process. Consider
a control sequence $\mu\coloneqq(\mu_{k})_{k\in[1,m]}$ that for some
$v\in[n]$ and some $k_{0}\in[1,m]$, $\mu_{k_{0}}^{v}\geq1$ when
$v\neq w_{k_{0}}$, or $v=w_{k_{0}}$ but $c_{k_{0}}^{v}-l_{k_{0}}^{v}\geq2$.
Recall that $c_{k}^{v}-l_{k}^{v}$ denotes the remaining equity or
``distance to default'' of node $v$ at $k$. Let $t$ be the realization
of the terminal time $T_{n}$ under the control sequence $\mu$. Given
the initial condition $(d^{-}(v),d^{+}(v),e_{0}^{v})_{v\in[n]}$,
$w$ and $\mu$, $c^{v}\coloneqq(c_{k}^{v})_{k\in[1,m]}$ and $l^{v}\coloneqq(l_{k}^{v})_{k\in[1,m]}$
for $v\in[n]$ are determined.

Construct another control sequence $\tilde{\mu}\coloneqq(\tilde{\mu}_{k})_{k\in[1,m]}$
for the same initial condition $(d^{-}(v),d^{+}(v),e_{0}^{v})_{v\in[n]}$,
which satisfies that 
\begin{enumerate}
\item $\tilde{\mu}_{k}^{\eta}=\mu_{k}^{\eta}$ for $\eta\neq v$ and $k\in[1,m]$. 
\item Let $\tilde{c}^{v}\coloneqq\left(\tilde{c}_{k}^{v}\right)_{k\in[1,m]}$
correspond to $\tilde{\mu}$ and $w$ for node $v$, then 
\begin{equation}
\tilde{\mu}_{k+1}^{v}=\begin{cases}
1 & \text{if }\tilde{c}_{k}^{v}-l_{k}^{v}=1,w_{k+1}=v\text{ and }\tilde{c}_{k}<c_{t},\;\forall k=0,\dots t-1,\\
0 & \text{otherwise}
\end{cases}\label{eq:utilde}
\end{equation}
\end{enumerate}
In other words, $\mu$ and $\tilde{\mu}$ are the same except that
interventions are not applied to node $v$ until $v$ has the distance
to default of one when selected. By (\ref{eq:utilde}), $\tilde{c}_{t}^{v}\leq c_{t}^{v}$.
Let $D_{t}$ and $\tilde{D}_{t}$ be the number of final defaulted
nodes by $t$ under $\mu$ and $\tilde{\mu}$, respectively, then
the following are the possible cases: 
\begin{enumerate}
\item $c_{k}^{v}-l_{k}^{v}\geq1$ $\forall k\in[1,t]$ and $c_{t}^{v}-l_{t}^{v}>1$,
then $v$ is liquid under both policies at $t$, thus $D_{t}=\tilde{D}_{t}$,
but $c_{t}>\tilde{c}_{t}$. 
\item $c_{k}^{v}-l_{k}^{v}\geq1$ $\forall k\in[1,t]$ and $c_{t}^{v}-l_{t}^{v}=1$,
then $v$ is liquid under both policies at $t$, so $D_{t}=\tilde{D}_{t}$
and $c_{t}=\tilde{c}_{t}$. 
\item $c_{t}^{v}-l_{t}^{v}\leq0$, then $v$ defaults under both policies,
so $D_{t}=\tilde{D}_{t}$ and $c_{t}=\tilde{c}_{t}$. 
\end{enumerate}
For every case we have 
\begin{eqnarray}
K\frac{1}{n}\sum_{w\in[n]}(\tilde{c}_{t}^{w}-e_{0}^{w})+\frac{\tilde{D}}{n} & \leq & K\frac{1}{n}\sum_{w\in[n]}(c_{t}^{w}-e_{0}^{w})+\frac{D}{n}\label{eq:ObjCompare}
\end{eqnarray}
with strict inequality for some cases. Note that $\mu$ and $\tilde{\mu}$
do not change the probability of $w$ and $w$ is arbitrary, so (\ref{eq:ObjCompare})
holds in expectation, i.e. 
\[
\tilde{J}_{n}<J_{n}.
\]
Thus $\mu$ cannot be an optimal control sequence.
\end{proof}
We see \propref{OptimalControlCandidates} implies that it is never
optimal to intervene on a node if it is not selected or has more than
one unit of equity remaining when selected. Let $(i,j,c,l)$ be the
state of a node, meaning it has the in and out degree $(i,j)$, sum
of the initial equity and the number of interventions $c$ and $l$
revealed in links. Note that by definition $l\leq i$. We characterize
nodes with states because nodes with the same state have the same
probability of being selected at each step and are statistically the
same in influencing other nodes. Note in particular:
\begin{enumerate}
\item $c=0$ denotes that the node has defaulted initially.
\item $c-l$ denotes the remaining equity or ``distance to default'',
i.e. the number of times of being selected before a node defaults.
Thus $c\leq l$ means the node has defaulted.
\item Because $l\leq i$ by definition, $i<c$ implies that a node invulnerable,
i.e. even all loans lent out to the counterparties are written down
from the balance sheet, the node still has positive remaining equity.
On the contrary, $0<c\leq i$ denotes the node has the possibility
to default, i.e. vulnerable.
\item In the beginning of the contagion process, all nodes are in states
of the form $(i,j,c,0)$.
\end{enumerate}
Then we define the state of the system at each step. Note that the
number of nodes that have defaulted initially $(c=0)$ or invulnerable
$(i<c)$ in the beginning will not change throughout the process,
so we only need to keep track of the nodes that are initially vulnerable
$(0<c\leq i)$ and currently liquid. Further note that the possible
states throughout the process for nodes that are vulnerable in the
beginning and liquid at a later step are
\[
\Gamma\coloneqq\{(i,j,c,l):0\leq i,0\leq j,0\leq l<c\leq i\text{ or }c=i+1,l=i\}.
\]

Note particularly the state $(i,j,i+1,i)$ is the result that a node
in state $(i,j,i,i-1)$ is selected and receives one intervention
and thus becomes invulnerable.
\begin{defn}
\label{def:S}(State variable $S$) Let $S_{k}^{i,j,c,l}$ denote
the number of nodes that are vulnerable initially and are in state
$(i,j,c,l)$ at step $k$, for $k=0,\ldots,m$ and $S_{k}\coloneqq(S_{k}^{i,j,c,l})_{(i,j,c,l)\in\Gamma}$
be the state of the system. Note in the following we may use $\alpha$
to represent $(i,j,c,l)\in\Gamma$ and write $S_{k}^{\alpha}$ instead
of $S_{k}^{i,j,c,l}$ to simplify the notation.
\end{defn}
Recall $m=m(n)$ is the number of the total in (or out) degree of
the network, which is also the maximum steps of the process.

Then we define the empirical probability of in, out degrees and initial
equity levels.
\begin{defn}
(empirical probability) Define the empirical probability of the triplet
(in degree, out degree, initial equity level) as 
\[
P_{n}(i,j,c)=\frac{1}{n}\left|\{v\in[n]\mid d^{-}(v)=i,d^{+}(v)=j,e_{0}^{v}=c\}\right|.
\]
 
\end{defn}
Note that $\sum_{c\geq0}P_{n}(i,j,c)=\frac{1}{n}\left|\{v\in[n]\mid d^{-}(v)=i,d^{+}(v)=j\}\right|$
represents the empirical probability of the in and out degree pair
$(i,j)$.

Previously we use $W_{k}$ to denote the selected node at step $k$.
Now with a little abuse of notation, let $W_{k}$ denote the state
of the selected node at step $k$, $k=1,\ldots,m$, so $W_{k}\in\Gamma^{+}\coloneqq\{(i,j,c,l):0\leq i,0\leq j,0\leq c,0\leq l\leq i\}$.
We consider a Markovian control policy $G_{n}=(g_{1}^{(n)}(S_{0},W_{1}),\ldots,g_{m}^{(n)}(S_{m-1},W_{m}))$
where $g_{k+1}^{(n)}:\mathbb{N}_{0}^{|\Gamma|}\times\Gamma^{+}\rightarrow\{0,1\}$
specifies the intervention at step $k+1$ on the selected node which
has state $W_{k+1}$ given the state $S_{k}$, where $\mathbb{N}_{0}\coloneqq\{0,1,2,\ldots\}$,
the set of nonnegative integer numbers.

Letting $P_{n}=(P_{n}(i,j,c))_{0\leq c\leq i}$, we rewrite the terms
$J_{n}=J_{G_{n}}(P_{n})$, $IT_{n}=IT_{T_{n}}=IT_{n}(G_{n},P_{n})$
and $D_{n}=D_{T_{n}}=D_{n}(G_{n},P_{n})$ in (\ref{eq:Obj_SCP}) based
on $G_{n}$ and $P_{n}$ as
\begin{eqnarray}
IT_{n}(G_{n},P_{n}) & = & \sum_{k=1}^{T_{n}}g_{k}^{(n)}(S_{k-1},W_{k})\nonumber \\
D_{n}(G_{n},P_{n}) & = & n\sum_{i,j}P_{n}(i,j,0)+n\sum_{i,j,1\leq c\leq i}P_{n}(i,j,c)-\sum_{(i,j,c,l)\in\Gamma}S_{T_{n}}^{i,j,c,l}\nonumber \\
 & = & n\sum_{i,j,0\leq c\leq i}P_{n}(i,j,c)-\sum_{(i,j,c,l)\in\Gamma}S_{T_{n}}^{i,j,c,l}.\nonumber \\
\label{eq:D_n}
\end{eqnarray}
Note that the first equality for $D_{n}(G_{n},P_{n})$ holds because
the nodes that default at the end of the process consist of two parts:
the nodes that have defaulted initially $n\sum_{i,j}P_{n}(i,j,0)$
and those nodes that are vulnerable initially and default during the
process $n\sum_{i,j,1\leq c\leq i}P_{n}(i,j,c)-\sum_{(i,j,c,l)\in\Gamma}S_{T_{n}}^{i,j,c,l}$. 
\begin{assumption}
\label{assu:Regularity} Consider a sequence $([n],\mathcal{E}_{n})$
of random networks, indexed by the size of the network $n$. For each
$n\in\mathbb{N},$$(d^{-}(v))_{v\in[n]},(d^{+}(v))_{v\in[n]}$ are
sequences of nonnegative integers with $\sum_{v\in[n]}d^{-}(v)=\sum_{v\in[n]}d^{+}(v)$
and such that for some probability distribution $p$ on $\mathbb{N}_{0}^{3}$
independent of $n$ with $\lambda:=\sum_{i,j,c}ip(i,j,c)=\sum_{i,j,c}jp(i,j,c)<\infty$,
the following holds
\end{assumption}
\begin{enumerate}
\item \label{enu:assumption2}$P_{n}(i,j,c)\rightarrow p(i,j,c)\ \forall\ i,j,c\geq0$
as $n\rightarrow\infty$.
\item $\sum_{v\in[n]}[(d^{-}(v))^{2}+(d^{+}(v))^{2}]=O(n)$.
\end{enumerate}
Note that the second assumption implies by uniform integrability that
$\frac{m(n)}{n}\rightarrow\lambda$ as $n\rightarrow\infty$ and recall
that $m(n)\coloneqq\sum_{v\in[n]}d^{-}(v)=\sum_{v\in[n]}d^{+}(v)$.
Since $k\leq m(n)$, for large $n$ it holds that $\frac{k}{n}\leq\frac{m(n)}{n}\leq\lambda+1$.
Assumption \ref{assu:Regularity} essentially implies the network
is sparse which is justified in many empirical study literature on
the structure of financial networks. For example, \citet{Furfine1999,Furfine2003}
and \citet{Bech2010} explore the federal funds market and find that
the network is sparse and exhibits the small-world phenomenon. 
\begin{rem}
Define $p\coloneqq(p(i,j,c))_{i,j,0\leq c\leq i}$. We need to stress
that the vector $p$ only includes $p(i,j,c)$ in the range $0\leq c\leq i$
because $c>i$ implies that the nodes are invulnerable in the beginning
and their total number will not change throughout the contagion process.
Nor do we intervene on them.
\end{rem}
Next we present our assumptions on the control functions $g_{k}^{(n)}$.
\begin{assumption}
\label{assu:G_n}Let $G_{n}=(g_{1}^{(n)},\ldots,g_{m}^{(n)})$ be
the a control policy (a sequence of control functions) for the contagion
process on a network of size $n$ where $n$ is large enough such
that $\frac{m(n)}{n}\leq\lambda+1$. Assume that
\[
g_{k+1}^{(n)}(s,w)=\begin{cases}
u^{i,j,c,c-1}(\frac{k}{n}) & \text{if }w=(i,j,c,c-1)\in\Phi\\
0 & \text{otherwise},
\end{cases}
\]
for $0\leq k\leq m-1$, where $\Phi\coloneqq\{(i,j,c,c-1):\ 0\leq i,0\leq j,1\leq c\leq i\}$.
Note that $\Phi$ includes possible states indicating the distance
to default equal to one and $\Phi\subset\Gamma$. Further $g_{k+1}^{(n)}(s,w)=0$
for $w\notin\Phi$ follows from \propref{OptimalControlCandidates}.
$u^{i,j,c,c-1}:[0,\lambda+1]\rightarrow\{0,1\}$ is a piecewise constant
function on $[0,\lambda+1]$, i.e. there is a partition of the interval
into a finite set of intervals such that $u^{i,j,c,c-1}$ is constant
$0$ or $1$ on each interval. Note in the following we may use $\beta$
to represent $(i,j,c,c-1)\in\Phi$ and write $u^{\beta}(\tau)$ instead
of $u^{i,j,c,c-1}(\tau)$ to simplify the notation. We may use $u_{\tau}^{\beta}$
in stead of $u^{\beta}(\tau)$ if necessary. Let $u=(u^{\beta})_{\beta\in\Phi}$
and $\Pi$ contain all piecewise constant vector function $u$ on
$[0,\lambda+1]$.
\end{assumption}
\begin{rem}
By this assumption the function $u$ is independent of the state but
only a function of the time. This implies that the control function
$g_{k+1}^{(n)}(s,w)$ depends on the scaled time $\frac{k}{n}$ and
the state of the currently selected node $w$ but not on the state
$s$. When the size of the network $n$ goes to infinity, the function
$u$ specifies the interventions for the ``asymptotic'' contagion
process. Later in \propref{Sconv} we can see that it is reasonable
to consider such function $u$ because given a function $u$, we can
predict the value of a deterministic process at any time to which
the scaled stochastic contagion process converge in probability. Moreover,
this type of control policies is the one that can be solved in the
optimal control problem (\ref{eq:NonAOCP}) we will introduce later.
\end{rem}
In summary, \assuref{Regularity} assumes the convergence of the empirical
probabilities of the in and out degrees and the initial equity. On
the other hand, \assuref{G_n} indicates that the control functions
depend on the scaled time and the state of the currently selected
node. These two assumptions allow us to define the following asymptotic
control problem by ensuring that the limits in the objective function
are well defined. 
\begin{defn}
Define the asymptotic control problem given $p=(p(i,j,c))_{i,j,0\leq c\leq i}$
as
\begin{align}
\tag{{ACP}} & \min_{u\in\Pi}\lim_{n\rightarrow\infty}J_{G_{n}}(P_{n})\label{eq:ACP}\\
= & \min_{u\in\Pi}\lim_{n\rightarrow\infty}K\mathbb{E}\frac{IT_{n}(G_{n},P_{n})}{n}+\mathbb{E}\frac{D_{n}(G_{n},P_{n})}{n}.\nonumber 
\end{align}
\end{defn}
In the following we will show the limits in (\ref{eq:ACP}) are well
defined by Wormald's theorem \citep{Wormald1999} for a sequence of
networks with $P_{n}$ and $G_{n}$ satisfying \assuref{Regularity}
and \assuref{G_n}, respectively.

\section{Dynamics of the default contagion process with interventions}

Recall that $IT_{k}$ is the accumulative number of interventions
up to step $k$, so 
\begin{align}
IT_{0} & =0\nonumber \\
IT_{k} & =\sum_{\ell=1}^{k}g_{\ell}(S_{\ell-1},W_{\ell})\nonumber \\
 & =\sum_{\ell=1}^{k}\sum_{\beta\in\Phi}\mathbbm1_{(W_{\ell}=\beta)}u^{\beta}(\frac{\ell}{n}).\nonumber \\
\end{align}
We shall see that $(S_{k},IT_{k})_{k=0,\ldots,m}$ is a controlled
Markov chain given a control policy $G_{n}$. In \figref{Def_tau_hat}
we illustrate for the same $(i,j)$ pair the states we consider as
well as their transition relations. 

\begin{figure}[H]
\centering{}\includegraphics[scale=0.6]{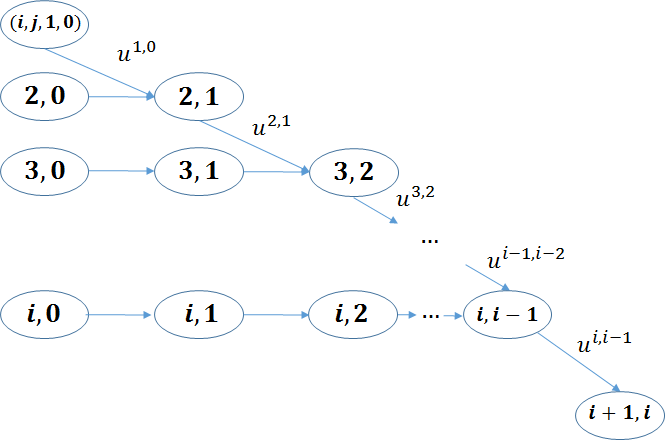}\caption{\label{fig:Def_tau_hat}The states considered for the same $(i,j)$
pair, $0\leq i$,$0\leq j$ and their transition relations.}
\end{figure}

To describe the transition probabilities, assume the state of the
selected node at step $k+1$ is $W_{k+1}=(i,j,c,l)$, for $k=0,\ldots,m-1$,
there are three possibilities:
\begin{enumerate}
\item The selected node has defaulted, i.e $c\leq l$ or the node is invulnerable,
i.e. $c>i$, then $S_{k+1}=S_{k},IT_{k+1}=IT_{k}$. 
\item The selected node is vulnerable but has the ``distance to default''
more than one, i.e $c-l\geq2$, then the node $w$ is selected with
probability $\frac{(i-l)S_{k}^{i,j,c,l}}{m-k}$ and
\begin{eqnarray}
S_{k+1}^{i,j,c,l} & = & S_{k}^{i,j,c,l}-1,\nonumber \\
S_{k+1}^{i,j,c,l+1} & = & S_{k}^{i,j,c,l+1}+1,\nonumber \\
IT_{k+1} & = & IT_{k},
\end{eqnarray}
while other entries of $S_{k+1}$ are the same as $S_{k}$.
\item The selected node has the ``distance to default'' of one, i.e $c-l=1$,
then the node is selected with probability $\frac{(i-c+1)S_{k}^{i,j,c,c-1}}{m-k}$
and by \assuref{G_n}, 
\begin{eqnarray}
S_{k+1}^{i,j,c,c-1} & = & S_{k}^{i,j,c,c-1}-1,\nonumber \\
S_{k+1}^{i,j,c+1,c} & = & S_{k}^{i,j,c+1,c}+g_{k+1}^{(n)}(S_{k},(i,j,c,c-1))\nonumber \\
 & = & S_{k}^{i,j,c+1,c}+u^{i,j,c,c-1}(\frac{k}{n}),\nonumber \\
IT_{k+1} & = & IT_{k}+u^{i,j,c,c-1}(\frac{k}{n}),
\end{eqnarray}
while other entries of $S_{k+1}$ are the same as $S_{k}$.
\end{enumerate}
Let $(\mathcal{F}_{k})_{k=0,\ldots,m}$ be the natural filtration
of $S_{k}$, $\Delta S_{k}^{\alpha}=S_{k+1}^{\alpha}-S_{k}^{\alpha}$,
$\alpha\in\Gamma$ and $\Delta IT_{k}=IT_{k+1}-IT_{k}$, it follows
that 
\begin{eqnarray}
E\left[\Delta S{}_{k}^{i,j,c,0}|\mathcal{F}_{k}\right] & = & -\frac{iS_{k}^{i,j,c,0}}{m-k}\qquad\text{for \ensuremath{1\leq c\leq i}},\nonumber \\
E\left[\Delta S_{k}^{i,j,c,l}|\mathcal{F}_{k}\right] & = & \frac{(i-l+1)S_{k}^{i,j,c,l-1}}{m-k}-\frac{(i-l)S_{k}^{i,j,c,l}}{m-k}\nonumber \\
 &  & \text{ for }3\leq c\leq i,1\leq l\leq c-2,\nonumber \\
E\left[\Delta S_{k}^{i,j,c,c-1}|\mathcal{F}_{k}\right] & = & \frac{(i-c+2)S_{k}^{i,j,c-1,c-2}}{m-k}u^{i,j,c-1,c-2}(\frac{k}{n})\nonumber \\
 &  & +\frac{(i-c+2)S_{k}^{i,j,c,c-2}}{m-k}-\frac{(i-c+1)S_{k}^{i,j,c,c-1}}{m-k}\nonumber \\
 &  & \text{ for }2\leq c\leq i,\nonumber \\
E\left[\Delta S_{k}^{i,j,i+1,i}|\mathcal{F}_{k}\right] & = & \frac{S_{k}^{i,j,i,i-1}}{m-k}u^{i,j,i,i-1}(\frac{k}{n}),\nonumber \\
E\left[\Delta IT_{k}|\mathcal{F}_{k}\right] & = & \sum_{(i,j,c,c-1)\in\Phi}\frac{(i-c+1)S_{k}^{i,j,c,c-1}}{m-k}u^{i,j,c,c-1}(\frac{k}{n}).
\end{eqnarray}

\section{\label{subsec:Convergence_Inter}Convergence of the default contagion
process with interventions}
\begin{defn}
\label{def:s}(ODEs of $s_{\tau}$) Given a set of piecewise constant
function $u=(u^{\beta})_{\beta\in\Phi}$ on $[0,\lambda]$, i.e. $u\in\Pi$,
define the system of ordinary differential equations (ODEs) of $s_{\tau}=(s_{\tau}^{\alpha})_{\alpha\in\Gamma}$
as
\begin{eqnarray}
\frac{ds_{\tau}^{i,j,c,0}}{d\tau} & = & -\frac{is_{\tau}^{i,j,c,0}}{\lambda-\tau}\qquad\text{for \ensuremath{1\leq c\leq i}},\nonumber \\
\frac{ds_{\tau}^{i,j,c,l}}{d\tau} & = & \frac{(i-l+1)s_{\tau}^{i,j,c,l-1}}{\lambda-\tau}-\frac{(i-l)s_{\tau}^{i,j,c,l}}{\lambda-\tau}\nonumber \\
 &  & \text{ for }3\leq c\leq i,1\leq l\leq c-2,\nonumber \\
\frac{ds_{\tau}^{i,j,c,c-1}}{d\tau} & = & \frac{(i-c+2)s_{\tau}^{i,j,c-1,c-2}}{\lambda-\tau}u^{i,j,c-1,c-2}(\tau)+\frac{(i-c+2)s_{\tau}^{i,j,c,c-2}}{\lambda-\tau}-\frac{(i-c+1)s_{\tau}^{i,j,c,c-1}}{\lambda-\tau}\nonumber \\
 &  & \text{ for }2\leq c\leq i,\nonumber \\
\frac{ds_{\tau}^{i,j,i+1,i}}{d\tau} & = & \frac{s_{\tau}^{i,j,i,i-1}}{\lambda-\tau}u^{i,j,i,i-1}(\tau).\nonumber \\
\label{eq:ODE_4}
\end{eqnarray}
The ODEs can be expressed in the form $\frac{ds_{\tau}}{d\tau}=h(\tau,s_{\tau};u_{\tau})$
where $h=(h^{\alpha})_{\alpha\in\Gamma}$.
\end{defn}
For what is needed below we analyze the solutions of the ODEs in \defref{s}
for a subinterval of $[0,\lambda]$ where $u(\tau)$ is a constant
vector function.
\begin{prop}
\label{prop:ODE_Sol_stau}Let $s_{\tau}=(s_{\tau}^{\alpha})_{\alpha\in\Gamma}$
satisfy the system of ordinary differential equations in \defref{s}
with the initial conditions $s_{\tau_{1}}=s_{1}\coloneqq(s_{1}^{\alpha})_{\alpha\in\Gamma}$
in the interval on $[\tau_{1},\tau_{2})\subseteq[0,\lambda)$ and
assume $u(\tau)$ is a constant vector function $u(\tau)=b\coloneqq(b^{\beta})_{\beta\in\Phi}$
where $b^{\beta}\in\{0,1\}$ on $[\tau_{1},\tau_{2})$, then the solution
$s_{\tau}$ on $[\tau_{1},\tau_{2})$ is
\begin{align}
s_{\tau}^{i,j,c,l} & =(\frac{\lambda-\tau}{\lambda-\tau_{1}})^{i-l}\sum_{r=0}^{l}s_{1}^{i,j,c,r}\binom{i-r}{l-r}(1-\frac{\lambda-\tau}{\lambda-\tau_{1}})^{l-r}\nonumber \\
 & \text{ for }2\leq c\leq i,\;0\leq l\leq c-2,\label{eq:Sol_stau_cl}\\
s_{\tau}^{i,j,c,c-1} & =(\frac{\lambda-\tau}{\lambda-\tau_{1}})^{i-c+1}\sum_{r=0}^{c-1}\sum_{q=r+1}^{c}\prod_{k=q}^{c-1}b^{i,j,k,k-1}s_{1}^{i,j,q,r}\binom{i-r}{c-1-r}(1-\frac{\lambda-\tau}{\lambda-\tau_{1}})^{c-1-r}\nonumber \\
 & \text{ for }1\leq c\leq i,\label{eq:Sol_stau_cc-1}\\
s_{\tau}^{i,j,i+1,i} & =s_{1}^{i,j,i+1,i}+\sum_{r=0}^{i-1}\sum_{q=r+1}^{i}\prod_{k=q}^{i}b^{i,j,k,k-1}s_{1}^{i,j,q,r}(1-\frac{\lambda-\tau}{\lambda-\tau_{1}})^{i-r}\nonumber \\
\label{eq:Sol_stau_iplus1i}
\end{align}
where $\prod_{k=c}^{c-1}b^{i,j,k,k-1}\coloneqq1$. As a direct result,
if we take the initial condition $s_{1}^{i,j,c,l}=p(i,j,c)\mathbbm1_{(l=0)}$
for $(i,j,c,l)\in\Gamma$ at $\tau_{1}=0$ , it follows that
\begin{align}
s_{\tau}^{i,j,c,l} & =p(i,j,c)\binom{i}{l}(1-\frac{\tau}{\lambda})^{i-l}(\frac{\tau}{\lambda})^{l}\nonumber \\
 & \text{ for }2\leq c\leq i,\;1\leq l\leq c-2.\label{eq:ODE_stau_cl_from0}
\end{align}
\end{prop}
We delegate the proof in \secref{Proofs}. 

In the following part our goal is to approximate $\frac{IT_{k}}{n}$
and $\frac{D_{k}}{n}$ as $n\rightarrow\infty$ given a function $u$.
However, the number of variables depends on $n$, so we need to bound
the terms associated with large in or out degree values. Fix $\epsilon>0$
and by \assuref{Regularity} we have that
\[
\lambda=\sum_{i,j,c}ip(i,j,c)=\sum_{i,j,c}jp(i,j,c)<\infty,
\]
then there exists an integer $M^{\epsilon}$ such that
\[
\sum_{i\geq M^{\epsilon}}\sum_{j,c}ip(i,j,c)+\sum_{j\geq M^{\epsilon}}\sum_{i,c}jp(i,j,c)<\epsilon,
\]
so 
\begin{align}
 & \sum_{i\lor j\geq M^{\epsilon},c}jp(i,j,c)\nonumber \\
= & \sum_{i\geq M^{\epsilon}}\sum_{j<M^{\epsilon}}\sum_{c}jp(i,j,c)+\sum_{i\geq M^{\epsilon}}\sum_{j\geq M^{\epsilon}}\sum_{c}jp(i,j,c)+\sum_{i<M^{\epsilon}}\sum_{j\geq M^{\epsilon}}\sum_{c}jp(i,j,c)\nonumber \\
\leq & \sum_{i\geq M^{\epsilon}}\sum_{j<M^{\epsilon}}\sum_{c}ip(i,j,c)+\sum_{i\geq M^{\epsilon}}\sum_{j\geq M^{\epsilon}}\sum_{c}jp(i,j,c)+\sum_{i<M^{\epsilon}}\sum_{j\geq M^{\epsilon}}\sum_{c}jp(i,j,c)\nonumber \\
< & \epsilon.\nonumber \\
\label{eq:HighOrderTerms_p}
\end{align}

We can prove similarly that there exists an integer $L^{\epsilon}$
such that $\sum_{i\lor j\geq L^{\epsilon},c}ip(i,j,c)<\epsilon$,
but without loss of generality we write $M^{\epsilon}$ instead of
$L^{\epsilon}$ in what follows. Moreover, by \assuref{Regularity},
as $n\rightarrow\infty$, 
\[
\sum_{i,j,c}iP_{n}(i,j,c)=\sum_{i,j,c}jP_{n}(i,j,c)\rightarrow\lambda<\infty,
\]
so for $n$ large enough, we can show that
\begin{align}
\sum_{i\lor j\geq M^{\epsilon},c}jP_{n}(i,j,c) & <\epsilon,\nonumber \\
\sum_{i\lor j\geq M^{\epsilon},c}iP_{n}(i,j,c) & <\epsilon.\label{eq:HighOrderTerms_Pn}
\end{align}

So we define the integer $M^{\epsilon}$ formally.
\begin{defn}
\label{def:Mepsilon}For any $\epsilon>0$, we define $M^{\epsilon}$
as the integer such that 
\begin{align*}
\sum_{i\lor j\geq M^{\epsilon},c}ip(i,j,c) & <\epsilon,\\
\sum_{i\lor j\geq M^{\epsilon},c}jp(i,j,c) & <\epsilon.
\end{align*}
\end{defn}
Accordingly, define
\begin{align*}
\Gamma^{\epsilon} & \coloneqq\{(i,j,c,l):i\lor j<M^{\epsilon},0\leq l<c\leq i\text{ or }c=i+1,l=i\},\\
\Phi^{\epsilon} & \coloneqq\{(i,j,c,c-1):i\lor j<M^{\epsilon},0\leq c\leq i\},\\
\hat{\lambda} & \coloneqq\lambda-\epsilon,
\end{align*}
where $a\lor b=\max\{a,b\}$.

Next we show that the scaled state variable $S_{k}$ and $IT_{k}$
converges in probability to the solution of the ODEs in \defref{s}
given the function $u$.
\begin{prop}
\label{prop:Sconv}Consider a sequence of networks with initial conditions
$(P_{n})_{n\geq1}$ satisfying \assuref{Regularity} and let $(G_{n})_{n\geq1}$
be the sequence of control policies for the contagion process on the
sequence of networks and $(G_{n})_{n\geq1}$ satisfy \assuref{G_n}
with the function $u=(u^{\beta})_{\beta\in\Phi^{\epsilon}}$, then
\begin{eqnarray}
\sup_{0\leq k\leq n\hat{\lambda}}\frac{S_{k}^{\alpha}}{n}-s_{\frac{k}{n}}^{\alpha} & = & O(n^{-\frac{1}{4}}),\nonumber \\
\sup_{0\leq k\leq n\hat{\lambda}}\frac{\tilde{IT}_{k}}{n}-it_{\frac{k}{n}}^{\epsilon} & = & O(n^{-\frac{1}{4}}),
\end{eqnarray}
with probability $1-O(n^{\frac{1}{4}}\exp(-n^{\frac{1}{4}}))$ for
$\alpha\in\Gamma^{\epsilon}$, where $s_{\tau}=(s_{\tau}^{\alpha})_{\alpha\in\Gamma^{\epsilon}}$
is the solution for the ODEs in \defref{s} with the initial conditions
$s_{0}^{i,j,c,l}=p(i,j,c)\mathbbm1_{(l=0)}$ and 
\begin{align*}
\tilde{IT_{0}} & =0,\\
\tilde{IT}_{k} & =\sum_{l=1}^{k}\sum_{\beta\in\Phi^{\epsilon}}\mathbbm1_{(W_{\ell}=\beta)}u^{\beta}(\frac{\ell}{n}),
\end{align*}
and
\begin{align}
\tilde{it}_{\tau} & =\int_{0}^{\tau}\sum_{(i,j,c,c-1)\in\Phi^{\epsilon}}\frac{(i-c+1)s_{t}^{i,j,c,c-1}}{\lambda-t}u^{i,j,c,c-1}(t)dt.\label{eq:it}
\end{align}
\end{prop}
From \propref{Sconv} we see that given $(P_{n})_{n\geq1}$ and $(G_{n})_{n\geq1}$
satisfying \assuref{Regularity} and \assuref{G_n}, respectively,
the scaled stochastic process $\frac{S_{k}}{n}$ converges to the
deterministic process $s_{\frac{k}{n}}$ for any $k$ in $[0,n\hat{\lambda}]$.
This justifies the control policy we consider in \assuref{G_n} because
given a control policy $G_{n}$ depending on the function $u$, we
can predict with high probability the scaled stochastic contagion
process at any time $k$. The proof of \propref{Sconv} is delegated
to \secref{Proofs}.

Next we discuss the convergence of the scaled number of defaults and
the process end time.
\begin{defn}
\label{def:Def_D} Define $D_{k}^{-}$ as the number of unrevealed
out links from the default set at step $k$.

Recall that $D_{k}$ denotes the number of defaulted nodes at step
$k$ which consist of two parts: the nodes that have defaulted initially
$n\sum_{i,j}P_{n}(i,j,0)$ and those nodes that are vulnerable initially
and default by step $k$, i.e. $n\sum_{i,j,1\leq c\leq i}P_{n}(i,j,c)-\sum_{(i,j,c,l)\in\Gamma}S_{k}^{i,j,c,l}$,
thus
\begin{eqnarray}
D_{k} & = & n\sum_{i,j}P_{n}(i,j,0)+n\sum_{i,j,1\leq c\leq i}P_{n}(i,j,c)-\sum_{(i,j,c,l)\in\Gamma}S_{k}^{i,j,c,l}\nonumber \\
 & = & n\sum_{i,j,0\leq c\leq i}P_{n}(i,j,c)-\sum_{(i,j,c,l)\in\Gamma}S_{k}^{i,j,c,l}.
\end{eqnarray}
\end{defn}
Similarly, among all defaulted nodes at step $k$ the nodes with out
degree $j$ consist of two parts: the nodes that have defaulted initially
$n\sum_{i}P_{n}(i,j,0)$ and those nodes that are vulnerable initially
and default by step $k$, $n\sum_{i,1\leq c\leq i}P_{n}(i,j,c)-\sum_{i,0\leq l<c\leq i}S_{k}^{i,j,c,l}$,
thus

\begin{eqnarray}
D_{k}^{-} & = & \sum_{j}j(n\sum_{i}P_{n}(i,j,0)+n\sum_{i,1\leq c\leq i}P_{n}(i,j,c)-\sum_{i,0\leq l<c\leq i}S_{k}^{i,j,c,l})-k\nonumber \\
 & = & n\sum_{i,j,0\leq c\leq i}jP_{n}(i,j,c)-\sum_{(i,j,c,l)\in\Gamma}jS_{k}^{i,j,c,l}-k.
\end{eqnarray}

Correspondingly we make the following definitions to approximate $\frac{D_{k}}{n}$
and $\frac{D_{k}^{-}}{n}$ as $n\rightarrow\infty$.
\begin{defn}
\label{def:Def_d}Define 
\begin{eqnarray}
d_{\tau} & = & \sum_{i,j,0\leq c\leq i}p(i,j,c)-\sum_{(i,j,c,l)\in\Gamma}s_{\tau}^{i,j,c,l}\nonumber \\
d_{\tau}^{-} & = & \sum_{i,j,0\leq c\leq i}jp(i,j,c)-\sum_{(i,j,c,l)\in\Gamma}js_{\tau}^{i,j,c,l}-\tau.\nonumber \\
\label{eq:EqOfdWithInt}
\end{eqnarray}
\end{defn}
\begin{prop}
\label{prop:DD_conv}Based on \defref{Def_D} and \defref{Def_d},
it follows that
\begin{align}
\sup_{0\leq k\leq n\hat{\lambda}}\left|\frac{D_{k}^{-}}{n}-d_{\frac{k}{n}}^{-}\right| & \leq o_{p}(1)+2\epsilon,\nonumber \\
\sup_{0\leq k\leq n\hat{\lambda}}\left|\frac{D_{k}}{n}-d_{\frac{k}{n}}\right| & \leq o_{p}(1)+2\epsilon.
\end{align}
\end{prop}
\begin{proof}
For some $[\tau_{1},\tau_{2})\subseteq[0,\lambda)$ on which $u(\tau)$
is a constant vector function, from \propref{ODE_Sol_stau} we can
show by summing all $s_{\tau}^{i,j,c,l}$, $(i,j,c,l)\in\Gamma$ with
the same $(i,j)$ that
\[
0\leq\sum_{c,l}s_{\tau}^{i,j,c,l}\leq\sum_{c,l}s_{\tau_{1}}^{i,j,c,l},
\]
so by induction we have
\[
\sum_{c,l}s_{\tau}^{i,j,c,l}\leq\sum_{1\leq c\leq i}p(i,j,c)
\]
and thus it follows from (\ref{eq:HighOrderTerms_p}) that
\begin{align}
0\leq & \sum_{i\lor j\geq M^{\epsilon}}\sum_{0\leq c\leq i}jp(i,j,c)-\sum_{(i,j,c,l)\in\Gamma\backslash\Gamma^{\epsilon}}js_{\tau}^{i,j,c,l}\nonumber \\
\leq & \sum_{i\lor j\geq M^{\epsilon}}\sum_{0\leq c\leq i}jp(i,j,c)<\epsilon.
\end{align}
Similarly because by the definition of $S_{k}^{i,j,c,l}$ for $1\leq k\leq m$,
for fixed $(i,j)$ pair, $(i,j,c,l)\in\Gamma$,
\[
0\leq\sum_{c,l}\frac{S_{k}^{i,j,c,l}}{n}\leq\sum_{1\leq c\leq i}P_{n}(i,j,c),
\]
thus it follows from \eqref{HighOrderTerms_Pn} that
\begin{align}
0\leq & \sum_{i\lor j\geq M^{\epsilon}}\sum_{0\leq c\leq i}jP_{n}(i,j,c)-\sum_{(i,j,c,l)\in\Gamma\backslash\Gamma^{\epsilon}}j\frac{S_{k}^{i,j,c,l}}{n}\nonumber \\
\leq & \sum_{i\lor j\geq M^{\epsilon}}\sum_{0\leq c\leq i}jP_{n}(i,j,c)<\epsilon.
\end{align}

For any $k$ where $0\leq k\leq\hat{\lambda}$, by \propref{Sconv}
it follows that
\begin{eqnarray}
\left|\frac{D_{k}^{-}}{n}-d_{\frac{k}{n}}^{-}\right| & = & |\sum_{i\lor j<M^{\epsilon},0\leq c\leq i}jP_{n}(i,j,c)-\sum_{(i,j,c,l)\in\Gamma^{\epsilon}}j\frac{S_{k}^{i,j,c,l}}{n}-\nonumber \\
 &  & -(\sum_{i\lor j<M^{\epsilon},0\leq c\leq i}jp(i,j,c)-\sum_{(i,j,c,l)\in\Gamma^{\epsilon}}js_{\frac{k}{n}}^{i,j,c,l})|+2\epsilon\nonumber \\
 & = & \left|\sum_{i\lor j<M^{\epsilon},0\leq c\leq i}j\left(P_{n}(i,j,c)-p(i,j,c)\right)-\sum_{(i,j,c,l)\in\Gamma^{\epsilon}}j\left(\frac{S_{k}^{i,j,c,l}}{n}-s_{\frac{k}{n}}^{i,j,c,l}\right)\right|+2\epsilon\nonumber \\
 & \leq & \sum_{i\lor j<M^{\epsilon},0\leq c\leq i}j\left|P_{n}(i,j,c)-p(i,j,c)\right|+\sum_{(i,j,c,l)\in\Gamma^{\epsilon}}j\left|\frac{S_{k}^{i,j,c,l}}{n}-s_{\frac{k}{n}}^{i,j,c,l}\right|+2\epsilon\nonumber \\
 & \leq & M^{\epsilon}|\Gamma^{\epsilon}|\left(o(1)+o_{p}(1)\right)+2\epsilon=o_{p}(1)+2\epsilon,\nonumber \\
\label{eq:D-Tod-_WithInt}
\end{eqnarray}
and similarly, 
\begin{eqnarray}
\left|\frac{D_{k}}{n}-d_{\frac{k}{n}}\right| & \leq & \left|\sum_{i\lor j<M^{\epsilon},0\leq c\leq i}P_{n}(i,j,c)-\sum_{(i,j,c,l)\in\Gamma^{\epsilon}}\frac{S_{k}^{i,j,c,l}}{n}-(\sum_{i\lor j<M^{\epsilon},0\leq c\leq i}p(i,j,c)-\sum_{(i,j,c,l)\in\Gamma^{\epsilon}}s_{\frac{k}{n}}^{i,j,c,l})\right|\nonumber \\
 &  & +2\epsilon\nonumber \\
 & = & \left|\sum_{i\lor j<M^{\epsilon},0\leq c\leq i}\left(P_{n}(i,j,c)-p(i,j,c)\right)-\sum_{(i,j,c,l)\in\Gamma^{\epsilon}}\left(\frac{S_{k}^{i,j,c,l}}{n}-s_{\frac{k}{n}}^{i,j,c,l}\right)\right|+2\epsilon\nonumber \\
 & \leq & \sum_{i\lor j<M^{\epsilon},0\leq c\leq i}\left|P_{n}(i,j,c)-p(i,j,c)\right|+\sum_{(i,j,c,l)\in\Gamma^{\epsilon}}\left|\frac{S_{k}^{i,j,c,l}}{n}-s_{\frac{k}{n}}^{i,j,c,l}\right|+2\epsilon\nonumber \\
 & \leq & |\Gamma^{\epsilon}|\left(o(1)+o_{p}(1)\right)+2\epsilon=o_{p}(1)+2\epsilon.\nonumber \\
\label{eq:DTodWithInt}
\end{eqnarray}
\end{proof}
To summarize the results we have so far, we have shown in \ref{prop:Sconv}
and (\ref{prop:DD_conv}) that the state variable $S_{k}$, the accumulative
interventions $IT_{k}$, the number of defaults $D_{k}$ and the number
of unrevealed out links from the default set $D_{k}^{-}$ after being
scaled by $n$ all converge to a deterministic process which depends
on the solution of the system of ODEs in \defref{s}. This convergence
applies to any $k$ before $n\hat{\lambda}$. Recall that $IT_{n}\coloneqq IT_{T_{n}}$
and $D_{n}\coloneqq D_{T_{n}}$. By \defref{Def_D}, $T_{n}=\min\{0\leq k\leq m:D_{k}^{-}=0\}$.
Additionally define $\tau_{f}=\inf\{0\leq\tau\leq\lambda:d_{\tau}^{-}=0\}$.
Next we show that when $\frac{T_{n}}{n}$ converges to $\tau_{f}$,
then $\frac{IT_{n}}{n}$ and $\frac{D_{n}}{n}$ also converge in probability
to the corresponding deterministic variables, $it_{\tau_{f}}$ and
$d_{\tau_{f}}$, which in light of the boundedness of $\frac{IT_{n}}{n}$
and $\frac{D_{n}}{n}$ further implies convergence in expectations.
\begin{prop}
\label{prop:Conv_obj}Consider a sequence of networks with initial
conditions $(P_{n})_{n\geq1}$ satisfying \assuref{Regularity} and
let $(G_{n})_{n\geq1}$ be the sequence of control policies for the
contagion processes on the sequence of networks and $(G_{n})_{n\geq1}$
satisfy \assuref{G_n} with the function $u$. If $\tau_{f}=\lambda$,
or $\tau_{f}<\lambda$ and $\frac{d}{d\tau}d_{\tau_{f}}^{-}<0$, it
follows that as $n\rightarrow\infty$,
\begin{eqnarray}
\frac{IT_{n}(G_{n},P_{n})}{n} & \overset{p}{\rightarrow} & it_{\tau_{f}}(u,p),\nonumber \\
\frac{D_{n}(G_{n},P_{n})}{n} & \overset{p}{\rightarrow} & d_{\tau_{f}}(u,p).\label{eq:ConvInProb}
\end{eqnarray}
where
\begin{align}
it_{\tau_{f}} & =\int_{0}^{\tau_{f}}\sum_{(i,j,c,c-1)\in\Phi}\frac{(i-c+1)s_{t}^{i,j,c,c-1}}{\lambda-t}u^{i,j,c,c-1}(t)dt.\label{eq:it_fullrange}
\end{align}
Further it follows that as $n\rightarrow\infty$,
\begin{eqnarray}
\mathbb{E}\frac{IT_{n}(G_{n},P_{n})}{n} & \rightarrow & it_{\tau_{f}}(u,p),\nonumber \\
\mathbb{E}\frac{D_{n}(G_{n},P_{n})}{n} & \rightarrow & d_{\tau_{f}}(u,p).\label{eq:ConvInExpectation}
\end{eqnarray}
\end{prop}
\begin{proof}
By (\ref{eq:HighOrderTerms_Pn}) for $n$ large enough and $1\leq k\leq m$,
we have
\begin{align*}
 & \frac{1}{n}\sum_{\ell=1}^{k}\sum_{i\lor j\geq M^{\epsilon}}\sum_{1\leq c\leq i}\mathbbm1_{(W_{\ell}=(i,j,c,c-1))}u^{i,j,c,c-1}(\frac{\ell}{n})\\
\leq & \frac{1}{n}\sum_{i\lor j\geq M^{\epsilon}}\sum_{1\leq c\leq i}iS_{k}^{i,j,c,c-1}\\
\leq & \sum_{i\lor j\geq M^{\epsilon},c}iP_{n}(i,j,c)<\epsilon.
\end{align*}
Similarly by (\ref{eq:HighOrderTerms_p}), for $\tau<\hat{\lambda}$,
\begin{align*}
 & \int_{0}^{\tau}\sum_{i\lor j\geq M^{\epsilon},1\leq c\leq i}\frac{(i-c+1)s_{t}^{i,j,c,c-1}}{\lambda-t}u^{i,j,c,c-1}(t)dt\\
\leq & \int_{0}^{\tau}\sum_{i\lor j\geq M^{\epsilon},1\leq c\leq i}\frac{ip(i,j,c)}{\lambda-t}dt\\
\leq & \epsilon\int_{0}^{\tau}\frac{1}{\lambda-t}dt\\
= & \epsilon\ln\frac{\lambda}{\lambda-\tau}\\
\leq & \epsilon\ln\frac{\lambda}{\epsilon}=O(\epsilon).
\end{align*}

For any $k$ where $0\leq k\leq n\hat{\lambda}$, by \propref{Sconv}
it follows that
\begin{align*}
\left|\frac{IT_{k}}{n}-it_{\frac{k}{n}}\right| & \leq|\frac{\tilde{IT}_{k}}{n}+\frac{1}{n}\sum_{\ell=1}^{k}\sum_{i\lor j\geq M^{\epsilon}}\sum_{1\leq c\leq i}\mathbbm1_{(W_{\ell}=(i,j,c,c-1))}u^{i,j,c,c-1}(\frac{\ell}{n})\\
 & -(\tilde{it}_{\frac{k}{n}}+\int_{0}^{\tau}\sum_{i\lor j\geq M^{\epsilon},1\leq c\leq i}\frac{(i-c+1)s_{t}^{i,j,c,c-1}}{\lambda-t}u^{i,j,c,c-1}(t)dt)|\\
 & \leq\left|\frac{\tilde{IT}_{k}}{n}-\tilde{it}_{\frac{k}{n}}\right|+O(\epsilon)\\
 & \leq o_{p}(1)+O(\epsilon),
\end{align*}
thus we have that

\begin{align}
\sup_{0\leq k\leq n\hat{\lambda}}\left|\frac{IT_{k}}{n}-it_{\frac{k}{n}}\right| & =o_{p}(1)+O(\epsilon).\label{eq:ITitConv}
\end{align}

If $\tau_{f}=\lambda$, it implies that $d_{\tau}^{-}>0$ for $\tau\in(0,\hat{\lambda})$,
then it follows from \propref{DD_conv} that $\frac{T_{n}}{n}=\hat{\lambda}+O(\epsilon)+o_{p}(1)$.
Then by the bounded change (\ref{eq:BoundedChange}), $\frac{D_{T_{n}}}{n}=\frac{D_{\left\lfloor n\hat{\lambda}\right\rfloor }}{n}+O(\epsilon)+o_{p}(1)$
and from \propref{DD_conv} again, $\frac{D_{\left\lfloor n\hat{\lambda}\right\rfloor }}{n}=d_{\hat{\lambda}}+O(\epsilon)+o_{p}(1)$.
$\left\lfloor \cdot\right\rfloor $ denotes the floor function. Further
by the continuity of $d_{\tau}$ on $[0,\lambda]$, $\frac{D_{T_{n}}}{n}=d_{\lambda}+O(\epsilon)+o_{p}(1)$.
Similarly, by (\ref{eq:ITitConv}) and the continuity of $it_{\tau}$
on $[0,\lambda]$, we have that $\frac{IT_{T_{n}}}{n}=it_{\lambda}+O(\epsilon)+o_{p}(1)$.

If $\tau_{f}<\lambda$ and $\frac{d}{d\tau}d_{\tau_{f}}^{-}<0$, by
\defref{s}, $s_{\tau}$ is continuous and thus by (\ref{eq:EqOfdWithInt})
$d_{\tau}^{-}$ is also continuous. So there exists some $\tau'>0$
such that $d_{\tau}^{-}<0$ for $\tau\in(\tau_{f},\tau_{f}+\tau')$
by the continuity of $d_{\tau}^{-}$. Since $\epsilon$ is arbitrary,
let $\epsilon$ be small enough such that $\inf_{\tau\in(\tau_{f},\tau_{f}+\tau')}d_{\tau}^{-}<-2\epsilon$
and $\hat{\tau}$ be the first time $d_{\tau}^{-}$ reaches the minimum.
Because $d_{\hat{\tau}}^{-}<-2\epsilon$, then by \propref{DD_conv}
$\frac{D_{\left\lfloor n\hat{\tau}\right\rfloor }^{-}}{n}<0$ with
high probability, so it holds that $\frac{T_{n}}{n}=\tau_{f}+O(\epsilon)+o_{p}(1)$.
Again by the continuity of $d_{\tau}$ and $it_{\tau}$ on $[0,\lambda]$,
\propref{DD_conv} and (\ref{eq:ITitConv}), $\frac{D_{T_{n}}}{n}=d_{\tau_{f}}+O(\epsilon)+o_{p}(1)$
and $\frac{IT_{T_{n}}}{n}=it_{\tau_{f}}+O(\epsilon)+o_{p}(1)$.

Recall that $IT_{n}\coloneqq IT_{T_{n}}$ and $D_{n}\coloneqq D_{T_{n}}$.
In both cases we conclude that (\ref{eq:ConvInProb}) holds by tending
$\epsilon\rightarrow0$.

To prove (\ref{eq:ConvInExpectation}), since $IT_{n}\leq m\leq(\lambda+0.1)n$
for large $n$ and $D_{n}\leq n$, $\frac{IT_{n}(G_{n},P_{n})}{n}$
and $\frac{D_{n}(G_{n},P_{n})}{n}$ are bounded and thus uniformly
integrable. For a sequence of uniformly integrable random variables,
convergence in probability implies convergence in expectation. Therefore
(\ref{eq:ConvInExpectation}) holds.
\end{proof}
Under the conditions in \propref{Conv_obj}, the asymptotic control
problem (\ref{eq:ACP}) becomes

\begin{equation}
\min_{u\in\Pi}\:K\cdot it_{\tau_{f}}(u,p)+d_{\tau_{f}}(u,p).\label{eq:OCP_Prior}
\end{equation}

In the following we write $u^{\beta}(\tau)=u_{\tau}^{\beta}$ and
let $u_{\tau}=(u_{\tau}^{\beta})_{\beta\in\Phi}$ and $u=(u_{\tau})_{\tau\in[0,\lambda]}$.

Substituting the expressions of $it_{\tau_{f}}(u,p)$ and $d_{\tau_{f}}(u,p)$
in (\ref{eq:it_fullrange}) and (\ref{eq:EqOfdWithInt}) respectively
into (\ref{eq:OCP_Prior}) and putting together the system of ordinary
differential equations of $s_{\tau}=(s_{\tau}^{\alpha})_{\alpha\in\Gamma}$,
$\frac{d}{d\tau}s_{\tau}=h(\tau,s_{\tau};u_{\tau})$ in \defref{s}
as well as the condition that determines $\tau_{f}$, $d_{\tau_{f}}^{-}=0$,
we attain the following deterministic optimal control problem. 
\begin{flushleft}
\begin{alignat}{1}
\tag{OCP}\min_{u,\tau_{f}} & \;K\cdot it_{\tau_{f}}(u,p)+d_{\tau_{f}}(u,p)\label{eq:NonAOCP}\\
\text{st}\nonumber \\
 & \frac{d}{d\tau}s_{\tau}=h(\tau,s_{\tau};u_{\tau})\nonumber \\
 & s_{0}^{i,j,c,l}=p(i,j,c)\mathbbm1_{(l=0)}\nonumber \\
 & d_{\tau_{f}}^{-}=0\nonumber \\
 & u_{\tau}^{\beta}\in\{0,1\}\;\forall\beta\in\Phi\nonumber \\
 & \tau_{f}\in[0,\lambda),\nonumber 
\end{alignat}
where $\frac{d}{d\tau}s_{\tau}=h(\tau,s_{\tau};u_{\tau})$ is defined
as in \defref{s} with $u^{\beta}(\tau)=u_{\tau}^{\beta}$ and
\begin{align}
it_{\tau_{f}}(u,p) & =\int_{0}^{\tau_{f}}\sum_{i,j,1\leq c\leq i}\frac{(i-c+1)s_{\tau}^{i,j,c,c-1}}{\lambda-\tau}u_{\tau}^{i,j,c,c-1}d\tau,\nonumber \\
d_{\tau_{f}}(u,p) & =\sum_{i,j,0\leq c\leq i}p(i,j,c)-\sum_{(i,j,c,l)\in\Gamma}s_{\tau_{f}}^{i,j,c,l},\nonumber \\
d_{\tau_{f}}^{-} & =\sum_{i,j,0\leq c\leq i}jp(i,j,c)-\sum_{(i,j,c,l)\in\Gamma}js_{\tau_{f}}^{i,j,c,l}-\tau_{f}.
\end{align}
\par\end{flushleft}

Some difficulties arise because (\ref{eq:NonAOCP}) is a infinite
dimensional optimal control problem. We provide a two step algorithm
to show that in light of \assuref{Regularity}, it suffices to solve
a finite dimensional problem to approximate the objective function
of the infinite dimensional problem. First we define the finite dimensional
optimal control problem.
\begin{defn}
\label{def:FOCP}For $\epsilon>0$, recall $M^{\epsilon}$ as in \defref{Mepsilon}.
Define the finite dimensional optimal control problem (FOCP) as (\ref{eq:NonAOCP})
with the indexes $(i,j)$ restricted to $i\lor j<M^{\epsilon}$. 
\end{defn}
\begin{rem}
\label{rem:BoundedDegree}The restriction of $(i,j)$ to $i\lor j<M^{\epsilon}$
indicates that we use only $p(i,j,c)$, $i\lor j<M^{\epsilon},$ $0\leq c\leq i$
in the calculation. It is equivalent to setting $p(i,j,c)=0$, $i\lor j\geq M^{\epsilon}$,
$0\leq c\leq i$ while keeping $p(i,j,c)$, $i\lor j<M^{\epsilon}$,
$0\leq c\leq i$ unchanged, which implies asymptotically nodes with
$i\lor j\geq M^{\epsilon}$ are all invulnerable. By the solution
of the ODE in \propref{ODE_Sol_stau}, it implies that $s_{\tau}^{\beta}=0$,
for $\beta\in\Gamma\backslash\Gamma^{\epsilon}.$ Note we use tilde
sign with the variables to indicate the indexes $(i,j)$ are in the
range $i\lor j<M^{\epsilon}$, for example, $\tilde{it}_{\tau_{f}}$,
$\tilde{d}_{\tau_{f}}$ and $\tilde{d}_{\tau_{f}}^{-}$. 

We have the following lemma regarding the objective functions of the
infinite and finite dimensional optimal control problems.
\end{rem}
\begin{lem}
\label{lem:TwoOCP}Let $j(u,\tau_{f},p)\coloneqq K\,it_{\tau_{f}}(u,p)+d_{\tau_{f}}(u,p)$
be the objective function for the unrestricted (\ref{eq:NonAOCP})
and $\tilde{j}(u,\tau_{f},p)\coloneqq K\,\tilde{it}{}_{\tau_{f}}(u,p)+\tilde{d}_{\tau_{f}}(u,p)$
for FOCP. Let $(u^{*},\tau_{f}^{*})$ be the optimal solution for
the unrestricted (\ref{eq:NonAOCP}) and $(\tilde{u},\tilde{\tau}_{f})$
be the optimal solution for FOCP, then for the same $p$ we have that
\[
\left|\tilde{j}(\tilde{u},\tilde{\tau}_{f},p)-j(u^{*},\tau_{f}^{*},p)\right|<(K+1)\epsilon.
\]
\end{lem}
\begin{proof}
First we give an algorithm to construct a policy $u$ based on $\tilde{u}$.
Let $(i,j,c,l)$ be the state of a node, so $i\lor j\geq M^{\epsilon}$
indicates that the node has either in or out degree greater than $M^{\epsilon}$
and $i\lor j<M^{\epsilon}$ indicates that the in and out degree is
bounded by $M^{\epsilon}$. 
\begin{lyxalgorithm}
\label{alg:TwoStep}A two step algorithm.
\end{lyxalgorithm}
\begin{enumerate}
\item Assume all nodes with $i\lor j\geq M^{\epsilon}$ are invulnerable
for now. This is because as stated in \remref{ConfigModel}, when
constructing the configuration model sequentially, we have the freedom
to reveal the out links only from the defaulted nodes with bounded
degrees $i\lor j<M^{\epsilon}$ until all the out links from defaulted
nodes with $i\lor j<M^{\epsilon}$ have been revealed. Let $\tilde{T}_{n}$
be the stopping time with $\tilde{D}_{\tilde{T}_{n}}$and $\tilde{IT}_{\tilde{T}_{n}}$
being the number of defaults and interventions by $\tilde{T}_{n}$,
respectively. Accordingly we can solve FOCP for the optimal $(\tilde{u},\tilde{\tau}_{f})$
with the corresponding $\tilde{d}_{\tilde{\tau}_{f}}^{-}$. As in
the proof of \propref{Conv_obj}, assume that $\tilde{\tau}_{f}=\lambda$,
or $\tilde{\tau}_{f}<\lambda$ and $\frac{d}{d\tau}\tilde{d}_{\tilde{\tau}_{f}}^{-}<0$,
then $\frac{\tilde{T}_{n}}{n}\overset{p}{\rightarrow}\tilde{\tau}_{f}$.
\item But the nodes with $i\lor j\geq M^{\epsilon}$ may have defaulted
by $\tilde{T}_{n}$. Next we look at these nodes, which consist of
initially defaulted nodes with $i\lor j\geq M^{\epsilon}$ and those
with $i\lor j\geq M^{\epsilon}$ that default because of the defaults
of nodes with $i\lor j<M^{\epsilon}$ in the first step. Now we reveal
their out links and we adopt a control policy to intervention at every
step until the end of the contagion process $T_{n}$. Under this policy,
there will be no additional defaults between $\tilde{T}_{n}$ and
$T_{n}$. Recall $D_{T_{n}}$ and $IT_{T_{n}}$ are the number of
defaults and interventions by $T_{n}$. 
\end{enumerate}
Through the above two steps, we construct a function $u$ as $u_{\tau}=\tilde{u}_{\tau}$,
$\tau\leq\tilde{\tau}_{f}$ and $u_{\tau}^{\beta}=1$, $\tilde{\tau}_{f}<\tau$,
$\beta\in\Phi$. Thus we have that
\begin{align*}
\frac{D_{T_{n}}}{n} & =\frac{\tilde{D}_{\tilde{T}_{n}}}{n}+\frac{D_{T_{n}}-\tilde{D}_{\tilde{T}_{n}}}{n}\\
 & \leq\frac{\tilde{D}_{\tilde{T}_{n}}}{n}+\sum_{i\lor j\geq M^{\epsilon},0\leq c\leq i}P_{n}(i,j,c),\\
\frac{IT_{T_{n}}}{n} & =\frac{\tilde{IT}{}_{\tilde{T}_{n}}}{n}+\frac{IT_{T_{n}}-\tilde{IT}{}_{\tilde{T}_{n}}}{n}\\
 & \leq\frac{\tilde{IT}{}_{\tilde{T}_{n}}}{n}+\sum_{i\lor j\geq M^{\epsilon},0\leq c\leq i}jP_{n}(i,j,c).
\end{align*}
Note that the first inequality holds because the number of defaulted
nodes with $i\lor j\geq M^{\epsilon}$ by $T_{n}$ is bounded above
by the number of initially defaulted and initially vulnerable nodes
with $i\lor j\geq M^{\epsilon}$. The second inequality holds because
the number of interventions after $\tilde{T}_{n}$ is bounded above
by the number of all out links from the nodes with $i\lor j\geq M^{\epsilon}$.
Taking expectation on both sides and letting $n\rightarrow\infty$,
from \propref{Conv_obj} and (\ref{eq:HighOrderTerms_Pn}) it follows
that
\begin{align*}
d_{\tau_{f}} & \leq\tilde{d}_{\tilde{\tau}_{f}}+\epsilon,\\
it_{\tau_{f}} & \leq\tilde{it}_{\tilde{\tau}_{f}}+\epsilon.
\end{align*}
By the definition of $j$ and $\tilde{j}$,
\[
j(u,\tau_{f},p)\leq\tilde{j}(\tilde{u},\tilde{\tau}_{f},p)+(K+1)\epsilon.
\]
Recall $(u^{*},\tau_{f}^{*})$ is the optimal solution for the unrestricted
(\ref{eq:NonAOCP}). Moreover, in the first step of \algref{TwoStep}
we treat all nodes with $i\lor j\geq M^{\epsilon}$ as invulnerable
and $(\tilde{u},\tilde{\tau}_{f})$ is the optimal solution, so it
provides the lower bound for the optimal objective function of the
unrestricted (\ref{eq:NonAOCP}), i.e. 
\[
\tilde{j}(\tilde{u},\tilde{\tau}_{f},p)\leq j(u^{*},\tau_{f}^{*},p).
\]

For the function $u$ we introduce through the two steps, we can calculate
$\tau_{f}$ and the objective function value $j(u,\tau_{f},p)$. Then
by the optimality of $(u^{*},\tau_{f}^{*})$, we have that
\[
j(u^{*},\tau_{f}^{*},p)\leq j(u,\tau_{f},p).
\]
In sum, we have that
\[
\tilde{j}(\tilde{u},\tilde{\tau}_{f},p)\leq j(u^{*},\tau_{f}^{*},p)\leq j(u,\tau_{f},p)\leq\tilde{j}(\tilde{u},\tilde{\tau}_{f},p)+(K+1)\epsilon,
\]
Thus the conclusion follows.
\end{proof}
By \lemref{TwoOCP} we only need to solve the finite dimensional optimal
control problem (FOCP) in \ref{def:FOCP}. Because $\epsilon$ can
be arbitrarily small, we can approximate the objective function of
the infinite dimensional problem to any precision. Given $p$ for
FOCP, the Pontryagin's maximum principle provides the necessary conditions
for the optimal control $\tilde{u}$ and end time $\tilde{\tau}_{f}$.
We can obtain the optimal asymptotic number of interventions $\tilde{it}{}_{\tilde{\tau}_{f}}$
and fraction of final defaults $\tilde{d}{}_{\tilde{\tau}_{f}}$,
which lead to the main results of our work. In the next section we
focus on solving FOCP and suppress the tilde sign for the variables
for notational convenience.

\section{\label{subsec:NecessaryCondition}Necessary conditions for the optimal
control problem}

In the following we solve the finite dimensional optimal control problem
(FOCP) in \defref{FOCP}. Throughout this section we understand that
the degrees are in the bounded range $i\lor j<M^{\epsilon}$ unless
specified otherwise. We also suppress the tilde sign for notational
convenience. Let $t=t(\tau)\coloneqq-\ln(\lambda-\tau)$, $t_{0}\coloneqq t(0)=-\ln\lambda$
and $t_{f}\coloneqq-\ln(\lambda-\tau_{f})$. Note $t(\tau)$ is a
strictly increasing function of $\tau$ and so is the inverse function
$\tau=\tau(t)$. We remark that we assume in the following that $\tau<\lambda$
which implies $t_{f}<\infty$, but later we can see that the solutions
of $s_{\tau}$, $u_{\tau}$ and $w_{\tau}$ do allow $\tau=\lambda$.
Then we can reformulate the optimal control problem (\ref{eq:NonAOCP})
into an autonomous one, i.e. the differential equations of the system
dynamics do not depend on time explicitly. Let $u_{t}=(u_{t}^{\beta})_{\beta\in\Phi^{\epsilon}}$
and without any confusion $u=(u_{t})_{t\geq t_{0}}$ (previously $u=(u_{\tau})_{\tau\in[0,\lambda]}$).
\begin{alignat}{1}
\tag{AOCP}\min_{u,t_{f}} & \;K\cdot it_{t_{f}}(u,p)+d_{t_{f}}(u,p)\label{eq:AOCP}\\
\label{eq:Obj_AOCP}\\
\text{st}\nonumber \\
\text{} & \frac{d}{dt}s_{t}=h(s_{t};u_{t})\nonumber \\
 & s_{t_{0}}^{i,j,c,l}=p(i,j,c)\mathbbm1_{(l=0)}\nonumber \\
 & d_{t_{f}}^{-}=0\nonumber \\
 & u_{t}^{\beta}\in\{0,1\}\;\forall\beta\in\Phi^{\epsilon}\nonumber \\
 & t_{f}\in[0,\lambda),\nonumber 
\end{alignat}
where $\frac{d}{dt}s_{t}=h(s_{t};u_{t})$ denotes the system of differential
equations
\begin{align}
\frac{ds_{t}^{i,j,c,0}}{dt} & =-is_{t}^{i,j,c,0}\text{ for }1\leq c\leq i,\nonumber \\
\frac{ds_{t}^{i,j,c,l}}{dt} & =(i-l+1)s_{t}^{i,j,c,l-1}-(i-l)s_{t}^{i,j,c,l}\nonumber \\
 & \text{ for }3\leq c\leq i,\;1\leq l\leq c-2,\nonumber \\
\frac{ds_{t}^{i,j,c,c-1}}{dt} & =(i-c+2)s_{t}^{i,j,c-1,c-2}u_{t}^{i,j,c-1,c-2}+(i-c+2)s_{t}^{i,j,c,c-2}\nonumber \\
 & -(i-c+1)s_{t}^{i,j,c,c-1}\nonumber \\
 & \text{ for }2\leq c\leq i,\nonumber \\
\frac{ds_{t}^{i,j,i+1,i}}{dt} & =s_{t}^{i,j,i,i-1}u_{t}^{i,j,i,i-1},
\end{align}
and 
\begin{align}
it_{t_{f}}(u,p) & =\int_{t_{0}}^{t_{f}}\sum_{i,j,1\leq c\leq i}(i-c+1)s_{t}^{i,j,c,c-1}u_{t}^{i,j,c,c-1}dt,\nonumber \\
d_{t_{f}}(u,p) & =\sum_{i,j,0\leq c\leq i}p(i,j,c)-\sum_{(i,j,c,l)\in\Gamma^{\epsilon}}s_{t_{f}}^{i,j,c,l},\nonumber \\
d_{t_{f}}^{-} & =\sum_{i,j,0\leq c\leq i}jp(i,j,c)-\sum_{(i,j,c,l)\in\Gamma^{\epsilon}}js_{t_{f}}^{i,j,c,l}-\tau_{f}\nonumber \\
 & =\sum_{i,j,0\leq c\leq i}jp(i,j,c)-\sum_{(i,j,c,l)\in\Gamma^{\epsilon}}js_{t_{f}}^{i,j,c,l}-\lambda(1-e^{t_{0}-t_{f}}).\label{eq:Outlinks_AOCP}
\end{align}

Note that (\ref{eq:Outlinks_AOCP}) follows from $\frac{1}{\lambda}=e^{t_{0}}$
and thus $\tau_{f}=\lambda-e^{-t_{f}}=\lambda(1-\frac{1}{\lambda}e^{-t_{f}})=\lambda(1-e^{t_{0}-t_{f}})$. 

To determine the necessary conditions for the optimal terminal time
$t_{f}^{*}$ and optimal control $u_{t}^{*}$ in (\ref{eq:AOCP}),
we need the extended Pontryagin maximum principle \lemref{PontryaginMPM}
presented in the appendix. Then we put together the objective function
(\ref{eq:Obj_AOCP}) and the necessary conditions to construct the
optimization problem (\ref{eq:OptProgram}) we will consider later.

Applying the extended Pontryagin maximum principle to the optimal
control problem (\ref{eq:AOCP}) yields the following necessary conditions
of optimality. Note in order to simplify notations, we suppress the
apostrophe $"*"$ in the following. In other words, we use $s_{t},u_{t},w_{t},v,t_{f}$
instead of $s_{t}^{*},u_{t}^{*},w_{t}^{*},v^{*},t_{f}^{*}$ to denote
the optimal values. First we present the correspondence of the notations
in \lemref{PontryaginMPM} and our application.

\begin{table}[H]
\centering{}%
\begin{tabular}{|c|c|}
\hline 
Notation in \lemref{PontryaginMPM} & Notation in our application\tabularnewline
\hline 
\hline 
$t$ & $t$\tabularnewline
\hline 
$t_{0}$ & $t_{0}$\tabularnewline
\hline 
$t_{f}$ & $t_{f}$\tabularnewline
\hline 
$(x_{t}^{i})_{i\in\{1,\ldots,n_{x}\}}$ & $(s_{t}^{\alpha})_{\alpha\in\Gamma^{\epsilon}}$\tabularnewline
\hline 
$(u_{t}^{i})_{i\in\{1,\ldots,n_{u}\}}$ & $(u_{t}^{\beta})_{\beta\in\Phi^{\epsilon}}$\tabularnewline
\hline 
$U$ & $\{0,1\}$\tabularnewline
\hline 
$\mathring{\lambda}$ & $\mathring{w}$\tabularnewline
\hline 
$(\lambda_{t}^{i})_{i\in\{1,\ldots,n_{x}\}}$ & $(w_{t}^{\alpha})_{\alpha\in\Gamma^{\epsilon}}$\tabularnewline
\hline 
$\ell(t,x_{t},u_{t})$ & $K\sum_{i,j,1\leq c\leq i}(i-c+1)s_{t}^{i,j,c,c-1}u_{t}^{i,j,c,c-1}$\tabularnewline
\hline 
$\phi(t_{f},x_{t_{f}})$ & $d_{t_{f}}$\tabularnewline
\hline 
$\psi_{k}(t_{f},x_{t_{f}})=0,\quad k=1,\ldots,n_{\psi}$ & $d_{t_{f}}^{-}=0$\tabularnewline
\hline 
\end{tabular}\caption{Correspondence of the notations in \lemref{PontryaginMPM} and \propref{NecessaryOptimalCond}.}
\end{table}

\begin{prop}
\label{prop:NecessaryOptimalCond}(Necessary conditions of optimality)
Let $(s_{t},u_{t})_{t\in[t_{0},t_{f}]}$ be the optimal state and
control trajectory of (\ref{eq:AOCP}) where $t_{f}$ is the optimal
terminal time. Define the Hamiltonian function 
\begin{eqnarray}
\mathcal{H}(s_{t},u_{t},w_{t}) & = & \sum_{i,j,1\leq c\leq i}w_{t}^{i,j,c,0}(-is_{t}^{i,j,c,0})\nonumber \\
 &  & +\sum_{i,j,2\leq c\leq i,1\leq l\leq c-1}w_{t}^{i,j,c,l}\left[(i-l+1)s_{t}^{i,j,c,l-1}-(i-l)s_{t}^{i,j,c,l}\right]\nonumber \\
 &  & +\sum_{i,j,2\leq c\leq i+1}(K+w_{t}^{i,j,c,c-1})(i-c+2)s_{t}^{i,j,c-1,c-2}u_{t}^{i,j,c-1,c-2},\nonumber \\
\label{eq:Ham_NC}
\end{eqnarray}
then there exist a piecewise continuously differentiable vector function
$w_{t}=(w_{t}^{\alpha})_{\alpha\in\Gamma^{\epsilon}}\in\hat{C}^{1}[t_{0},\infty)^{\left|\Gamma^{\epsilon}\right|}$
and a scalar $v\in\mathbb{R}$ such that the following conditions
hold: 
\begin{enumerate}
\item \label{enu:Cond1_NC}The optimal control $u_{t}$ satisfies that $\forall t\in[t_{0},t_{f}]$,
$1\leq c\leq i$, 
\begin{equation}
\label{eq:u_NC}
\end{equation}
if $s_{t}^{i,j,c,c-1}>0$,
\begin{alignat*}{1}
u_{t}^{i,j,c,c-1} & =\begin{cases}
0 & \text{ if }w_{t}^{i,j,c+1,c}>-K\\
1 & \text{ if }w_{t}^{i,j,c+1,c}<-K\\
0\text{ or }1 & \text{ if }w_{t}^{i,j,c+1,c}=-K,
\end{cases}
\end{alignat*}
if $s_{t}^{i,j,c,c-1}=0$, 
\begin{align*}
u_{t}^{i,j,c,c-1} & =0\text{ or }1.
\end{align*}
\item \label{enu:Cond2_NC}For $2\leq c\leq i$, $0\leq l\leq c-2$, 
\begin{equation}
\frac{d}{dt}w_{t}^{i,j,c,l}=(i-l)(w_{t}^{i,j,c,l}-w_{t}^{i,j,c,l+1}),
\end{equation}
and $\text{ for }$ $1\leq c\leq i$, 
\begin{equation}
\frac{d}{dt}w_{t}^{i,j,c,c-1}=(i-c+1)(w_{t}^{i,j,c,c-1}-(K+w_{t}^{i,j,c+1,c})u_{t}^{i,j,c,c-1}),
\end{equation}
and 
\begin{equation}
\frac{d}{dt}w_{t}^{i,j,i+1,i}=0.
\end{equation}

We denote the set of ordinary differential equations for $w_{t}$
as $\frac{d}{dt}w_{t}=h'(w_{t};u_{t})$.
\item \label{enu:Cond3_NC}$\mathcal{H}(s_{t},u_{t},w_{t})$ is a constant
for $t\in[t_{0},t_{f}]$. 
\item \label{enu:Cond4_NC}Transversal conditions 
\begin{alignat}{1}
w_{t_{f}}^{i,j,c,l} & =vj-1\qquad\forall(i,j,c,l)\in\Gamma^{\epsilon},\label{eq:w_TC}\\
\mathcal{H}(s_{t_{f}},u_{t_{f}},w_{t_{f}}) & =-ve^{-t_{f}},\label{eq:H_TC}\\
d_{t_{f}}^{-} & =\sum_{i,j,0\leq c\leq i}jp(i,j,c)-\sum_{(i,j,c,l)\in\Gamma^{\epsilon}}js_{t_{f}}^{i,j,c,l}-\lambda(1-e^{t_{0}-t_{f}})\nonumber \\
 & =0.\label{eq:S_TC}
\end{alignat}
\end{enumerate}
\end{prop}
\begin{proof}
Let $(\mathring{w},w_{t})$ be the adjoint variables then $\mathring{w}=1$,
since otherwise the necessary conditions of optimality becomes independent
of the cost functional in (\ref{eq:AOCP}). The Hamiltonian function
(\ref{eq:Ham_NC}) is a direct result of (\ref{eq:Ham_MPM}). Note
that $n_{\psi}=1$ and 
\[
\psi(t,s)=\sum_{i,j,0\leq c\leq i}jp(i,j,c)-\sum_{(i,j,c,l)\in\Gamma^{\epsilon}}js^{i,j,c,l}-\lambda(1-e^{t_{0}-t}).
\]
Taking partial derivative yields $\frac{\partial}{\partial s}\psi(t_{f},s_{t_{f}})=(j,j,\ldots,j)$
which has rank $1$.

Since the Hamiltonian function is affine in the control variable $u_{t}$,
by condition (\ref{enu:Cond1_MPM}) of \lemref{PontryaginMPM}, we
attain that, for $1\leq c\leq i$,
\[
u_{t}^{i,j,c,c-1}=\begin{cases}
0 & \text{ if }(K+w_{t}^{i,j,c+1,c})s_{t}^{i,j,c,c-1}>0\\
1 & \text{ if }(K+w_{t}^{i,j,c+1,c})s_{t}^{i,j,c,c-1}<0\\
0\text{ or }1 & \text{if }(K+w_{t}^{i,j,c+1,c})s_{t}^{i,j,c,c-1}=0.
\end{cases}
\]
By distinguishing the two cases $s_{t}^{i,j,c,c-1}>0$ and $s_{t}^{i,j,c,c-1}=0$,
we have the equivalent form in (\ref{eq:u_NC}).

Taking partial derivative of $\mathcal{H}$ with regard to $s$ yields
the differential equations of $w_{t}$ in condition (\ref{enu:Cond2_NC}).
Note that $\mathcal{H}$ is autonomous, then according to condition
(\ref{enu:Cond3_MPM}) of \lemref{PontryaginMPM}, $\mathcal{H}(s_{t},u_{t},w_{t})$
is a constant for $t\in[t_{0},t_{f}]$, which is condition (\ref{enu:Cond3_NC}).

Then define 
\begin{eqnarray}
\mbox{\ensuremath{\Psi}}(t,s) & \coloneqq & \sum_{i,j,0\leq c\leq i}p(i,j,c)-\nonumber \\
 &  & \sum_{(i,j,c,l)\in\Gamma^{\epsilon}}s^{i,j,c,l}+v(\sum_{i,j,0\leq c\leq i}jp(i,j,c)-\sum_{(i,j,c,l)\in\Gamma^{\epsilon}}js_{t_{f}}^{i,j,c,l}-\lambda(1-e^{t_{0}-t_{f}}))\nonumber \\
\end{eqnarray}
 and taking partial derivatives with respect to $s$ and $t$ respectively
by condition (\ref{enu:Cond4_MPM}) of \lemref{PontryaginMPM} together
with the terminal condition (\ref{eq:Outlinks_AOCP}) leads to condition
(\ref{enu:Cond4_NC}).
\end{proof}
\begin{rem}
\label{rem:SingularControl}(Singular control policy) Observe that
if the coefficient of $u_{t}^{i,j,c,c-1}$ in the Hamiltonian function
$\mathcal{H}(s_{t},u_{t},w_{t})$ (\ref{eq:Ham_NC}), i.e. $(i-c+1)(K+w_{t}^{i,j,c+1,c})s_{t}^{i,j,c,c-1}$
vanishes, $u_{t}^{i,j,c,c-1}=0$ or $1$ both satisfy (\ref{enu:Cond1_MPM})
of \lemref{PontryaginMPM}, i.e. minimizing $\mathcal{H}(s_{t},u_{t},w_{t})$.
In other words, the Pontryagin maximum principle in \lemref{PontryaginMPM}
cannot determine the optimal control $u_{t}^{i,j,c,c-1}$ in this
case. Moreover, since $i-c+1>0$, so if $(K+w_{t}^{i,j,c+1,c})s_{t}^{i,j,c,c-1}=0$
can be sustained over some interval $(\theta_{1},\theta_{2})\subset[t_{0},t_{f}]$,
then $u_{t}^{i,j,c,c-1}$ can be $0$ or $1$ at any time on $(\theta_{1},\theta_{2})$
and switch arbitrarily often between $0$ and $1$. In the terminology
of optimal control theory, the control $u_{t}$ is ``singular''
on $(\theta_{1},\theta_{2})$ and the corresponding portion of the
state trajectory $s_{t}$ on $(\theta_{1},\theta_{2})$ is called
a singular arc. Further note that $(K+w_{t}^{i,j,c+1,c})s_{t}^{i,j,c,c-1}=0$,
$t\in(\theta_{1},\theta_{2})$ implies two cases: $s_{t}^{i,j,c,c-1}=0$
or $s_{t}^{i,j,c,c-1}>0$, $w_{t}^{i,j,c+1,c}=-K$, $t\in(\theta_{1},\theta_{2})$.
We can show that in the first case any feasible control policy $u_{t}^{i,j,c,c-1}$
will not affect other entries of $s_{t}$ and the second case only
occurs when $c=i$ and $(\theta_{1},\theta_{2})=(t_{0},t_{f})$.
\end{rem}
Now the differential equations for $s_{t}$ in (\ref{eq:AOCP}) and
$w_{t}$ in condition (\ref{enu:Cond2_NC}) of \propref{NecessaryOptimalCond}
constitute a two-point boundary value problem (TPBVP) because for
$s_{t}$ the boundary values are given at $t=t_{0}$ while for $w_{t}$
at $t=t_{f}$ as follows
\begin{align}
\tag{TPBVP}\frac{d}{dt}s_{t} & =h(s_{t};u_{t})\label{eq:TPBVP}\\
s_{t_{0}}^{i,j,c,l} & =p(i,j,c)\mathbbm1_{(l=0)}\qquad\forall(i,j,c,l)\in\Gamma^{\epsilon}\nonumber \\
\frac{d}{dt}w_{t} & =h'(w_{t};u_{t})\nonumber \\
w_{t_{f}}^{i,j,c,l} & =vj-1\qquad\forall(i,j,c,l)\in\Gamma^{\epsilon}\nonumber 
\end{align}
By solving the differential equations we are able to find the optimal
control policy $(u_{t})_{t\in[t_{0},t_{f}]}$ stated in \thmref{theorem3}
and the optimal state trajectory $(s_{t})_{t\in[t_{0},t_{f}]}$ which
leads to the conclusions of \thmref{theorem4} about the optimal asymptotic
fraction of final defaulted nodes.

\section{Solutions of the necessary conditions}

Throughout this section we understand that the degrees are in the
bounded range $i\lor j<M^{\epsilon}$ unless specified otherwise.
The main difficulty of solving the TPBVP arises from the fact that
$w_{t}$ and $s_{t}$ depend on $u_{t}$ which depends on $w_{t}$
and $s_{t}$ recursively according to (\ref{enu:Cond1_NC}) of \propref{NecessaryOptimalCond}.
To disentangle the recursive dependence, the idea is to analyze the
properties of $s_{t}$ based on which we are able to either derive
the properties of $w_{t}$ or provide explicit solutions of $w_{t}$
under different cases. By (\ref{enu:Cond1_NC}) of \propref{NecessaryOptimalCond}
we attain the optimal control policy $u_{t}$ which allows us to solve
for $s_{t}$. At this point $(s_{t},u_{t},w_{t})$ are all expressed
in terms of the variables $(v,t_{f},t_{s})$. Since the optimal $(s_{t},u_{t},w_{t})$
satisfies the two equations in the necessary conditions of \propref{NecessaryOptimalCond},
i.e. the Hamiltonian function (\ref{eq:H_TC}) at $t=t_{f}$ and the
terminal condition (\ref{eq:S_TC}), while minimizing the objective
function (\ref{eq:Obj_AOCP}), we have the following optimization
problem for $(v,t_{f},t_{s})$.
\begin{alignat}{1}
\min_{v,t_{f},t_{s}} & K\cdot it_{t_{f}}(u,p)+d_{t_{f}}(u,p)\label{eq:OP_beforeSubstitute}\\
\label{eq:Obj_forOP}\\
\text{st}\quad & \mathcal{H}(s_{t_{f}},u_{t_{f}},w_{t_{f}})=-ve^{-t_{f}}\label{eq:H_TC_forOP}\\
 & d_{t_{f}}^{-}=0\label{eq:S_TC_forOP}\\
 & t_{0}\leq t_{s}\leq t_{f}\nonumber \\
 & v\in\mathbb{R},\nonumber 
\end{alignat}
where $s_{t}$, $u_{t}$ and $w_{t}$ are functions of $(v,t_{f},t_{s})$
and
\begin{align}
it_{t_{f}}(u,p) & =\int_{t_{0}}^{t_{f}}\sum_{i,j,1\leq c\leq i}(i-c+1)s_{t}^{i,j,c,c-1}u_{t}^{i,j,c,c-1}dt,\nonumber \\
d_{t_{f}}(u,p) & =\sum_{i,j,0\leq c\leq i}p(i,j,c)-\sum_{(i,j,c,l)\in\Gamma^{\epsilon}}s_{t_{f}}^{i,j,c,l},\nonumber \\
d_{t_{f}}^{-} & =\sum_{i,j,0\leq c\leq i}jp(i,j,c)-\sum_{(i,j,c,l)\in\Gamma^{\epsilon}}js_{t_{f}}^{i,j,c,l}-\lambda(1-e^{t_{0}-t_{f}}).
\end{align}

After substituting $(s_{t},u_{t},w_{t})$ expressed in terms of $(v,t_{f},t_{s})$
into the optimization problem (\ref{eq:OP_beforeSubstitute}), we
are able to solve the optimal $(v,t_{f},t_{s})$ based on which we
can calculate the optimal $(s_{t},u_{t},w_{t})$. Further we change
the variables from $(v,t_{f},t_{s})$ to $(v,y,z)$ to further simplify,
thus we attain the optimization problem (\ref{eq:OptProgram}) later.

It turns out that we only need $w_{t_{f}}$ as well as $u_{t}$ and
$s_{t}$ in (\ref{eq:OP_beforeSubstitute}). From (\ref{eq:TPBVP})
we know that $w_{t_{f}}^{\beta}=vj-1$ for $\beta\in\Gamma^{\epsilon}$.
For $u_{t}$ and $s_{t}$, we give out their solutions in the following
directly due to the limited space of the paper.
\begin{lem}
\label{lem:OC}The optimal control policy $u_{t}$ in terms of the
variables $(v,t_{f},t_{s})$ is given as below.

For $1\leq c\leq i$ except $c=i$ when $vj-1=-K$, $\forall t\in[t_{0},t_{f}]$,
\[
u_{t}^{i,j,c,c-1}=\mathbbm1_{(t\geq t^{i,j,c})},
\]
where
\[
t^{i,j,c}=\begin{cases}
t_{f} & \text{ if }vj-1\geq-K\\
t_{f}+\ln\left(1+\frac{K+vj-1}{(i-c)K}\right) & \text{ if }vj-1<-K\text{ and }1\leq c<i+\frac{K+vj-1}{K(1-e^{t_{0}-t_{f}})}\\
t_{0} & \text{otherwise},
\end{cases}
\]
If $vj-1=-K$, $\forall t\in[t_{0},t_{f}]$, 
\[
u_{t}^{i,j,i,i-1}=\mathbbm1_{(t\geq t_{s})}\text{ for some }t_{s}\in[t_{0},t_{f}].
\]

\end{lem}
The following is the solution for $s_{t}$.
\begin{lem}
\label{lem:s_t_Sol}Letting $p(i,j,i+1)=0$, under the optimal control
policy in \lemref{OC}, we have the following solutions of $s_{t}$
for the two-point boundary value problem (TPBVP).
\end{lem}
\begin{enumerate}
\item For $2\leq c\leq i$, $0\leq l\leq c-2$ and $c=1$, $l=0$, 
\begin{equation}
s_{t}^{i,j,c,l}=p(i,j,c)\binom{i}{l}\left(e^{t_{0}-t}\right)^{i-l}(1-e^{t_{0}-t})^{l}.\label{eq:s_ijcl_Sol_General_1}
\end{equation}
\item \label{enu:s_t_AjL-K}If $vj-1<-K$, consider $t\in[t^{i,j,h},t^{i,j,h-1})$,
for some $1\leq h\leq i$ where 
\[
t^{i,j,h}=\begin{cases}
t_{f} & \text{if }h=0\\
t_{f}+\ln(1+\frac{K+vj-1}{(i-h)K}) & \text{if }1\leq h<i+\frac{K+vj-1}{K(1-e^{t_{0}-t_{f}})}\\
t_{0} & \text{otherwise}.
\end{cases}
\]
If $1\leq c<h$, 
\begin{align}
s_{t}^{i,j,c,c-1} & =p(i,j,c)\binom{i}{c-1}\left(e^{t_{0}-t}\right)^{i-c+1}(1-e^{t_{0}-t})^{c-1}.\label{eq:s_cl_t_Sol_cLh_AjL-K}
\end{align}
If $h\leq c\leq i+1$, 
\begin{align}
s_{t}^{i,j,c,c-1} & =\binom{i}{c-1}\left(e^{t_{0}-t}\right)^{i-c+1}\sum_{m=h}^{c}p(i,j,m)\nonumber \\
 & \sum_{n=0}^{m-1}\binom{c-1}{n}\left(1-e^{t_{0}-t^{i,j,m}}\right)^{n}(e^{t_{0}-t^{i,j,m}}-e^{t_{0}-t})^{c-1-n}.\nonumber \\
\label{eq:s_cc-1_t_Sol_cGh_AjL-K}
\end{align}
\item \label{enu:s_t_AjG-K}If $vj-1>-K$, for $1\leq c\leq i+1,$ $t\in[t_{0},t_{f}]$,
\begin{equation}
s_{t}^{i,j,c,c-1}=p(i,j,c)\binom{i}{c-1}\left(e^{t_{0}-t}\right)^{i-c+1}(1-e^{t_{0}-t})^{c-1}.\label{eq:s_cl_t_Sol_AjG-K}
\end{equation}
\item If $vj-1=-K$, for $1\leq c\leq i$, $t\in[t_{0},t_{f}]$, 
\[
s_{t}^{i,j,c,c-1}=p(i,j,c)\binom{i}{c-1}\left(e^{t_{0}-t}\right)^{i-c+1}(1-e^{t_{0}-t})^{c-1},
\]
 and 
\begin{equation}
s_{t}^{i,j,i+1,i}=p(i,j,i)\mathbbm1_{\{t_{s}\leq t)}[(1-e^{t_{0}-t})^{i}-(1-e^{t_{0}-t_{s}})^{i}],\label{eq:s_ip1i_t_Sol_AjE-K}
\end{equation}
where $t_{s}\in[t_{0},t_{f}]$. 
\end{enumerate}
The solutions of $w_{t}$ and $s_{t}$ can be verified by substituting
into (\ref{eq:TPBVP}). Since (\ref{eq:S_TC}) and (\ref{eq:H_TC})
require the state variable value particularly at $t=t_{f}$, we can
apply \lemref{s_t_Sol} at $t=t_{f}$ to attain $s_{t_{f}}$. Next
we substitute $s_{t}$, $u_{t}$ and $w_{t_{f}}$ into the optimization
program (\ref{eq:OP_beforeSubstitute}) leading to the following results.
\begin{prop}
\label{prop:OP_in_exponent}Based on the solutions of optimal $s_{t}$
in \lemref{s_t_Sol} (particularly $s_{t_{f}}$), $u_{t}$ in \lemref{OC}
and $w_{t_{f}}^{\beta}=vj-1$, $\forall\beta\in\Gamma^{\epsilon}$,
letting 
\[
t^{i,j,c}=\begin{cases}
t_{f} & \text{if }K+vj-1\geq0\text{ or }c=0\\
t_{f}+\ln(1+\frac{K+vj-1}{(i-c)K}) & \text{if }K+vj-1<0\text{ and }1\leq c<i+\frac{K+vj-1}{K(1-e^{t_{0}-t_{f}})},\\
t_{0} & \text{otherwise},
\end{cases}
\]
 the Hamiltonian equation (\ref{eq:H_TC_forOP}) at $t=t_{f}$ is
equivalent to 
\begin{eqnarray}
 &  & \sum_{j}\max\{-K,vj-1\}\sum_{i}i\sum_{c=1}^{i}p(i,j,c)\sum_{m=c}^{i}\binom{i-1}{m-1}\left(e^{t_{0}-t_{f}}\right)^{i-m+1}\nonumber \\
 &  & \sum_{n=0}^{c-1}\binom{m-1}{n}\left(1-e^{t_{0}-t^{i,j,c}}\right)^{n}(e^{t_{0}-t^{i,j,c}}-e^{t_{0}-t_{f}})^{m-1-n}\nonumber \\
 & = & v\lambda e^{t_{0}-t_{f}}.\nonumber \\
\label{eq:H_NCSol}
\end{eqnarray}
The terminal condition (\ref{eq:S_TC_forOP}) is equivalent to
\begin{eqnarray}
 &  & \sum_{i}\sum_{j}j\{\sum_{c=0}^{i}p(i,j,c)\sum_{n=c}^{i}\binom{i}{n}(1-e^{t_{0}-t^{i,j,c}})^{n}(e^{t_{0}-t^{i,j,c}})^{i-n}\nonumber \\
 &  & -\mathbbm1_{(vj-1=-K)}p(i,j,i)[(1-e^{t_{0}-t_{f}})^{i}-(1-e^{t_{0}-t_{s}})^{i}]\}\nonumber \\
 & = & \lambda(1-e^{t_{0}-t_{f}}).\nonumber \\
\label{eq:Outlink_NCSol}
\end{eqnarray}
And the objective function (\ref{eq:Obj_forOP}) 
\begin{eqnarray}
 &  & K\cdot it_{t_{f}}(u,p)+d_{t_{f}}(u,p)\nonumber \\
 & = & K\sum_{i}\sum_{j}\{\sum_{c=1}^{i}p(i,j,c)\nonumber \\
 &  & \sum_{m=c}^{i}\sum_{n=0}^{c-1}(m-c+1)\binom{i}{m}\binom{m}{n}(e^{t_{0}-t_{f}})^{i-m}\left(1-e^{t_{0}-t^{i,j,c}}\right)^{n}(e^{t_{0}-t^{i,j,c}}-e^{t_{0}-t_{f}})^{m-n}\nonumber \\
 &  & +\mathbbm1_{(vj-1=-K)}p(i,j,i)[(1-e^{t_{0}-t_{f}})^{i}-(1-e^{t_{0}-t_{s}})^{i}]\}\nonumber \\
 &  & +\sum_{j}\sum_{i}\{\sum_{c=0}^{i}p(i,j,c)\sum_{n=c}^{i}\binom{i}{n}(1-e^{t_{0}-t^{i,j,c}})^{n}(e^{t_{0}-t^{i,j,c}})^{i-n}\nonumber \\
 &  & -\mathbbm1_{(vj-1=-K)}p(i,j,i)[(1-e^{t_{0}-t_{f}})^{i}-(1-e^{t_{0}-t_{s}})^{i}]\}.\nonumber \\
\label{eq:Obj_NCSol}
\end{eqnarray}
\end{prop}
We further simplify the expressions in \propref{OP_in_exponent}.
Define 
\begin{align}
y & =1-e^{t_{0}-t_{f}},\nonumber \\
z & =1-e^{t_{0}-t_{s}},\nonumber \\
x^{i,j,c,c-1} & =1-e^{t_{0}-t^{i,j,c}}\nonumber \\
 & =\begin{cases}
y & \text{if }K+vj-1\geq0\text{ or }c=0\\
1-(1-y)\frac{(i-c)K}{(i-c+1)K+vj-1} & \text{if }K+vj-1<0\text{ and }1\leq c<i+\frac{K+vj-1}{Ky}\\
0 & \text{otherwise},
\end{cases}\nonumber \\
\label{eq:x_and_xijc}
\end{align}
then because $t_{0}\leq t_{s}\leq t_{f}$ and the function $1-e^{t_{0}-t}$
is increasing in $t$, $0\leq z\leq y\leq1$. As a result, we can
substitute the new variables $(y,v,z)$ into the objective function
(\ref{eq:Obj_forOP}), the Hamiltonian equation (\ref{eq:H_TC_forOP})
and the terminal condition (\ref{eq:S_TC_forOP}). Moreover, we add
the definition of $x^{i,j,c,c-1}$ and $0\leq z\leq y\leq1$ and as
a result, we obtain a new optimization problem defined as (\ref{eq:OptProgram})
in \secref{MainResults}. After solving (\ref{eq:OptProgram}) for
$(y^{*},v^{*},z^{*})$, we are able to calculate $u_{t}^{*}$ and
$s_{t}^{*}$ (or $u_{\tau}^{*}$ and $s_{\tau}^{*}$ after changing
the time index) in order to present \thmref{theorem3} and \thmref{theorem4}.

\section{\label{sec:MainResults}Main results}

\subsection{Contagion process with no interventions}

We first present the theorem when no interventions are provided in
the contagion process. For $\epsilon>0$, recall $M^{\epsilon}$ is
defined as in \defref{Mepsilon} and note that all indexes $(i,j)$
are in the range $i\lor j<M^{\epsilon}$ in what follows. 
\begin{defn}
\label{def:DefOfI(x)}($I$ function) Define a function $I:\ [0,1]\rightarrow[0,1]$
as 
\begin{equation}
I(y)\coloneqq\frac{1}{\lambda}\sum_{i\lor j<M^{\epsilon}}j\sum_{c=0}^{i}p(i,j,c)\mathbb{P}(\text{Bin}(i,y)\geq c)\label{eq:I}
\end{equation}
where $\text{Bin}(i,y)$ denotes a binomial random variable with $i$
trials and the probability of occurrence $y$, so $\mathbb{P}(\text{Bin}(i,y)\geq c)=\sum_{m=c}^{i}\binom{i}{m}y^{m}(1-y)^{i-m}$.
$I(y)$ is constructed to represent the asymptotic scaled total out
degree from the default set at scaled time $y$ and attains its form
(\ref{eq:I}) from the solution of a set of differential equations.
It can be interpreted as the expected number of defaults if an end
node of a randomly selected directed link defaults with probability
$y$. 

Since $I(0)=\frac{1}{\lambda}\sum_{i\lor j<M^{\epsilon}}jp(i,j,0)\geq0$,
and from the definition of $\lambda$, 
\[
I(1)=\frac{1}{\lambda}\sum_{i\lor j<M^{\epsilon}}j\sum_{c=0}^{i}p(i,j,c)\leq1,
\]
and $I(y)$ is continuous and increasing, it has at least one fixed
point in $[0,1]$. Further define 
\[
J(y)\coloneqq\sum_{i\lor j<M^{\epsilon}}\sum_{c=0}^{i}p(i,j,c)\mathbb{P}(\text{Bin}(i,y)\geq c).
\]
\end{defn}
\begin{thm}
\label{thm:theorem1} (Extends from theorem 3.8 of \citet{Amini2013})
Consider a sequence of networks with initial conditions $(P_{n})_{n\geq1}$
satisfying \assuref{Regularity} where $p=(p(i,j,c))_{i,j,0\leq c\leq i}$
such that $p(i,j,c)=0$ for $i\lor j\geq M^{\epsilon},0\leq c\leq i$
and no interventions are implemented, let $y^{*}\in[0,1]$ be the
smallest fixed point of $I$.
\end{thm}
\begin{enumerate}
\item If $y^{*}=1$, then asymptotically almost all nodes default during
the contagion process, i.e. 
\[
\frac{D_{n}}{n}\overset{p}{\to}1.
\]
\item If $y^{*}<1$ and it is a stable fixed point, i.e. $I'(y^{*})<1$,
then asymptotically the fraction of final defaulted nodes 
\[
\frac{D_{n}}{n}\overset{p}{\to}J(y^{*}).
\]
Particularly, if $I(0)=0$ and $I'(0)<1$, then 
\[
\frac{D_{n}}{n}\overset{p}{\to}\sum_{i\lor j<M^{\epsilon}}p(i,j,0).
\]
\end{enumerate}
\begin{rem}
Theorem \ref{thm:theorem1} states that the stopping time of the default
contagion process is fully governed by the smallest fixed point of
the function $I(y)$ and the asymptotic fraction of defaulted nodes
at the end of the process can be 1, 0 or a fractional number, representing
almost all nodes default, almost no nodes default or a partial number
of nodes default, respectively. 
\end{rem}

\subsection{Contagion process with interventions}

We present the theorems for the contagion process with interventions
as the result of solving the finite dimensional optimal control problem
(FOCP) in \defref{FOCP}. For $\epsilon>0$, recall $M^{\epsilon}$
is defined as in \defref{Mepsilon}. By \lemref{TwoOCP}, the optimal
objective function value of FOCP is within $(K+1)\epsilon$ distance
to the optimal objective function value of the unrestricted (\ref{eq:NonAOCP}).
Note that all indexes $(i,j)$ are in the range $i\lor j<M^{\epsilon}$
in what follows. From \remref{BoundedDegree}, for a given vector
$p=(p(i,j,c))_{0\leq i,0\leq j,0\leq c\leq i}$, the restriction of
$(i,j)$ to $i\lor j<M^{\epsilon}$ indicates that we use only $p(i,j,c)$,
$i\lor j<M^{\epsilon},$ $0\leq c\leq i$ in the calculation. It is
equivalent to setting $p(i,j,c)=0$, $i\lor j\geq M^{\epsilon}$,
$0\leq c\leq i$ while keeping $p(i,j,c)$, $i\lor j<M^{\epsilon}$,
$0\leq c\leq i$ unchanged, which implies asymptotically nodes with
$i\lor j\geq M^{\epsilon}$ are all invulnerable. 

First we define the optimization problem (\ref{eq:OptProgram}) based
on which we present \thmref{theorem3} and \thmref{theorem4}.
\begin{defn}
\label{def:IJtilde}($\tilde{I}$ and $\tilde{J}$ function) Let $x=(x^{\beta})_{\beta\in\Phi^{\epsilon}}$
where $x^{\beta}=x^{\beta}(y,v)$ and $p=(p(i,j,c))_{i,j,0\leq c\leq i}$.
We define the functions $\tilde{I}(y,v,z)$ and $\tilde{J}(y,v,z)$
as 
\begin{alignat}{1}
\tilde{I}(y,v,z) & =\frac{1}{\lambda}\sum_{i\lor j<M^{\epsilon}}j\nonumber \\
 & \left[\sum_{c=0}^{i}p(i,j,c)\mathbb{P}(\text{Bin}(i,x^{i,j,c,c-1})\geq c)-\mathbbm1_{(vj-1=-K)}p(i,j,i)(y^{i}-z^{i})\right],\label{eq:Itilde}\\
\tilde{J}(y,v,z) & =\sum_{i\lor j<M^{\epsilon}}\nonumber \\
 & \left[\sum_{c=0}^{i}p(i,j,c)\mathbb{P}(\text{Bin}(i,x^{i,j,c,c-1})\geq c)-\mathbbm1_{(vj-1=-K)}p(i,j,i)(y^{i}-z^{i})\right].\label{eq:Jtilde}
\end{alignat}
\end{defn}
Note we may write $\tilde{I}(y;v,z)$ and $\tilde{J}(y;v,z)$ to indicate
that we treat $y$ as the variable and $v,z$ as the fixed parameters.
To present the main results, we define the following optimization
problem.
\begin{defn}
Define the following optimization problem.
\begin{flalign}
\tag{OP}\min_{y,v,z} & \quad K\cdot it(y,v,z)+\tilde{J}(y,v,z)\label{eq:OptProgram}\\
\label{eq:Obj_OptProgram}\\
\text{st}\qquad & (1-y)\tilde{H}(y,v)=\lambda v(1-y)\nonumber \\
\label{eq:Hf in y}\\
 & \tilde{I}(y,v,z)=y\nonumber \\
\label{eq:Outlinks_f in y}\\
 & x^{i,j,c,c-1}=\begin{cases}
y & \text{if }K+vj-1\geq0\text{ or }c=0\\
1-(1-y)\frac{(i-c)K}{(i-c+1)K+vj-1} & \text{if }K+vj-1<0\text{ and }1\leq c<i+\frac{K+vj-1}{Ky}\\
0 & \text{otherwise}
\end{cases}\nonumber \\
 & \forall(i,j,c,c-1)\in\Phi^{\epsilon}\nonumber \\
 & 0\leq z\leq y\leq1\nonumber \\
 & y,v,z\in\mathbb{R},\nonumber 
\end{flalign}
where
\begin{align}
it(y,v,z) & =\sum_{i\lor j<M^{\epsilon}}\{\sum_{c=1}^{i}p(i,j,c)\nonumber \\
 & \sum_{m=c}^{i}\sum_{n=0}^{c-1}(m-c+1)\mathbb{P}(\text{Multin}(i,x^{i,j,c,c-1},y-x^{i,j,c,c-1},1-y)=(n,m-n,i-m))\nonumber \\
 & -\mathbbm1_{(vj-1=-K)}p(i,j,i)(y^{i}-z^{i})\},\\
\tilde{H}(y,v) & =\sum_{i\lor j<M^{\epsilon}}\max\{-K,vj-1\}i\sum_{c=1}^{i}p(i,j,c)\nonumber \\
 & \left[\mathbb{P}(\text{Bin}(i-1,y)\geq c-1)-\mathbb{P}(\text{Bin}(i-1,x^{i,j,c,c-1})\geq c)\right],\nonumber \\
\end{align}
where $\text{Bin}(i,y)$ denotes a binomial random variable in $i$
trials with the probability of occurrence $y$, so $\mathbb{P}(\text{Bin}(i,y)\geq c)=\sum_{m=c}^{i}\binom{i}{m}y^{m}(1-y)^{i-m}$
and $\text{Multin}(i,x,y,1-x-y)=(a,b,i-a-b)$ denotes a multinomial
distribution in $i$ trials with the probabilities of occurrence of
each of three types being $x$,$y$ and $1-x-y$, and turns out to
have $a$, $b$ and $i-a-b$ occurrences of each type, respectively,
so $\mathbb{P}(\text{Multin}(i,x,y,1-x-y)=(a,b,i-a-b))=\binom{i}{a,b,i-a-b}x^{a}y^{b}(1-x-y)^{i-a-b}$. 
\end{defn}
Note that $x^{*}=x(y^{*},v^{*})$, the function at the optimal $(y^{*},v^{*})$.
Then we are ready to present the next main theorem about the optimal
control policy. 
\begin{thm}
\label{thm:theorem3} For the asymptotic control problem (\ref{eq:ACP}),
consider a sequence of networks with initial conditions $(P_{n})_{n\geq1}$
satisfying \assuref{Regularity} where $p=(p(i,j,c))_{i,j,0\leq c\leq i}$
such that $p(i,j,c)=0$ for $i\lor j\geq M^{\epsilon},0\leq c\leq i$
and let $(G_{n})_{n\geq1}$ be the sequence of control policies for
the contagion process on the sequence of networks and $(G_{n})_{n\geq1}$
satisfy \assuref{G_n}. Moreover, let $(y^{*},v^{*},z^{*})$ be the
optimal solution for the optimization problem (\ref{eq:OptProgram}).
If $y^{*}=1$, or $y^{*}\in[0,1)$ and it is a stable fixed point
of $\tilde{I}(y;v^{*},z^{*})$, i.e. $\tilde{I}'(y^{*};v^{*},z^{*})<1$,
the optimal control policy $G_{n}^{*}=(g_{1}^{(n)*},\ldots,g_{m}^{(n)*})$
satisfies that for $0\leq k\leq m-1$,
\[
g_{k+1}^{(n)*}(s,w)=\begin{cases}
\mathbbm1_{(k\geq n\lambda(x^{*})^{i,j,c,c-1})} & \text{if }w=(i,j,c,c-1)\in\Phi^{\epsilon}\text{ except }c=i\text{ and }v^{*}j-1=-K\\
\mathbbm1_{(k\geq n\lambda z^{*})} & \text{if }w=(i,j,i,i-1)\text{ and }v^{*}j-1=-K\\
0 & \text{otherwise},
\end{cases}
\]

where $x^{*}=x(y^{*},v^{*})$ and 
\begin{align}
 & x^{i,j,c,c-1}=\begin{cases}
y & \text{if }K+vj-1\geq0\text{ or }c=0\\
1-(1-y)\frac{(i-c)K}{(i-c+1)K+vj-1} & \text{if }K+vj-1<0\text{ and }1\leq c<i+\frac{K+vj-1}{Ky}\\
0 & \text{otherwise},
\end{cases}
\end{align}
$(i,j,c,c-1)\in\Phi^{\epsilon}$.
\end{thm}
The next theorem states conclusions for the asymptotic fraction of
final defaulted nodes under the optimal policy satisfying \thmref{theorem3}.
\begin{thm}
\label{thm:theorem4}For the asymptotic control problem (\ref{eq:ACP}),
consider a sequence of networks with initial conditions $(P_{n})_{n\geq1}$
satisfying \assuref{Regularity} where $p=(p(i,j,c))_{i,j,0\leq c\leq i}$
such that $p(i,j,c)=0$ for $i\lor j\geq M^{\epsilon},0\leq c\leq i$
and let $(G_{n})_{n\geq1}$ be the sequence of control policies for
the contagion process on the sequence of networks and $(G_{n})_{n\geq1}$
satisfy \assuref{G_n}. Moreover, let $(y^{*},v^{*},z^{*})$ be the
optimal solution for the optimization problem (\ref{eq:OptProgram}).
Then under the optimal control policy $G_{n}^{*}$ as in \thmref{theorem3},
we have the following conclusions for the asymptotic fraction of final
defaulted nodes. 
\begin{enumerate}
\item \label{enu:1OfThm4}If $y^{*}=1$, then asymptotically almost all
nodes default during the default contagion process, i.e. 
\[
\frac{D_{n}}{n}\overset{p}{\to}1.
\]
\item \label{enu:2OfThm4}If $y^{*}\in[0,1)$ and it is a stable fixed point
of $\tilde{I}(y;v^{*},z^{*})$, i.e. $\tilde{I}'(y^{*};v^{*},z^{*})<1$,
then asymptotically the fraction of final defaulted nodes 
\begin{eqnarray}
\frac{D_{n}}{n} & \overset{p}{\to} & \tilde{J}(y^{*},v^{*},z^{*}),
\end{eqnarray}
where $x^{i,j,c,c-1}$ in $\tilde{J}$ is defined as
\begin{align}
 & x^{i,j,c,c-1}=\begin{cases}
y & \text{if }K+vj-1\geq0\text{ or }c=0\\
1-(1-y)\frac{(i-c)K}{(i-c+1)K+vj-1} & \text{if }K+vj-1<0\text{ and }1\leq c<i+\frac{K+vj-1}{Ky}\\
0 & \text{otherwise}
\end{cases}\nonumber \\
 & \forall(i,j,c,c-1)\in\Phi^{\epsilon}.
\end{align}

Particularly, if $y^{*}=0$ and $\tilde{I}'(0;v^{*},z^{*})<1$, then
\[
\frac{D_{n}}{n}\overset{p}{\to}\sum_{i\lor j<M^{\epsilon}}p(i,j,0),
\]
 i.e. the final defaulted nodes only consist of those having defaulted
initially. 
\end{enumerate}
\end{thm}
In \thmref{theorem4} the first case indicates that the network is
highly vulnerable and interventions are costly, then the regulator
rather lets the whole network default without implementing any interventions,
while in the second case interventions are less expensive or the contagion
effect is not as high, it is best for the regulator to implement interventions
to counteract the contagion process.

\subsection{Discussion and summary}

We observe that (\ref{eq:OptProgram}) is a nonlinear programming
problem and the key to solve it depends on solving the two equations
(\ref{eq:Hf in y}) and (\ref{eq:Outlinks_f in y}). Here we give
an algorithm to solve (\ref{eq:OptProgram}) numerically. 
\begin{lyxalgorithm}
The algorithm for solving (\ref{eq:OptProgram}) numerically.
\begin{enumerate}
\item Let $y=z$ and solve (\ref{eq:Hf in y}) and (\ref{eq:Outlinks_f in y})
for $y$ and $v$, by for example Newton's method, such that $0\leq y\leq1$.
\item Let $v=\frac{1-K}{j}$ for $j>0$ and solve (\ref{eq:Hf in y}) and
(\ref{eq:Outlinks_f in y}) for $y$ and $z$ such that $0\leq z\leq y\leq1$
for each $j\in\{1,\ldots,M^{\epsilon}\}$.
\item Choose among all the solutions above the one that minimize the objective
function (\ref{eq:Obj_OptProgram}).
\end{enumerate}
\end{lyxalgorithm}
\begin{figure}[H]
\centering{}\includegraphics[scale=0.5]{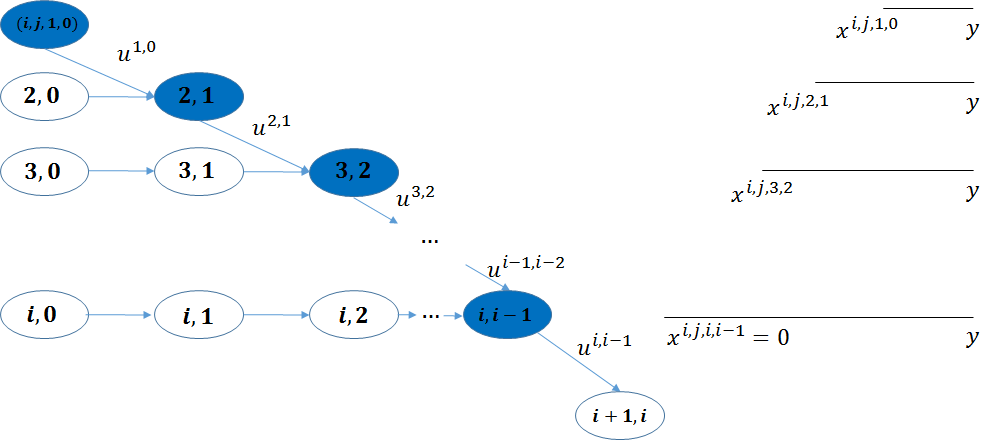}\caption{\label{fig:OptPolicy_Summary}Optimal intervention policy summary}
\end{figure}

Recall that a node is in state $(i,j,c,l)$ if it has in and out degree
pair $(i,j)$, sum of initial equity and accumulative interventions
$c$ (called total buffer) and number of revealed in links $l$. Similar
to \citet{Amini2013}, we call the in links to a node that has ``distance
to default'' of one as ``contagious'' links. So a node in state
$(i,j,1,0)$ has $i$ contagious links and a node in state $(i,j,2,1)$
has $i-1$ contagious links and so on, as shown in \figref{OptPolicy_Summary}
the states in blue color. The insights from the optimal interventions
policy are summarized as follows.
\begin{enumerate}
\item It is never optimal to intervene on a node if it is not selected or
has at least two units of remaining equity when selected. Thus the
optimal control policy described in \thmref{theorem3} only specifies
that whether we should intervene on a node that, when selected, has
``distance to default'' of one, i.e. $l=c-1$. In other words, the
use of interventions is to counteract the effects of contagious links.
\item The optimal control policy depends strongly on $K$, the relative
cost of interventions. At a higher $K$ value, interventions are costly
and the regulator rather lets the contagion to evolve without any
interventions, while at a lower $K$ value the regulator implements
more and more interventions, even a ``complete'' intervention strategy,
that is, intervening on every selected node with ``distance to default''
of one since the beginning of the process. 
\item The optimal control policy is ``monotonic'' with respect to the
number of out degree of a node. Take $v<0$ for example. There exists
a cutoff value of the out degree (specified by $\frac{1-K}{v}$) and
it is only optimal to intervene on a node with out degree larger than
this cutoff value and not otherwise, regardless of its in degree,
total buffer and revealed in links. For nodes with out degree equal
to the cutoff value, we have the singular control case that only those
with state $(i,j,i,i-1)$\textemdash total buffer equal to the in
degree and one unit larger than the number of revealed in links\textemdash needs
interventions and the starting time of interventions is specified
by the variable $z$ from the optimization problem (\ref{eq:OptProgram}). 
\item For nodes that are candidates to receive interventions, the starting
time of interventions (specified by the definition of $x^{i,j,c,c-1}$)
is ``monotonic'' in terms of the total buffer. The higher the current
equity is, the earlier we should begin to intervene as illustrated
in \figref{OptPolicy_Summary} that $x^{i,j,c,c-1}$ is decreasing
in $c$. Moreover, the starting time is also ``monotonic'' in terms
of the in degree and the out degree. Again take $v<0$ for example.
The smaller the in degree is or larger the out degree is, the earlier
the intervention should begin. So we focus on systematically important
nodes and nodes that are close to invulnerability. In other words,
we bailout the system by protecting the nodes that would incur large
impact on the network after defaulting and by bringing nodes into
invulnerable states. Note it is counterintuitive that if there is
a positive fraction of the network strongly interlinked by contagious
links, we do not prioritize interventions on them unless they have
large out degrees. 
\item Once we have begun intervening on a node we should keep implementing
interventions on it every time it is selected. In other words, we
do not allow nodes that have received interventions to default.
\end{enumerate}
Following the optimal policy we are able to achieve earlier termination
time of the contagion process and smaller fraction of final defaulted
nodes indeed. We can quantify the improvement by comparing $\tilde{I}$
and $\tilde{J}$ in \thmref{theorem4} with $I$ and $J$ defined
in \thmref{theorem1}, respectively. Note in the following we suppress
the apostrophe $"*"$ and the indexes $(i,j)$ are in the range $i\lor j<M^{\epsilon}$.
\begin{enumerate}
\item $\tilde{I}(y,v,z)$ plays the same role as $I(y)$ in \thmref{theorem1},
which represents the asymptotic scaled total out degree from the default
set at scaled time $y$. Since 
\begin{align}
 & I(y)-\tilde{I}(y,v,z)\nonumber \\
 & =\frac{1}{\lambda}\sum_{i\lor j<M^{\epsilon}}j\{\sum_{c=0}^{i}p(i,j,c)[\mathbb{P}(\text{Bin}(i,y)\geq c)-\mathbb{P}(\text{Bin}(i,x^{i,j,c,c-1})\geq c)]\nonumber \\
 & +\mathbbm1_{(vj-1=-K)}p(i,j,i)(y^{i}-z^{i})\},\nonumber \\
\end{align}
and note that 
\[
x^{i,j,c,c-1}\leq y\text{ for }(i,j,c,c-1)\in\Phi^{\epsilon},
\]
thus 
\[
\mathbb{P}(\text{Bin}(i,y)\geq c)-\mathbb{P}(\text{Bin}(i,x^{i,j,c,c-1})\geq c)\geq0.
\]
Then for the same initial conditions $p=\left(p(i,j,c)\right)_{i\lor j<M^{\epsilon},0\leq c\leq i}$,
the smallest fixed point of $\tilde{I}(y;v^{*},z^{*})$ is no greater
than that of $I(y),$ which implies that the default contagion process
under optimal interventions terminates no later than under no interventions. 
\item Similarly $\tilde{J}(y,v,z)$ plays the same role as $J(y)$ in \thmref{theorem1},
which represents the asymptotic fraction of final defaulted nodes
under the optimal control policy. The difference
\begin{eqnarray}
 &  & J(y)-\tilde{J}(y,v,z)\nonumber \\
 & = & \sum_{i\lor j<M^{\epsilon}}\{\sum_{c=0}^{i}p(i,j,c)[\mathbb{P}(\text{Bin}(i,y)\geq c)-\mathbb{P}(\text{Bin}(i,x^{i,j,c,c-1})\geq c)]\nonumber \\
 &  & +\mathbbm1_{(vj-1=-K)}p(i,j,i)(y^{i}-z^{i})\}\geq0\nonumber \\
\end{eqnarray}
quantifies the fraction of nodes that are protected from default because
of the optimal control policy.
\end{enumerate}

\section{\label{sec:Proofs}Proofs}

\subsection{Proof of \propref{ODE_Sol_stau}}
\begin{proof}
Assume $u(\tau)=b\coloneqq(b^{\beta})_{\beta\in\Phi}$ for $\tau\in[\tau_{1},\tau_{2})\subseteq[0,\lambda)$,
$0\leq\tau_{1}<\tau_{2}<\infty$ and $b^{\beta}\in\{0,1\}$. Note
that the ODEs are ``separable'' in that $s_{\tau}^{i,j,c,l}$ only
depends on the entries of $s_{\tau}$ with the same $(i,j)$, so we
can only focus on the system of ODEs with the same $(i,j)$. For the
same $(i,j)$, define $\Gamma_{i,j}\coloneqq\{(c,l):0\leq l<c\leq i\text{ or }c=i+1,l=i\}$.
For simplicity purpose after dropping $i,j$ in the superscripts,
we have the set of ODEs for $\tau\in[\tau_{1},\tau_{2})$, 
\begin{eqnarray}
\frac{ds_{\tau}^{c,0}}{d\tau} & = & -\frac{is_{\tau}^{c,0}}{\lambda-\tau}\qquad\text{for \ensuremath{1\leq c\leq i}},\nonumber \\
\frac{ds_{\tau}^{c,l}}{d\tau} & = & \frac{(i-l+1)s_{\tau}^{c,l-1}}{\lambda-\tau}-\frac{(i-l)s_{\tau}^{c,l}}{\lambda-\tau}\nonumber \\
 &  & \text{ for }3\leq c\leq i,1\leq l\leq c-2,\nonumber \\
\frac{ds_{\tau}^{c,c-1}}{d\tau} & = & \frac{(i-c+2)s_{\tau}^{c-1,c-2}}{\lambda-\tau}b^{c-1,c-2}+\frac{(i-c+2)s_{\tau}^{c,c-2}}{\lambda-\tau}-\frac{(i-c+1)s_{\tau}^{c,c-1}}{\lambda-\tau}\nonumber \\
 &  & \text{ for }2\leq c\leq i,\nonumber \\
\frac{ds_{\tau}^{i,j,i+1,i}}{d\tau} & = & \frac{s_{\tau}^{i,i-1}}{\lambda-\tau}b^{i,i-1},\nonumber \\
\end{eqnarray}
with the initial condition $s_{\tau_{1}}=s_{1}\coloneqq(s_{1}^{c,l})_{(c,l)\in\Gamma_{i,j}}$.
Letting $t=-\ln(\lambda-\tau)$, $t_{1}=-\ln(\lambda-\tau_{1})$ and
$t_{2}=-\ln(\lambda-\tau_{2})$, we have an autonomous system of ODEs
for $t\in[t_{1},t_{2})$ that
\begin{align}
\tag{AODE}\label{eq:AODEPart}\\
 & \frac{ds_{t}^{c,0}}{dt}=-is_{t}^{c,0}\text{ for }1\leq c\leq i,\label{eq:ODEPart_s_c0}\\
 & \frac{ds_{t}^{c,l}}{dt}=(i-l+1)s_{t}^{c,l-1}-(i-l)s_{t}^{c,l}\nonumber \\
 & \text{ for }3\leq c\leq i,\;1\leq l\leq c-2,\label{eq:ODEPart_s_cl}\\
 & \frac{ds_{t}^{c,c-1}}{dt}=(i-c+2)s_{t}^{c-1,c-2}b^{c-1,c-2}+(i-c+2)s_{t}^{c,c-2}-(i-c+1)s_{t}^{c,c-1}\nonumber \\
 & \text{ for }2\leq c\leq i,\label{eq:ODEPart_s_cc-1}\\
 & \frac{ds_{t}^{i+1,i}}{dt}=s_{t}^{i,i-1}b^{i,i-1},\label{eq:ODEPart_s_iplus1i}
\end{align}
with the initial condition $s_{t_{1}}=s_{1}$.
\begin{lem}
\label{lem:PartODE_Sol_AnyInitial_t}Let $s_{t}=(s_{t}^{c,l})_{(c,l)\in\Gamma_{i,j}}$
satisfy the system of ordinary differential equations in (\ref{eq:AODEPart})
with the initial conditions $s_{t_{1}}=s_{1}\coloneqq(s_{1}^{c,l})_{(c,l)\in\Gamma_{i,j}}$in
the interval $[t_{1},t_{2})$ and assume $u(t)$ is a constant vector
function $u(t)=b\coloneqq(b^{\beta})_{\beta\in\Phi}$ where $b^{\beta}\in\{0,1\}$
on $[t_{1},t_{2})$, then the solution $s_{t}$ on $[t_{1},t_{2})$
is
\begin{align}
s_{t}^{c,l} & =e^{(i-l)(t_{1}-t)}\sum_{r=0}^{l}s_{1}^{c,r}\binom{i-r}{l-r}(1-e^{t_{1}-t})^{l-r}\nonumber \\
 & \text{ for }2\leq c\leq i,\;0\leq l\leq c-2,\label{eq:Sol_s_cl}\\
s_{t}^{c,c-1} & =e^{(i-c+1)(t_{1}-t)}\sum_{r=0}^{c-1}\sum_{q=r+1}^{c}\prod_{k=q}^{c-1}b^{k,k-1}s_{1}^{q,r}\binom{i-r}{c-1-r}(1-e^{t_{1}-t})^{c-1-r}\nonumber \\
 & \text{ for }1\leq c\leq i,\label{eq:Sol_s_cc-1}\\
s_{t}^{i+1,i} & =s_{1}^{i+1,i}+\sum_{r=0}^{i-1}\sum_{q=r+1}^{i}\prod_{k=q}^{i}b^{k,k-1}s_{1}^{q,r}(1-e^{t_{1}-t})^{i-r}\label{eq:Sol_s_iplus1i}
\end{align}
where $\prod_{k=c}^{c-1}b^{k,k-1}\coloneqq1$ .
\end{lem}
\begin{proof}
We first prove (\ref{eq:Sol_s_cl}) by induction on $l$ for the same
$c$, $2\leq c\leq i$. As the base case $l=0$, (\ref{eq:ODEPart_s_c0})
admits the solution that 
\begin{equation}
s_{t}^{c,0}=s_{1}^{c,0}e^{i(t_{0}-t)},\label{eq:basecase}
\end{equation}
which is consistent with (\ref{eq:Sol_s_cl}). Assume (\ref{eq:Sol_s_cl})
with $l$ replaced by $l-1$, $1\leq l\leq c-2$ is the solution for
(\ref{eq:ODEPart_s_cl}), then it follows from (\ref{eq:ODEPart_s_cl})
that
\begin{eqnarray}
s_{t}^{c,l} & = & s_{1}^{c,l}e^{(i-l)(t_{1}-t)}+\int_{t_{1}}^{t}(i-l+1)s_{y}^{c,l-1}e^{(i-l)(y-t)}dy\nonumber \\
 & = & s_{1}^{c,l}e^{(i-l)(t_{1}-t)}+(i-l+1)\int_{t_{1}}^{t}e^{(i-l+1)(t_{1}-y)}e^{(i-l)(y-t)}\sum_{r=0}^{l-1}s_{1}^{c,r}\binom{i-r}{l-1-r}(1-e^{t_{1}-y})^{l-1-r}dy\nonumber \\
 & = & s_{1}^{c,l}e^{(i-l)(t_{1}-t)}+(i-l+1)e^{(i-l)(t_{1}-t)}\sum_{r=0}^{l-1}s_{1}^{c,r}\binom{i-r}{l-1-r}\int_{t_{1}}^{t}e^{t_{1}-y}(1-e^{t_{1}-y})^{l-1-r}dy\nonumber \\
 & = & s_{1}^{c,l}e^{(i-l)(t_{1}-t)}+e^{(i-l)(t_{1}-t)}\sum_{r=0}^{l-1}s_{1}^{c,r}\binom{i-r}{l-1-r}\frac{i-l+1}{l-r}(1-e^{t_{1}-t})^{l-r}\nonumber \\
 & = & s_{1}^{c,l}e^{(i-l)(t_{1}-t)}+e^{(i-l)(t_{1}-t)}\sum_{r=0}^{l-1}s_{1}^{c,r}\binom{i-r}{l-r}(1-e^{t_{1}-t})^{l-r},\nonumber \\
\end{eqnarray}
which is the same as (\ref{eq:Sol_s_cl}). 

For (\ref{eq:Sol_s_cc-1}) we prove by induction on $c$, $1\leq c\leq i$.
As the base case $c=1$, the solution of (\ref{eq:ODEPart_s_c0})
is that
\[
s_{t}^{1,0}=s_{1}^{1,0}e^{i(t_{1}-t)},
\]
which is consistent with (\ref{eq:Sol_s_cc-1}). Assume (\ref{eq:Sol_s_cc-1})
with $c$ replaced by $c-1$, $1\leq c\leq i$ is the solution for
(\ref{eq:ODEPart_s_cc-1}), then it follows from (\ref{eq:ODEPart_s_cc-1})
that
\begin{align}
s_{t}^{c,c-1} & =s_{1}^{c,c-1}e^{(i-c+1)(t_{1}-t)}+(i-c+2)\int_{t_{1}}^{t}(s_{y}^{c-1,c-2}b^{c-1,c-2}+s_{y}^{c,c-2})e^{(i-c+1)(y-t)}dy,\label{eq:Sol_scc-1_proof}
\end{align}
where by the induction assumption and (\ref{eq:Sol_s_cl}) that the
second term
\begin{align}
 & (i-c+2)\int_{t_{1}}^{t}(s_{y}^{c-1,c-2}b^{c-1,c-2}+s_{y}^{c,c-2})e^{(i-c+1)(y-t)}dy\nonumber \\
= & (i-c+2)\int_{t_{1}}^{t}e^{(i-c+1)(y-t)}e^{(i-c+2)(t_{1}-y)}[\sum_{r=0}^{c-2}\sum_{q=r+1}^{c-1}\prod_{k=q}^{c-1}b^{k,k-1}s_{1}^{q,r}\binom{i-r}{c-2-r}(1-e^{t_{1}-y})^{c-2-r}\nonumber \\
 & +\sum_{r=0}^{c-2}s_{1}^{c,r}\binom{i-r}{c-2-r}(1-e^{t_{1}-y})^{c-2-r}]dy\nonumber \\
= & (i-c+2)e^{(i-c+1)(t_{1}-t)}[\sum_{r=0}^{c-2}\sum_{q=r+1}^{c-1}\prod_{k=q}^{c-1}b^{k,k-1}s_{1}^{q,r}\binom{i-r}{c-2-r}\int_{t_{1}}^{t}e^{t_{1}-y}(1-e^{t_{1}-y})^{c-2-r}dy\nonumber \\
 & +\sum_{r=0}^{c-2}s_{1}^{c,r}\binom{i-r}{c-2-r}\int_{t_{1}}^{t}e^{t_{1}-y}(1-e^{t_{1}-y})^{c-2-r}dy]\nonumber \\
= & e^{(i-c+1)(t_{1}-t)}[\sum_{r=0}^{c-2}\sum_{q=r+1}^{c-1}\prod_{k=q}^{c-1}b^{k,k-1}s_{1}^{q,r}\binom{i-r}{c-2-r}\frac{i-c+2}{c-1-r}(1-e^{t_{1}-t})^{c-1-r}\nonumber \\
 & +\sum_{r=0}^{c-2}s_{1}^{c,r}\binom{i-r}{c-2-r}\frac{i-c+2}{c-1-r}(1-e^{t_{1}-t})^{c-1-r}]\nonumber \\
= & e^{(i-c+1)(t_{1}-t)}[\sum_{r=0}^{c-2}\sum_{q=r+1}^{c-1}\prod_{k=q}^{c-1}b^{k,k-1}s_{1}^{q,r}\binom{i-r}{c-1-r}(1-e^{t_{1}-t})^{c-1-r}\nonumber \\
 & +\sum_{r=0}^{c-2}\prod_{k=c}^{c-1}b^{k,k-1}s_{1}^{c,r}\binom{i-r}{c-1-r}(1-e^{t_{1}-t})^{c-1-r}]\nonumber \\
= & e^{(i-c+1)(t_{1}-t)}\sum_{r=0}^{c-2}\sum_{q=r+1}^{c}\prod_{k=q}^{c-1}b^{k,k-1}s_{1}^{q,r}\binom{i-r}{c-1-r}(1-e^{t_{1}-t})^{c-1-r}.\nonumber \\
\label{eq:ODESol_Interstep}
\end{align}
Substituting (\ref{eq:ODESol_Interstep}) back into (\ref{eq:Sol_scc-1_proof})
yields
\begin{align}
s_{t}^{c,c-1} & =s_{1}^{c,c-1}e^{(i-c+1)(t_{1}-t)}+e^{(i-c+1)(t_{1}-t)}\sum_{r=0}^{c-2}\sum_{q=r+1}^{c}\prod_{k=q}^{c-1}b^{k,k-1}s_{1}^{q,r}\binom{i-r}{c-1-r}(1-e^{t_{1}-t})^{c-1-r}\nonumber \\
 & =e^{(i-c+1)(t_{1}-t)}\sum_{r=0}^{c-1}\sum_{q=r+1}^{c}\prod_{k=q}^{c-1}b^{k,k-1}s_{1}^{q,r}\binom{i-r}{c-1-r}(1-e^{t_{1}-t})^{c-1-r},\nonumber \\
\end{align}
which proves (\ref{eq:Sol_s_cc-1}). Then by (\ref{eq:Sol_s_cc-1})
with $c=i$, we have the expression for $s_{t}^{i,i-1}$and (\ref{eq:ODEPart_s_iplus1i})
allows the solution that

\begin{align}
s_{t}^{i+1,i} & =s_{1}^{i+1,i}+\int_{t_{1}}^{t}s_{y}^{i,i-1}b^{i,i-1}dy\nonumber \\
 & =s_{1}^{i+1,i}+\int_{t_{1}}^{t}e^{t_{1}-y}\sum_{r=0}^{i-1}\sum_{q=r+1}^{i}\prod_{k=q}^{i}b^{k,k-1}s_{1}^{q,r}(i-r)(1-e^{t_{1}-y})^{i-1-r}dy\nonumber \\
 & =s_{1}^{i+1,i}+\sum_{r=0}^{i-1}\sum_{q=r+1}^{i}\prod_{k=q}^{i}b^{k,k-1}s_{1}^{q,r}\int_{t_{1}}^{t}(i-r)e^{t_{1}-y}(1-e^{t_{1}-y})^{i-1-r}dy\nonumber \\
 & =s_{1}^{i+1,i}+\sum_{r=0}^{i-1}\sum_{q=r+1}^{i}\prod_{k=q}^{i}b^{k,k-1}s_{1}^{q,r}(1-e^{t_{1}-t})^{i-r},
\end{align}
which proves (\ref{eq:Sol_s_iplus1i}).
\end{proof}
By changing the variable $t$ to $\tau$ by $t=-\ln(\lambda-\tau)$,
(\ref{eq:Sol_stau_cl}), (\ref{eq:Sol_stau_cc-1}) and (\ref{eq:Sol_stau_iplus1i})
follow from \lemref{PartODE_Sol_AnyInitial_t}. Let the initial condition
be $s_{1}^{i,j,c,l}=p(i,j,c)\mathbbm1_{(l=0)}$ for $(i,j,c,l)\in\Gamma$
at $\tau_{1}=0$, then (\ref{eq:ODE_stau_cl_from0}) follows.
\end{proof}

\subsection{Proof of \propref{Sconv}}
\begin{proof}
For the following proof we need to adapt the Wormald's theorem \lemref{(Wormald-1999)}
in the appendix. For notational convenience we suppress the tilde
sign for $\tilde{IT}$, $\tilde{it}$. Since $\frac{m}{n}\rightarrow\lambda$
as $n\rightarrow\infty$, for the given $\epsilon$ and $\hat{\lambda}=\lambda-\epsilon$,
we can find $n_{0}\in\mathbb{N}$, such that $0<\hat{\lambda}<\frac{m}{n}<\lambda+0.1$
for $n\geq n_{0}$. Let $z=(z^{\alpha})_{\alpha\in\Gamma^{\epsilon}}$
and 
\[
U=\left\{ (\tau,z,it)\in\mathbb{R}^{\left|\Gamma^{\epsilon}\right|+2}:\:-\epsilon<\tau<\hat{\lambda},\,-\epsilon<z^{\alpha}<1.1\,,-\epsilon<it<\lambda+0.1\right\} ,
\]
then $U$ contains the closure of 
\[
\left\{ (0,z,0):\,\mathbb{P}(S_{0}^{\alpha}=z^{\alpha}n,\,\forall\alpha\in\Gamma^{\epsilon},IT_{0}=0)\neq0\text{ for some }n\right\} .
\]

Define the stopping time $T_{U}=\min\{1\leq k\leq m:(\frac{k}{n},\frac{S_{k}}{n},\frac{IT_{k}}{n})\notin U\}$.

By \defref{S} and definition of $IT_{k}$, $0\leq S_{k}^{\alpha}\leq n$,
$\alpha\in\Gamma^{\epsilon}$ and $0\leq IT_{k}\leq(\lambda+0.1)n$
hold $\forall$$k\geq0$ and $n\geq n_{0}$. Recall that $S_{k}=(S_{k}^{\alpha})_{\alpha\in\Gamma^{\epsilon}}$
and $\frac{S_{k}}{n}=(\frac{S_{k}^{\alpha}}{n})_{\alpha\in\Gamma^{\epsilon}}$.
The following conditions hold:
\begin{enumerate}
\item \label{enu:Sconvs1}For $0\leq k<T_{U}$ and $\alpha\in\Gamma^{\epsilon}$,
\begin{align}
\left|S_{k+1}^{\alpha}-S_{k}^{\alpha}\right| & \leq1,\nonumber \\
|IT_{k+1}-IT_{k}| & \leq1,\label{eq:BoundedChange}
\end{align}
i.e. $\rho_{1}=1$. 
\item There exists $\rho_{2}=O(n^{-1})$ such that for $0\leq k<T_{U}$
and $\alpha\in\Gamma^{\epsilon}$, 
\begin{align}
\left|\mathbb{E}\left(S_{k+1}^{\alpha}-S_{k}^{\alpha}\mid\mathcal{F}_{k}\right)-h^{\alpha}\left(\frac{k}{n},\frac{S_{k}}{n}\right)\right| & \leq\rho_{2},\nonumber \\
\left|\mathbb{E}\left(IT_{k+1}-IT_{k}\mid\mathcal{F}_{k}\right)-h_{0}\left(\frac{k}{n},\frac{S_{k}}{n}\right)\right| & \leq\rho_{2},\label{eq:IT trend}
\end{align}
where $h=(h^{\alpha})_{\alpha\in\Gamma^{\epsilon}}$ and $h_{0}$
are
\begin{align}
h^{i,j,c,l}\left(t,z\right) & =\begin{cases}
-\frac{iz^{i,j,c,0}}{\lambda-t} & \text{if }1\leq c\leq i,l=0\\
\frac{(i-l+1)z^{i,j,c,l-1}}{\lambda-t}-\frac{(i-l)z^{i,j,c,l}}{\lambda-t} & \text{if }3\leq c\leq i,1\leq l\leq c-2\\
\frac{(i-c+2)z^{i,j,c-1,c-2}}{\lambda-t}u^{i,j,c-1,c-2}(t)\\
+\frac{(i-c+2)z^{i,j,c,c-2}}{\lambda-t}-\frac{(i-c+1)z^{i,j,c,c-1}}{\lambda-t} & \text{if }2\leq c\leq i\\
\frac{z^{i,j,i,i-1}}{\lambda-t}u^{i,j,i,i-1}(t) & \text{if }c=i,l=i-1
\end{cases}\nonumber \\
h_{0}(t,z) & =\sum_{(i,j,c,c-1)\in\Phi^{\epsilon}}\frac{(i-c+1)z^{i,j,c,c-1}}{\lambda-t}u^{i,j,c,c-1}(t).
\end{align}
 In particular, (\ref{eq:IT trend}) follows from
\begin{align}
 & \left|\sum_{(i,j,c,c-1)\in\Phi^{\epsilon}}\frac{(i-c+1)S_{k}^{i,j,c,c-1}}{m-k}u^{i,j,c,c-1}(\frac{k}{n})-\sum_{(i,j,c,c-1)\in\Phi^{\epsilon}}\frac{(i-c+1)\frac{S_{k}^{i,j,c,c-1}}{n}}{\lambda-\frac{k}{n}}u^{i,j,c,c-1}(\frac{k}{n})\right|\nonumber \\
\leq & \sum_{(i,j,c,c-1)\in\Phi^{\epsilon}}\left|\frac{(i-c+1)\frac{S_{k}^{i,j,c,c-1}}{n}}{\frac{m}{n}-\frac{k}{n}}u^{i,j,c,c-1}(\frac{k}{n})-\frac{(i-c+1)\frac{S_{k}^{i,j,c,c-1}}{n}}{\lambda-\frac{k}{n}}u^{i,j,c,c-1}(\frac{k}{n})\right|\nonumber \\
= & O(n^{-1}).\nonumber \\
\end{align}
\end{enumerate}
However, for $\beta\in\Phi^{\epsilon}$, $h^{\beta}$ and $h_{0}$
are not Lipschitz continuous because $u^{\beta}(\tau)$ can have step
changes on $[0,\lambda)$. So we need to adapt the proof. Note that
$u^{\beta}(\tau)$ is piecewise constant $\{0,1\}$ valued function
thus $h^{\beta}(\tau,s)$ and $h_{0}(\tau,s)$ are Lipschitz continuous
in each interval where $\tilde{u}=(u^{\beta})_{\beta\in\Phi^{\epsilon}}$
is a constant vector and then we can apply \lemref{(Wormald-1999)}.
In the following define $s_{\tau}(\tau',x)$ as the solution of the
ODEs, 
\begin{align}
\frac{d}{d\tau}s_{\tau} & =h(\tau,s_{\tau}),
\end{align}
 which are understood as for $\alpha\in\Gamma^{\epsilon}$, 
\begin{equation}
\frac{d}{d\tau}s_{\tau}^{\alpha}=h^{\alpha}(\tau,s_{\tau}),
\end{equation}
 with initial condition at $\tau'$, $s_{\tau'}=x\coloneqq(x^{\alpha})_{\alpha\in\Gamma^{\epsilon}}$. 

In what follows define the points where any component of $\tilde{u}(\tau)$
has a step change. $\tau_{l}\coloneqq\inf\{\tau>\tau_{l-1}:u^{\beta}(\tau)\text{ has a step change for some }\ensuremath{\beta\in\Phi^{\epsilon}}\}\land\hat{\lambda}$
for $l\geq1$ and $\tau_{0}\coloneqq0$. Also let $k_{l}=\left\lceil n\tau_{l}\right\rceil $,
where $\left\lceil \cdot\right\rceil $ is the ceiling function. As
a result, $k_{l}-1<n\tau_{l}\leq k_{l}$. Recall the initial condition
$s_{0}=(s_{0}^{\alpha})_{\alpha\in\Gamma^{\epsilon}}$ with $s_{0}^{i,j,c,l}=p(i,j,c)\mathbbm1_{(l=0)}$.
Because every $u^{\beta}$ for $\beta\in\Phi^{\epsilon}$ has only
a finite number of step changes on $[0,\lambda)$ and on the other
hand $\Phi^{\epsilon}$ is a finite set, there are in total a finite
number of step changes for all the functions in $\tilde{u}$ on $[0,\lambda)$. 

Then by \lemref{(Wormald-1999)}, let $\rho=n^{-\frac{1}{4}}$, it
follows that 
\begin{equation}
\sup_{0\leq k\leq k_{1}-1}\frac{S_{k}^{\alpha}}{n}-s_{\frac{k}{n}}^{\alpha}(0,s_{0})=O(n^{-\frac{1}{4}})
\end{equation}
with probability $1-O(n^{\frac{1}{4}}\exp(-n^{\frac{1}{4}}))$, $\forall\:\alpha\in\Gamma^{\epsilon}$.
Note that we will write ``with probability $1-O(n^{\frac{1}{4}}\exp(-n^{\frac{1}{4}}))$''
as whp hereinafter. 

In particular, we have that
\begin{equation}
\frac{S_{k_{1}-1}^{\alpha}}{n}-s_{\frac{k_{1}-1}{n}}^{\alpha}(0,s_{0})=O(n^{-\frac{1}{4}})\;\text{ whp.}\label{eq:Ssbeforek1}
\end{equation}

Additionally by \lemref{(Wormald-1999)} again we have that
\begin{equation}
\sup_{k_{1}\leq k\leq k_{2}-1}\frac{S_{k}^{\alpha}}{n}-s_{\frac{k}{n}}^{\alpha}(\frac{k_{1}}{n},\frac{S_{k_{1}}}{n})=O(n^{-\frac{1}{4}})\;\text{ whp.}\label{eq:Ss_afterk1}
\end{equation}

Note that
\begin{equation}
\left|\frac{S_{k_{1}}^{\alpha}}{n}-\frac{S_{k_{1}-1}^{\alpha}}{n}\right|\leq\frac{1}{n}\,\forall\alpha\in\Gamma^{\epsilon},\label{eq:S_bounded}
\end{equation}
and by the Lipschitz continuity of $s_{\tau}^{\alpha}(0,s_{0})$ on
$(0,\tau_{1}^{-})$, 
\begin{equation}
s_{\frac{k_{1}-1}{n}}^{\alpha}(0,s_{0})-s_{\tau_{1}}^{\alpha}(0,s_{0})=O(n^{-1}).\label{eq:s_Lip}
\end{equation}
So by (\ref{eq:Ssbeforek1}), (\ref{eq:S_bounded}) and (\ref{eq:s_Lip}),
we have
\begin{align}
\left|\frac{S_{k_{1}}^{\alpha}}{n}-s_{\tau_{1}}^{\alpha}(0,s_{0})\right| & \leq\left|\frac{S_{k_{1}}^{\alpha}}{n}-\frac{S_{k_{1}-1}^{\alpha}}{n}\right|+\left|\frac{S_{k_{1}-1}^{\alpha}}{n}-s_{\frac{k_{1}-1}{n}}^{\alpha}(0,s_{0})\right|\nonumber \\
 & +\left|s_{\frac{k_{1}-1}{n}}^{\alpha}(0,s_{0})-s_{\tau_{1}}^{\alpha}(0,s_{0})\right|\nonumber \\
 & =n^{-1}+O(n^{-\frac{1}{4}})+O(n^{-1})\;\text{ whp.}\label{eq:Ss_tau1_conv}
\end{align}
Thus we have that
\[
\left\Vert \frac{S_{k_{1}}}{n}-s_{\tau_{1}}(0,s_{0})\right\Vert =O(n^{-\frac{1}{4}})+O(n^{-1})\;\text{ whp.}
\]
where $\left\Vert \eta\right\Vert $ is the norm for $\eta\in\mathbb{R}^{|\Gamma^{\epsilon}|}$.

From \propref{ODE_Sol_stau} we see that the derivative of $s_{\tau}^{\alpha}$
with respect to the initial time and initial condition is bounded
for $\tau_{1}<\tau<\hat{\lambda}<\lambda$, i.e.
\begin{equation}
\left\Vert \frac{\partial s_{\tau}^{\alpha}(\tau',x)}{\partial(\tau',x)}\right\Vert \leq M_{1}<\infty\label{eq:s_ContDep}
\end{equation}
where $M_{1}$ is a constant.  Recall that $|\frac{k_{1}}{n}-\tau_{1}|<n^{-1}$,
so by (\ref{eq:Ss_tau1_conv}) and (\ref{eq:s_ContDep}), it follows
that 
\begin{align}
 & s_{\tau}^{\alpha}(\frac{k_{1}}{n},\frac{S_{k_{1}}}{n})-s_{\tau}^{\alpha}(\tau_{1},s_{\tau_{1}}(0,s_{0}))\nonumber \\
 & =s_{\tau}^{\alpha}(\frac{k_{1}}{n},\frac{S_{k_{1}}}{n})-s_{\tau}^{\alpha}(0,s_{0})\nonumber \\
 & =O(n^{-\frac{1}{4}})+O(n^{-1})\text{\; whp},\label{eq:s_closeness}
\end{align}
 for $\tau\in(\tau_{1},\tau_{2})$. So it follows that
\[
\sup_{k_{1}\leq k\leq k_{2}-1}\frac{S_{k}^{\alpha}}{n}-s_{\frac{k}{n}}^{\alpha}(0,s_{0})=O(n^{-\frac{1}{4}})\;\text{ whp.}
\]
Similarly for $IT_{k}$, define $it_{\tau}(\tau',x,y)$ as the solution
of
\[
\frac{d}{d\tau}it_{\tau}=h_{0}(\tau,s_{\tau}),
\]
 with the initial condition at $\tau'$, $(s_{\tau'},it_{\tau'})=(x,y)$.
Applying \lemref{(Wormald-1999)} for $IT_{k}$ and $it_{\tau}$ gives
that, 
\begin{align}
\sup_{0\leq k\leq k_{1}-1}\frac{IT_{k}}{n}-it_{\frac{k}{n}}(0,s_{0},0) & =O(n^{-\frac{1}{4}})\;\text{ whp},\nonumber \\
\sup_{k_{1}\leq k\leq k_{2}-1}\frac{IT_{k}}{n}-it_{\frac{k}{n}}(\frac{k_{1}}{n},\frac{S_{k_{1}}}{n},\frac{IT_{k_{1}}}{n}) & =O(n^{-\frac{1}{4}})\;\text{ whp}.
\end{align}
In particular, 
\begin{equation}
\frac{IT_{k_{1}-1}}{n}-it_{\frac{k_{1}-1}{n}}(0,s_{0},0)=O(n^{-\frac{1}{4}})\;\text{ whp},\label{eq:ITit}
\end{equation}
Further note that
\begin{equation}
\left|\frac{IT_{k_{1}}}{n}-\frac{IT_{k_{1}-1}}{n}\right|\leq\frac{1}{n}\,\forall\alpha\in\Gamma^{\epsilon},\label{eq:IT_bounded}
\end{equation}
 and by the Lipschitz continuity of $it_{\tau}(0,s_{0},0)$ on $(0,\tau_{1}^{-})$,
\begin{equation}
it_{\frac{k_{1}-1}{n}}(0,s_{0},0)-it_{\tau_{1}}(0,s_{0},0)=O(n^{-1}).\label{eq:it_Lip}
\end{equation}
So by (\ref{eq:ITit}), (\ref{eq:IT_bounded}) and (\ref{eq:it_Lip})
we have
\begin{align}
\left|\frac{IT_{k_{1}}}{n}-it_{\tau_{1}}(0,s_{0},0)\right| & \leq\left|\frac{IT_{k_{1}}}{n}-\frac{IT_{k_{1}-1}}{n}\right|+\left|\frac{IT_{k_{1}-1}}{n}-it_{\frac{k_{1}-1}{n}}(0,s_{0},0)\right|\nonumber \\
 & +\left|it_{\frac{k_{1}-1}{n}}(0,s_{0},0)-it_{\tau_{1}}(0,s_{0},0)\right|\nonumber \\
 & =n^{-1}+O(n^{-\frac{1}{4}})+O(n^{-1})\;\text{ whp.}\label{eq:ITit_tau1_conv1}
\end{align}
Recall we have proved in (\ref{eq:Ss_tau1_conv}) that 
\begin{equation}
\left|\frac{S_{k_{1}}^{\alpha}}{n}-s_{\tau_{1}}^{\alpha}(0,s_{0})\right|\leq O(n^{-\frac{1}{4}})+O(n^{-1})\text{\; whp}.\label{eq:Ss_tau1_conv_1}
\end{equation}
Here we apply the fact we shall prove later that the derivative of
$it_{\tau}$ with respect to the initial time and initial condition
is bounded for $\tau$ in an interval on which $\tilde{u}$ is a constant
vector function and $\tau<\hat{\lambda}$, i.e. 
\begin{equation}
\left\Vert \frac{\partial it_{\tau}(\tau',x,y)}{\partial(\tau',x,y)}\right\Vert \leq M_{2}<\infty\label{eq:it_ContDep}
\end{equation}
for some constant $M_{2}$. Recall that $|\frac{k_{1}}{n}-\tau_{1}|<n^{-1}$,
so by (\ref{eq:ITit_tau1_conv1}), (\ref{eq:Ss_tau1_conv_1}) and
(\ref{eq:it_ContDep}), we have that
\begin{align}
 & it_{\tau}(\frac{k_{1}}{n},\frac{S_{k_{1}}}{n},\frac{IT_{k_{1}}}{n})-it_{\tau}(\tau_{1},s_{\tau_{1}}(0,s_{0}),it_{\tau_{1}}(0,s_{0},0))\nonumber \\
 & =it_{\tau}(\frac{k_{1}}{n},\frac{S_{k_{1}}}{n},\frac{IT_{k_{1}}}{n})-it_{\tau}(0,s_{0},0)\nonumber \\
 & =O(n^{-\frac{1}{4}})+O(n^{-1})\text{\; whp},
\end{align}

for $\tau\in(\tau_{1},\tau_{2})$. So it follows that
\[
\sup_{k_{1}\leq k\leq k_{2}-1}\frac{IT_{k}}{n}-it_{\frac{k}{n}}(0,s_{0},0)=O(n^{-\frac{1}{4}})\;\text{ whp.}
\]
We can repeat the above procedure every time any $u^{\beta}(\tau)$
has a step change, $\beta\in\Phi^{\epsilon}$ and there are only a
finite number of step changes in $[0,\lambda)$. Because $s_{\tau}^{\alpha}\leq1$
and $it_{\tau}\leq\lambda$, $d_{\infty}\left((s_{\tau},it_{\tau}),\partial U\right)\geq0.1\geq Cn^{-\frac{1}{4}}$,
for a sufficiently large constant $C$. Thus the supremum of $\tau$
that $(s_{\tau},it_{\tau})$ can be extended to the boundary of $U$
is $\hat{\lambda}$, i.e. in (\ref{eq:suptime}) of \lemref{(Wormald-1999)},
\begin{align*}
\sigma & =\sup\left\{ \tau\geq0:\,d_{\infty}\left((s_{\tau},it_{\tau}),\partial U\right)\geq Cn^{-\frac{1}{4}}\right\} \\
 & =\hat{\lambda}.
\end{align*}

So it follows that
\begin{align}
\sup_{0\leq k\leq n\hat{\lambda}}\frac{S_{k}^{\alpha}}{n}-s_{\frac{k}{n}}^{\alpha}(0,s_{0}) & =O(n^{-\frac{1}{4}})\;\text{ whp},\nonumber \\
\sup_{0\leq k\leq n\hat{\lambda}}\frac{IT_{k}}{n}-it_{\frac{k}{n}}(0,s_{0},0) & =O(n^{-\frac{1}{4}})\;\text{ whp.}\label{eq:SswithinU}
\end{align}

At last we prove that the derivative of $it_{\tau}$ with respect
to the initial time and initial condition is bounded as in (\ref{eq:it_ContDep}).
Note first that $it_{\tau}$ with initial condition $\bar{s}=(s_{\tau'},it_{\tau'})$
at $\tau=\tau'$ in an interval where $\tilde{u}$ is a constant vector
function $b=(b^{\beta})_{\beta\in\Phi^{\epsilon}}$ satisfies that
\[
it_{\tau}=it_{\tau'}+\int_{\tau'}^{\tau}\sum_{(i,j,c,c-1)\in\Phi^{\epsilon}}\frac{(i-c+1)b^{i,j,c,c-1}}{\lambda-u}s_{u}^{i,j,c,c-1}(\tau',s_{\tau'})du.
\]
First we show that the derivative of $it_{\tau}$ with respect to
the initial condition $\bar{s}$ is bounded. 
\[
\frac{\partial it_{\tau}}{\partial\bar{s}}=e^{\text{last}}+\int_{\tau'}^{\tau}\sum_{(i,j,c,c-1)\in\Phi^{\epsilon}}\frac{(i-c+1)b^{i,j,c,c-1}}{\lambda-u}\frac{\partial s_{u}^{i,j,c,c-1}(\tau',s_{\tau'})}{\partial\bar{s}}du
\]
where $e^{\text{last}}$ is a vector of zeros except an entry of one
at the last. Thus
\[
\left\Vert \frac{\partial it_{\tau}}{\partial\bar{s}}\right\Vert \leq1+\int_{\tau'}^{\tau}\sum_{(i,j,c,c-1)\in\Phi^{\epsilon}}\frac{(i-c+1)b^{i,j,c,c-1}}{\lambda-u}\left\Vert \frac{\partial s_{u}^{i,j,c,c-1}(\tau',s_{\tau'})}{\partial\bar{s}}\right\Vert du.
\]
 By (\ref{eq:s_ContDep}), $\left\Vert \frac{\partial s_{u}^{i,j,c,c-1}}{\partial\bar{s}}\right\Vert <M_{1}$
and thus $\left\Vert \frac{\partial it_{\tau}}{\partial\bar{s}}\right\Vert $
is bounded for $\tau<\hat{\lambda}$. Next we show that the derivative
of $it_{\tau}$ with respect to the initial time $\tau'$ is bounded.
By the Leibniz integral rule, we have
\begin{align}
\frac{\partial it_{\tau}}{\partial\tau'} & =-\sum_{(i,j,c,c-1)\in\Phi^{\epsilon}}\frac{(i-c+1)b^{i,j,c,c-1}}{\lambda-\tau'}s_{\tau'}^{i,j,c,c-1}(\tau',s_{\tau'})\nonumber \\
 & +\int_{\tau'}^{\tau}\sum_{(i,j,c,c-1)\in\Phi^{\epsilon}}\frac{(i-c+1)b^{i,j,c,c-1}}{\lambda-u}\frac{\partial s_{u}^{i,j,c,c-1}(\tau',s_{\tau'})}{\partial\tau'}du,
\end{align}
where $s_{\tau'}^{i,j,c,c-1}(\tau',s_{\tau'})=s_{\tau'}^{i,j,c,c-1}$.
Since by (\ref{eq:s_ContDep}), $\left|\frac{\partial s_{u}^{i,j,c,c-1}(\tau',s_{\tau'})}{\partial\tau'}\right|$
is bounded thus $\left|\frac{\partial it_{\tau}}{\partial\tau'}\right|$
is bounded for $\tau<\hat{\lambda}$. Thus we proved (\ref{eq:it_ContDep}).
\end{proof}

\subsection{Proof of \thmref{theorem1}}
\begin{proof}
For the contagion process without intervention, we relate our model
to the auxiliary model used in the proof in \citet{Amini2013}. 

Recall that in \subsecref{Dynamics} we are given a set of nodes $[n]$
and the sequence of degrees $(d^{-}(v),d^{+}(v))_{v\in[n]}$ as well
as the initial equity levels $\left(e_{0}^{v}\right)_{v\in[n]}$.
For each node $v$ we assign each in stub a number ranging in $\{1,\ldots,d^{-}(v)\}$.
Let $\Sigma^{v}$ be the set of all permutations of the in stubs of
node $v\in[n]$, then a permutation $\tau\in\Sigma^{v}$ specifies
the order in which $v$ receives shocks through the in stubs. 

Because every in stub of $v$ represents one unit of loan, $v$ will
default after $e_{0}^{v}$ of its in stubs have been connected (or
$e_{0}^{v}$ of its in links have been revealed) for every permutation
$\tau\in\Sigma^{v}$. In other words, if we define $\theta(v,\tau)$
to be the number of shocks that $v$ can sustain if the order in which
the in stubs are connected is specified by $\tau$, then it follows
that $\theta(v,\tau)=e_{0},\forall\tau\in\Sigma^{v}$ and
\begin{alignat}{1}
 & \frac{\left|\{(v,\tau)\mid v\in[n],\tau\in\Sigma^{v},d^{-}(v)=i,d^{+}(v)=j,\theta(v,\tau)=c\}\right|}{n\mu(i,j)i!}\nonumber \\
= & \frac{\left|\{v\mid v\in[n],\tau\in\Sigma^{v},d^{-}(v)=i,d^{+}(v)=j,e_{0}^{v}=c\}\right|}{n\mu(i,j)}\nonumber \\
= & \sigma(i,j,c).
\end{alignat}

Then \assuref{Regularity} satisfies the assumption $4.1$ and $4.2$
in \citet{Amini2013}. Moreover, under no intervention, the random
graph generated in \subsecref{Dynamics} conforms to the model defined
in definition $5.4$ in \citet{Amini2013} with in and out degree
sequences $(d^{-}(v),d^{+}(v))_{v\in[n]}$ and default thresholds
$\left(e_{0}^{v}\right)_{v\in[n]}$. So by their theorem $3.8$ we
achieve the conclusions of \thmref{theorem1}. 
\end{proof}

\subsection{Proof of \thmref{theorem3}}
\begin{proof}
To simplify the notations we suppress the apostrophe $"*"$. In \lemref{OC}
we have presented the optimal control policy $(u_{t})_{t\in[t_{0},t_{f}]}$
in terms of $t$,$t_{0}$,$t_{f}$,$t_{s}$,$t^{i,j,c}$. Recall that
in (\ref{eq:x_and_xijc}) we have the following relations,
\begin{eqnarray}
t & = & -\ln(\lambda-\tau),\nonumber \\
t_{0} & = & -\ln\lambda,\nonumber \\
y & = & 1-e^{t_{0}-t_{f}},\nonumber \\
z & = & 1-e^{t_{0}-t_{s}},\nonumber \\
x^{i,j,c,c-1} & = & 1-e^{t_{0}-t^{i,j,c}}\nonumber \\
 & = & \begin{cases}
y & \text{if }K+vj-1\geq0\text{ or }c=0\\
1-(1-y)\frac{(i-c)K}{(i-c+1)K+vj-1} & \text{if }K+vj-1<0\text{ and }1\leq c<i+\frac{K+vj-1}{Ky}\\
0 & \text{otherwise},
\end{cases}\nonumber \\
\end{eqnarray}
so we can change the variable from $t$ to $\tau$. Particularly we
apply mapping $f(t)=1-e^{t_{0}-t}$ which is a strictly increasing
function in $t$, then we have the following correspondences:

\begin{table}[H]
\begin{centering}
\begin{tabular}{|c|c|}
\hline 
Variable & After mapping $1-e^{t_{0}-t}$\tabularnewline
\hline 
\hline 
$t$ & $\frac{\tau}{\lambda}$\tabularnewline
\hline 
$t_{0}$ & $0$\tabularnewline
\hline 
$t^{i,j,c}$ & $x^{i,j,c,c-1}$\tabularnewline
\hline 
$t_{s}$ & $z$\tabularnewline
\hline 
$t_{f}$ & $y$\tabularnewline
\hline 
\end{tabular}
\par\end{centering}
\caption{\label{tab:VarMapping}Correspondence of variables through mapping
$1-e^{t_{0}-t}$. }
\end{table}

We replace each variable $t$,$t_{0}$,$t_{f}$,$t_{s}$,$t^{i,j,c}$
in \lemref{OC} with its corresponding variable in \tabref{VarMapping}
resulting in the expressions for $u_{\tau}^{i,j,c,c-1}$. At last
by \assuref{G_n} on the relations between the control policy $G_{n}=(g_{1}^{(n)},\ldots,g_{m}^{(n)})$
and the function $u$, we have the conclusion in \thmref{theorem3}.
\end{proof}

\subsection{Proof of \thmref{theorem4}}
\begin{proof}
By the definition of $\tilde{I}(y;v,z)$ and $\tilde{J}(y;v,z)$ in
(\ref{eq:Itilde}), $d_{\tau_{f}}^{-}$ and $d_{\tau_{f}}$ with $i\lor j<M^{\epsilon}$
in (\ref{eq:EqOfdWithInt}) at $\tau=\tau_{f}$ becomes 
\begin{alignat}{1}
d_{\tau_{f}}^{-} & =\sum_{i\lor j<M^{\epsilon}}j\left[\sum_{c=0}^{i}p(i,j,c)\mathbb{P}(\text{Bin}(i,x^{i,j,c,c-1})\geq c)\right.\nonumber \\
 & \left.-\mathbbm1_{(vj-1=-K)}p(i,j,i)\left((\frac{\tau_{f}}{\lambda})^{i}-z^{i}\right)\right]-\tau_{f}\nonumber \\
 & =\lambda\left(\tilde{I}(\frac{\tau_{f}}{\lambda};v,z)-\frac{\tau_{f}}{\lambda}\right)\nonumber \\
d_{\tau_{f}} & =\sum_{i\lor j<M^{\epsilon}}\left[\sum_{c=0}^{i}p(i,j,c)\mathbb{P}(\text{Bin}(i,x^{i,j,c,c-1})\geq c)\right.\nonumber \\
 & \left.-\mathbbm1_{(vj-1=-K)}p(i,j,i)\left((\frac{\tau_{f}}{\lambda})^{i}-z^{i}\right)\right]\nonumber \\
 & =\tilde{J}(\frac{\tau_{f}}{\lambda};v,z).
\end{alignat}
Suppose $(y^{*},v^{*},z^{*})$ is an optimal solution for the optimization
problem (\ref{eq:OptProgram}) and note that $y^{*}$ is the fixed
point of $\tilde{I}(y;v^{*},z^{*})$ and $y^{*}=\frac{\tau_{f}^{*}}{\lambda}$.
\begin{enumerate}
\item If $y^{*}=1$, then $\tau_{f}^{*}=\lambda$. By the definition of
$d_{\tau_{f}}^{-}$, it can only occur when $\sum_{i\lor j<M^{\epsilon}}j\sum_{c=0}^{i}p(i,j,c)=\lambda$
and $z^{*}=\frac{\tau_{f}^{*}}{\lambda}=1$, thus we have $d_{\tau_{f}^{*}}=d_{\lambda}=1$,
then by \propref{Conv_obj}, 
\[
\frac{D_{n}}{n}\overset{p}{\rightarrow}1,
\]
which proves (\ref{enu:1OfThm4}) of \thmref{theorem4}.
\item If $y^{*}<1$ and $\tilde{I}'(y^{*};v^{*},z^{*})<1$, then $\tau_{f}^{*}<\lambda$
and $\frac{d}{d\tau}d_{\tau_{f}^{*}}^{-}=\tilde{I}'(\frac{\tau_{f}^{*}}{\lambda};v^{*},z^{*})-1<0$.
Again it follows from \propref{Conv_obj}, 
\[
\frac{D_{n}}{n}\overset{p}{\rightarrow}d_{\tau_{f}^{*}}=\tilde{J}(y^{*};v^{*},z^{*}).
\]
which proves (\ref{enu:2OfThm4}) of \thmref{theorem4}. This concludes
the proof of \thmref{theorem4}. 
\end{enumerate}
It is important to note that the two cases in \thmref{theorem4} corresponds
to $\tau_{f}^{*}=\lambda$, and $\tau_{f}^{*}<\lambda$, $\frac{d}{d\tau}d_{\tau_{f}^{*}}^{-}<0$,
respectively. By \propref{Conv_obj} they guarantees that the limits
of $\mathbb{E}\frac{IT_{n}(G_{n},P_{n})}{n}$ and $\mathbb{E}\frac{D_{n}(G_{n},P_{n})}{n}$
in (\ref{eq:ACP}) as $n\rightarrow\infty$ are well defined, which
are $it_{\tau_{f}}$ and $d_{\tau_{f}}$, respectively.
\end{proof}

\part{\label{part:NumericalEg}Numerical Experiments}

\section{Introduction}

Consider a sequence of networks with the number of nodes $n$ growing
to infinity, whose in and out degrees are between 1 and 10, and each
node's in degree equal to its own out degree, i.e. $d^{-}(v)=d^{+}(v)$,
$v\in[n]$, respectively, so we call either the in or out degree as
the degree of the node. This allows us to combine two indexes $i$
and $j$ into one index $i$, so the state of a node becomes $(i,c,l)$
and the empirical probability $P_{n}$ and the limiting probability
$p$ of the degree and initial equity become $P_{n}(i,c)$ and $p(i,c)$
for $(i,c)\in\Gamma^{'}\coloneqq\{(i,c)\in\mathbb{N}_{0}^{2}:1\leq i\leq10,0\leq c\leq10\}$. 

Next we decide on the limiting probability $p$. Note that $\Gamma^{'}$
contains three initial types of nodes: defaulted (with $c=0$), vulnerable
(with $c\leq i$) and invulnerable (with $c>i$). In this numerical
experiment, we set the total fraction of initial defaults as $\xi$
and assume the fraction of initial defaults is the same across all
degrees, i.e. $p(i,0)=\frac{\xi}{10}$ for $i\in[1,10]$. For the
initially liquid nodes, the joint probability of the degree and initial
equity conditional on being liquid is constructed through a binormal
copula with correlation $\rho$ and two marginal probabilities. The
marginal probabilities of the degree and initial equity are assumed
to follow the Zipf's law, i.e.
\begin{eqnarray}
\mathbb{P}(\text{deg}=i) & = & \frac{i^{-(1+a_{1})}}{\sum_{i=1}^{10}i^{-(1+a_{1})}}\nonumber \\
\mathbb{P}(\text{initial equity}=c) & = & \frac{c^{-(1+a_{2})}}{\sum_{i=1}^{10}c^{-(1+a_{2})}},\nonumber \\
\end{eqnarray}
where $a_{1},a_{2}>0$. The Zipf's law is a form of the power law
with Pareto tails, which is observed for the distribution of the degrees
and equity levels of the financial networks in many empirical studies,
see e.g. \citet{Boss2004a,Bech2010}. 

In a network of size $n$ with the joint probability $P_{n}(i,c)$
of the degree and initial equity, a contagion process under interventions
occurs as described in \subsecref{Dynamics}. Recall that we only
need to consider intervening on a node that, when selected, has only
one unit of equity left, i.e. a node with ``distance to default''
equal to one. Here we consider two types of intervention policies,
the optimal policy and the alternative policy: intervening on nodes
with degree between 8 and 10 and ``distance to default'' equal to
one from the beginning of the process. The alternative policy is usually
the one adopted by the central bank or government in a real financial
crisis setting. 

Our objective is to verify the convergence in probability of $\frac{IT_{n}}{n}$
and $\frac{D_{n}}{n}$ as well as the convergence of the scaled termination
time $\frac{T_{n}}{m}$ as stated in \propref{Conv_obj}. Moreover,
we shall study the convergence rate of the standard deviation and
IQR (interquartile range) to examine if the asymptotic variables provide
good approximations to realistic $n$ values. 

Under the optimal policy in the form given in \thmref{theorem3},
the limits for $\frac{IT_{n}}{n}$, $\frac{D_{n}}{n}$ and $\frac{T_{n}}{m}$
as $n\rightarrow\infty$ are $it(y^{*},v^{*},z^{*})$, $\tilde{J}(y^{*},v^{*},z^{*})$
and $y^{*}$, respectively in (\ref{eq:OptProgram}) where $(y^{*},v^{*},z^{*})$
is the optimal solution. On the other hand, the alternative policy
is that for $0\leq k\leq m-1$,
\[
g_{k+1}^{(n),\text{alt}}(s,w)=\begin{cases}
\mathbbm1_{(k\geq n\lambda x_{\text{alt}}^{i,c,c-1})} & \text{if }w=(i,c,c-1)\in\Phi'\\
0 & \text{otherwise},
\end{cases}
\]
where for $(i,c,c-1)\in\Phi'$,

\[
x_{\text{alt}}^{i,c,c-1}=\begin{cases}
0 & \text{ if }i\in\{8,9,10\}\\
y & \text{otherwise}.
\end{cases}
\]

Then the limits for $\frac{IT_{n}}{n}$, $\frac{D_{n}}{n}$ and $\frac{T_{n}}{m}$
as $n\rightarrow\infty$ can be calculated as:

\begin{align*}
 & \frac{IT_{n}}{n}\overset{p}{\to}\\
 & \sum_{i=0}^{10}\sum_{c=1}^{i}p(i,c)\sum_{m=c}^{i}\sum_{n=0}^{c-1}(m-c+1)\mathbb{P}(\text{Multin}(i,x_{\text{alt}}^{i,c,c-1},y-x_{\text{alt}}^{i,c,c-1},1-y)=(n,m-n,i-m)),\\
 & \frac{D_{n}}{n}\overset{p}{\to}\sum_{i=0}^{10}\sum_{c=0}^{i}p(i,c)\mathbb{P}(\text{Bin}(i,x_{\text{alt}}^{i,c,c-1})\geq c),\\
 & \frac{T_{n}}{m}\overset{p}{\to}y
\end{align*}
where $y$ is the solution of $\frac{1}{\lambda}\sum_{i=0}^{10}i\sum_{c=0}^{i}p(i,c)\mathbb{P}(\text{Bin}(i,x_{\text{alt}}^{i,c,c-1})\geq c)=y$,
and $\mathbb{P}(\text{Bin}(i,y)\geq c)=\sum_{m=c}^{i}\binom{i}{m}y^{m}(1-y)^{i-m}$
, $\mathbb{P}(\text{Multin}(i,x,y,1-x-y)=(a,b,i-a-b))=\binom{i}{a,b,i-a-b}x^{a}y^{b}(1-x-y)^{i-a-b}$.

\section{Simulation}

\subsection{The set up}

We have the following setup.
\begin{enumerate}
\item A sequence of six networks with increasing number of nodes $n\in\{5^{4},6^{4},\ldots,10^{4}\}$
and there are 100 runs for each network under each intervention policy. 
\item To determine the asymptotic fraction $p(\cdot,\cdot)$ of the degree
and initial equity pair $(i,c)$ where $(i,c)\in\Gamma^{'}$, we set
the following parameters. 
\begin{enumerate}
\item The fraction of initial defaults $\xi=0.5$, indicating half of the
nodes have defaulted. We assume in this numerical experiment that
the fraction of initial defaults is the same across all degrees, thus
$p(i,0)=\frac{\xi}{10}$ for $i\in[1,10]$. 
\item The conditional probability of the degree and initial equity for liquid
nodes $p(i,e)$, $i\in\{1,\ldots,10\}$, $e\in\{1,\ldots,10\}$ is
determined by a binormal coupula with the exponents of the marginal
probabilities of the degree and initial equity $(a_{1},a_{2})=(0.8,0.7)$
and the correlation coefficient $\rho=0.9$. Note that a smaller $a_{1}$
indicates larger fraction of nodes with higher degrees, thus higher
connectivity and a smaller $a_{2}$ indicates larger fraction of nodes
with higher initial equities, and $\rho$ implies how likely that
higher degree nodes have higher initial equities.
\end{enumerate}
\item After determining the asymptotic fraction $p(\cdot,\cdot)$, we construct
a sequence of empirical fractions $P_{n}(\cdot,\cdot)$ for each network
that converge to $p(\cdot,\cdot)$ by
\begin{equation}
P_{n}(i,c)=\frac{[np(i,c)]}{n}\qquad(i,c)\in\Gamma^{'},\label{eq:empirical fraction}
\end{equation}
where $[\cdot]$ is the round function. In other words, the number
of nodes with degree $i$ and initial equity $c$ are $[np(i,c)]$
for a network of $n$ nodes. 
\item We consider two intervention policies described as before. 
\item The relative cost for the interventions $K=0.5$.
\end{enumerate}

\subsection{Simulation results}

In the following we suppress the $n$ in the subscripts. We show the
plots for $\frac{IT}{n}$, $\frac{D}{n}$ and $\frac{T}{m}$ under
the optimal and alternative policies.
\begin{enumerate}
\item Under either policy and for each variable, there are four plots in
each figure. The first two plots are two boxplots. The above boxplot
visualizes five summary statistics (min, mean$-$standard deviation,
mean, mean$+$standard deviation, max) while the bottom boxplot uses
a another set of summary statistics (1st quartile$-$1.5IQR,1st quartile,
median, 3rd quartile, 3rd quartile$+$1.5IQR) and the data outside
the range are treated as outliers, where IQR stands for interquartile
range, i.e. the difference between the third and the first quartiles. 
\item The blue dashed horizontal line in each plot indicates the theoretical
values for the limits of $\frac{IT}{n}$, $\frac{D}{n}$ and $\frac{T}{m}$
with $p(\cdot,\cdot)$ and the red solid line in each box indicates
those values calculated with $P_{n}(\cdot,\cdot)$ for each $n$.
We calculate the theoretical values in both ways because for small
$n$, $P_{n}(\cdot,\cdot)$ determined by (\ref{eq:empirical fraction})
has a relatively large rounding error and thus deviates from $p(\cdot,\cdot)$.
Calculating using $P_{n}(\cdot,\cdot)$ instead of $p(\cdot,\cdot)$
can effectively remove the deviations in the inputs to the model.
Moreover, $P_{n}(\cdot,\cdot)$ is different for different $n$ values,
thus the theoretical values of a variable calculated with $P_{n}(\cdot,\cdot)$
are also different for different $n$'s.
\item The black dots in the boxplots indicates the results of 100 runs and
they are jittered by a random amount left and right to avoid overplotting.
From the black dots we can see the distributions of the results. Note
that the black dots in the above and bottom boxplots show the same
results for the same $n$. They look different because they are jittered
by a different random amount.
\item The last two plots in every figure shows the log-log plot of the standard
deviation and IQR of $\frac{IT}{n}$, $\frac{D}{n}$ and $\frac{T}{m}$
against $n$ and a fitted straight line with the slope.
\end{enumerate}
From the simulation results, we make the following conclusions.
\begin{enumerate}
\item From the boxplots of $\frac{IT}{n}$, $\frac{D}{n}$ and $\frac{T}{m}$
under both interventions policies, we observe that the mean or median
converge to the calculated theoretical value with shrinking standard
deviation or IQR. Because the theoretical value is a constant given
the joint probability of degree and initial equity $p(\cdot,\cdot)$,
convergence of mean to the theoretical value with variance converging
to zero is equivalent to convergence in probability, this observation
provides evidence for the convergences in probability of $\frac{IT}{n}$,
$\frac{D}{n}$ and $\frac{T}{m}$ to their theoretical values.
\item Be comparing the blue dashed line and the red solid line we see that
the mean or median is closer to the red solid line, i.e. the theoretical
value calculated with $P_{n}(\cdot,\cdot)$ instead of $p(\cdot,\cdot)$.
This reflects the rounding error caused by (\ref{eq:empirical fraction})
in the inputs into the calculation. By using the more accurate fraction
we observe that the closeness of the mean or median to the theoretical
value does not vary in $n$ although the results of different runs
are more and more concentrated around the mean or median as $n$ grows.
\item The log-log plots of the standard deviation and IQR of each variable
with the fitted straight lines further show that both of them decrease
with power law tails, i.e. in the form of $z=Cx^{-a}$ where $C$
is a constant and $a>0$ is the exponent. The absolute value of the
slope of the straight line serves as the exponent. It is interesting
to observe that the exponents for the standard deviation and IQR are
close to each other. Moreover, the exponents are close to each other
under both intervention policies and for different variables. This
implies that the dispersions of all variables converge to zero at
roughly the same rate under both policies.
\end{enumerate}

\subsection{Summary}

To summarize the simulation part, we can make the following conclusions.
\begin{enumerate}
\item The convergences of $\frac{IT}{n}$, $\frac{D}{n}$ and $\frac{T}{m}$
to their theoretical values are supported by the simulation results.
It is worth noting that the closeness of the mean or median to the
theoretical value does not vary for different $n$ after the rounding
error in the initial fractions are removed, but the dispersion of
the variable shrinks as $n$ grows.
\item The dispersion of each variable decreases following a power law. The
exponents are close to each other under both intervention policies
and for all variables, indicating a uniform convergence rate for the
dispersions of all the variables under both policies.
\end{enumerate}
\newpage{}

\begin{figure}[H]
\includegraphics[scale=0.75]{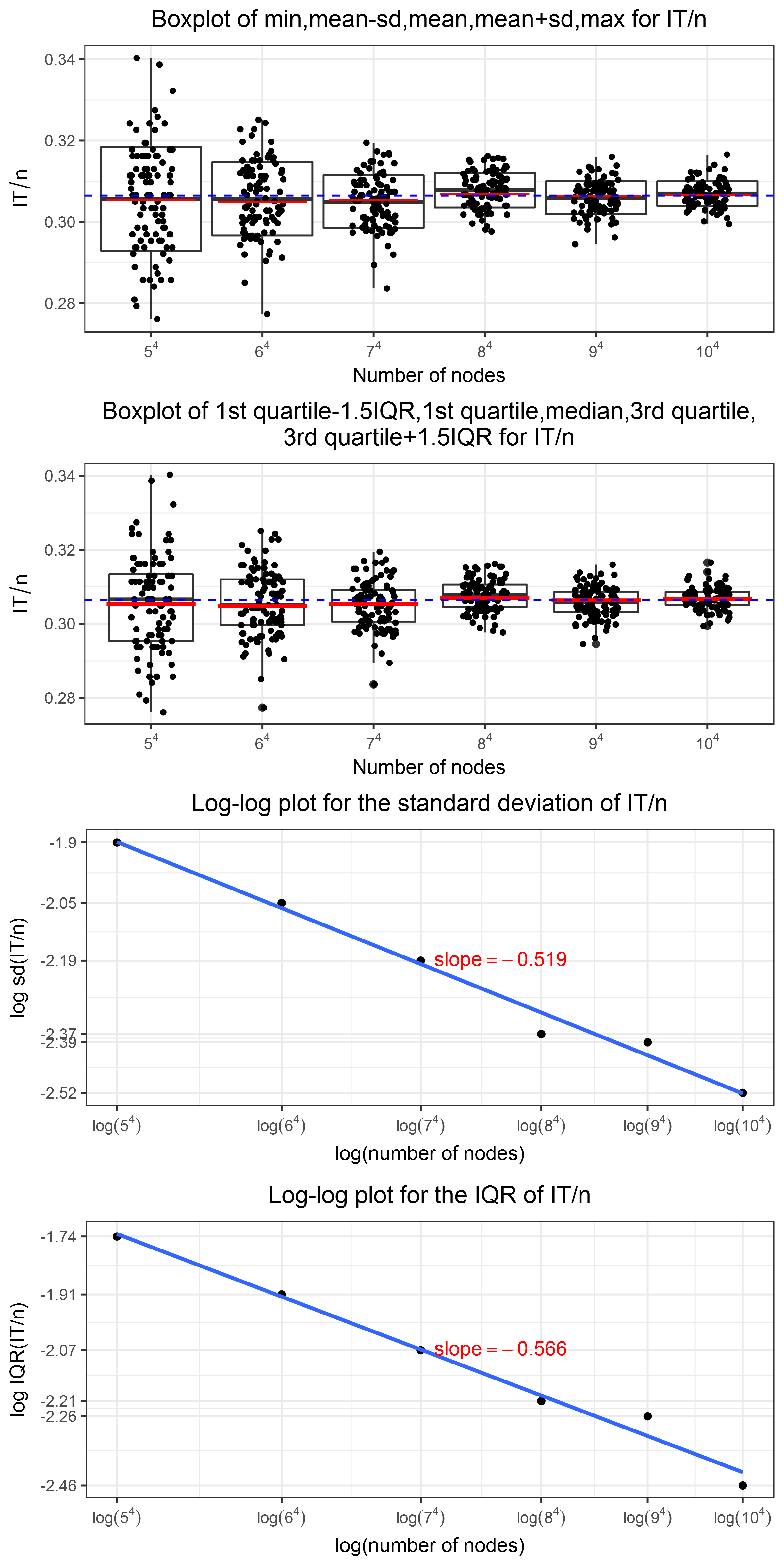}\caption{The boxplot and log-log plot of standard deviation and IQR for $IT/n$
under optimal policy}
\end{figure}

\begin{figure}[H]
\includegraphics[scale=0.75]{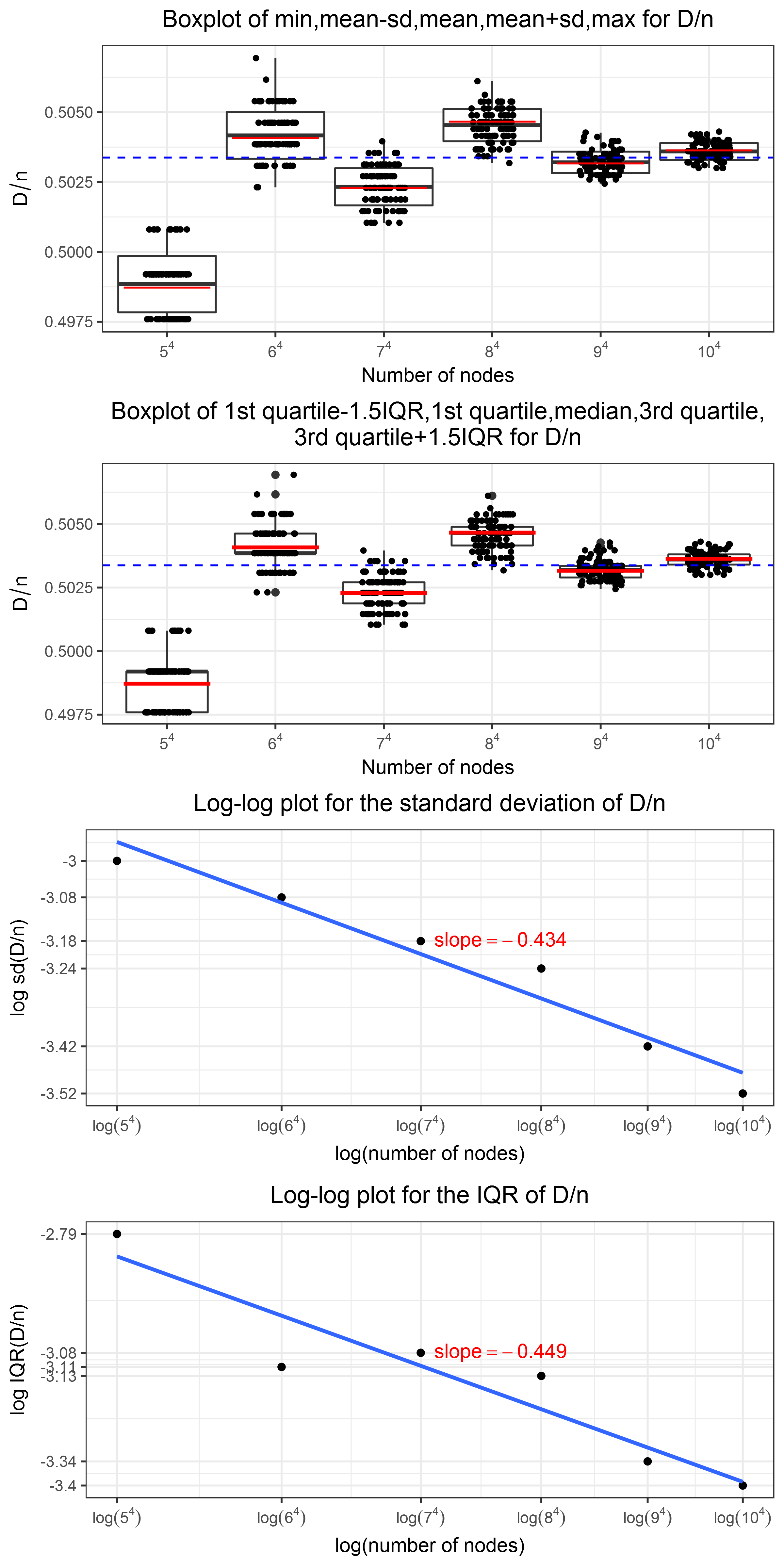}\caption{The boxplot and log-log plot of standard deviation and IQR for $D/n$
under optimal policy}
\end{figure}

\begin{figure}[H]
\includegraphics[scale=0.75]{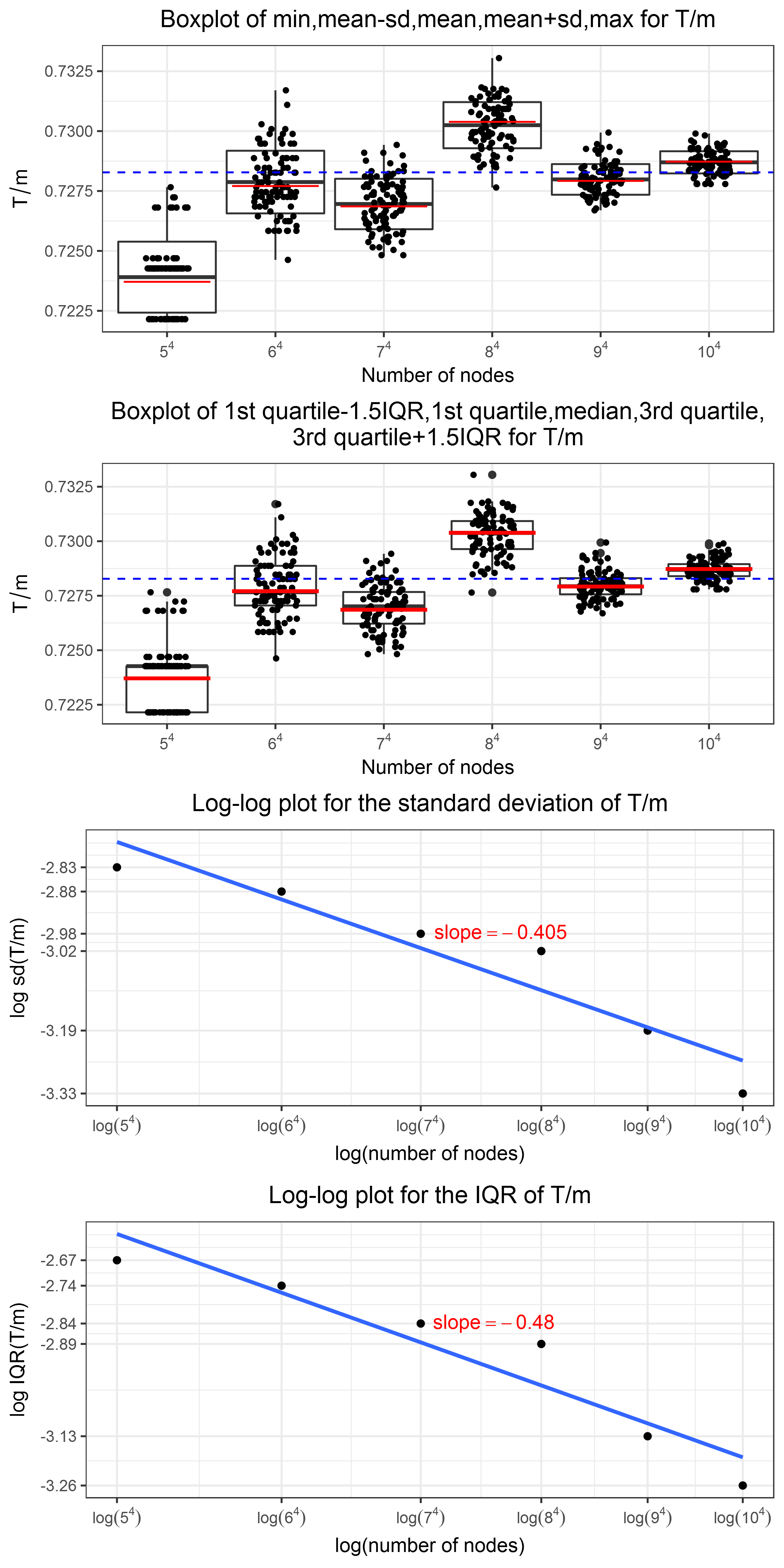}\caption{The boxplot and log-log plot of standard deviation and IQR for $T/m$
under optimal policy}
\end{figure}

\begin{figure}[H]
\includegraphics[scale=0.75]{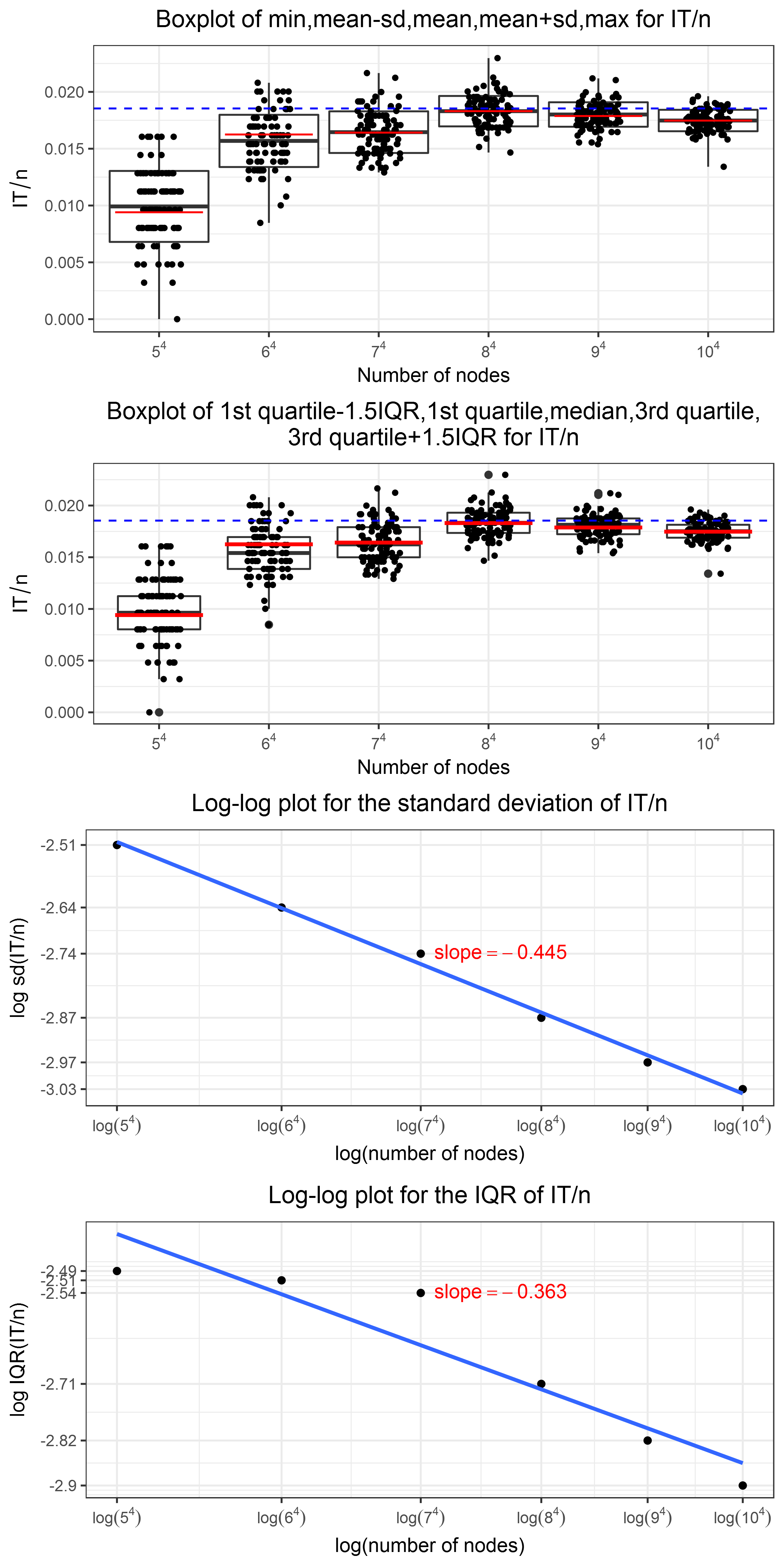}\caption{The boxplot and log-log plot of standard deviation and IQR for $IT/n$
under alternative policy}
\end{figure}

\begin{figure}[H]
\includegraphics[scale=0.75]{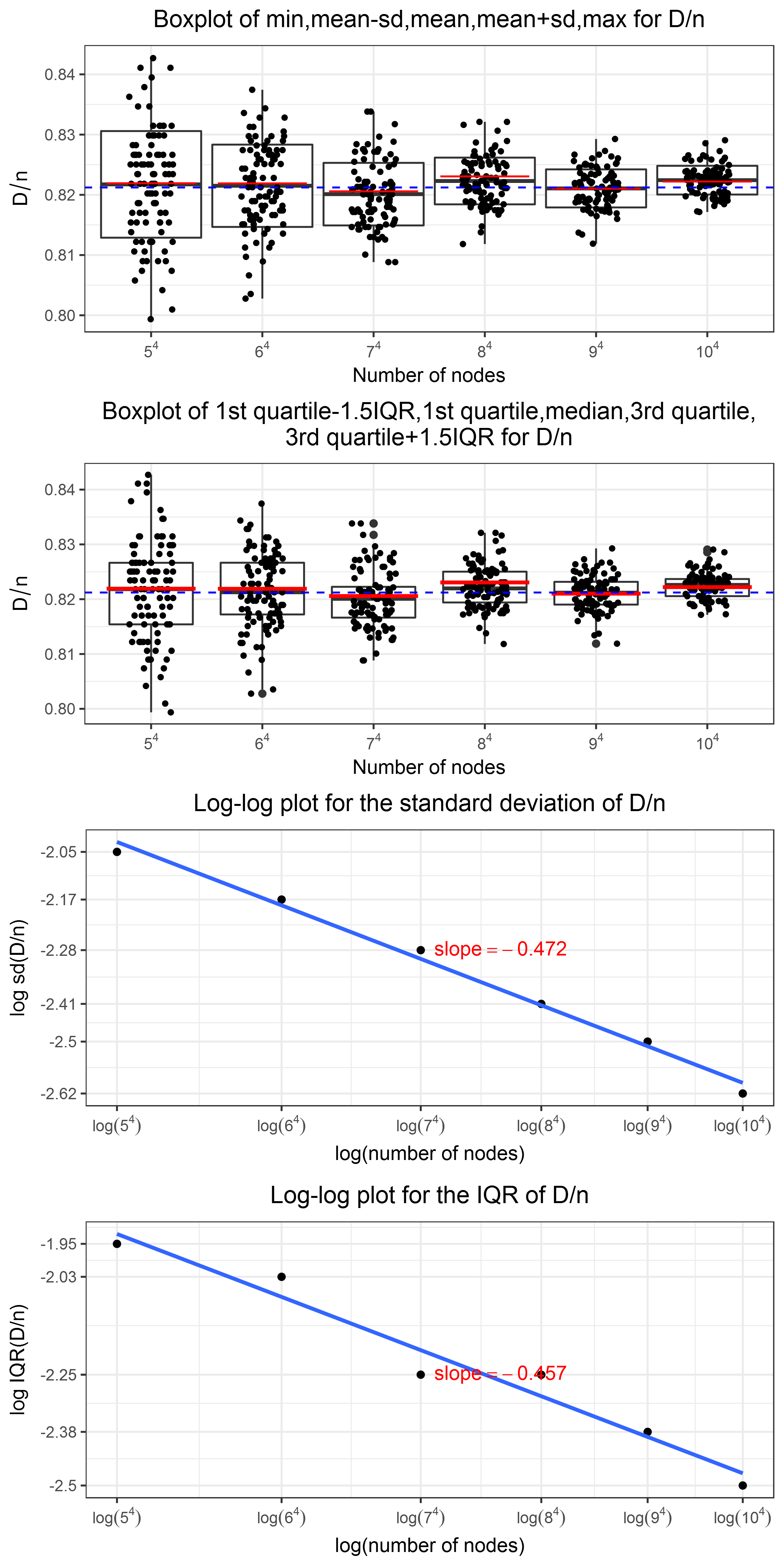}\caption{The boxplot and log-log plot of standard deviation and IQR for $D/n$
under alternative policy}
\end{figure}

\begin{figure}[H]
\includegraphics[scale=0.75]{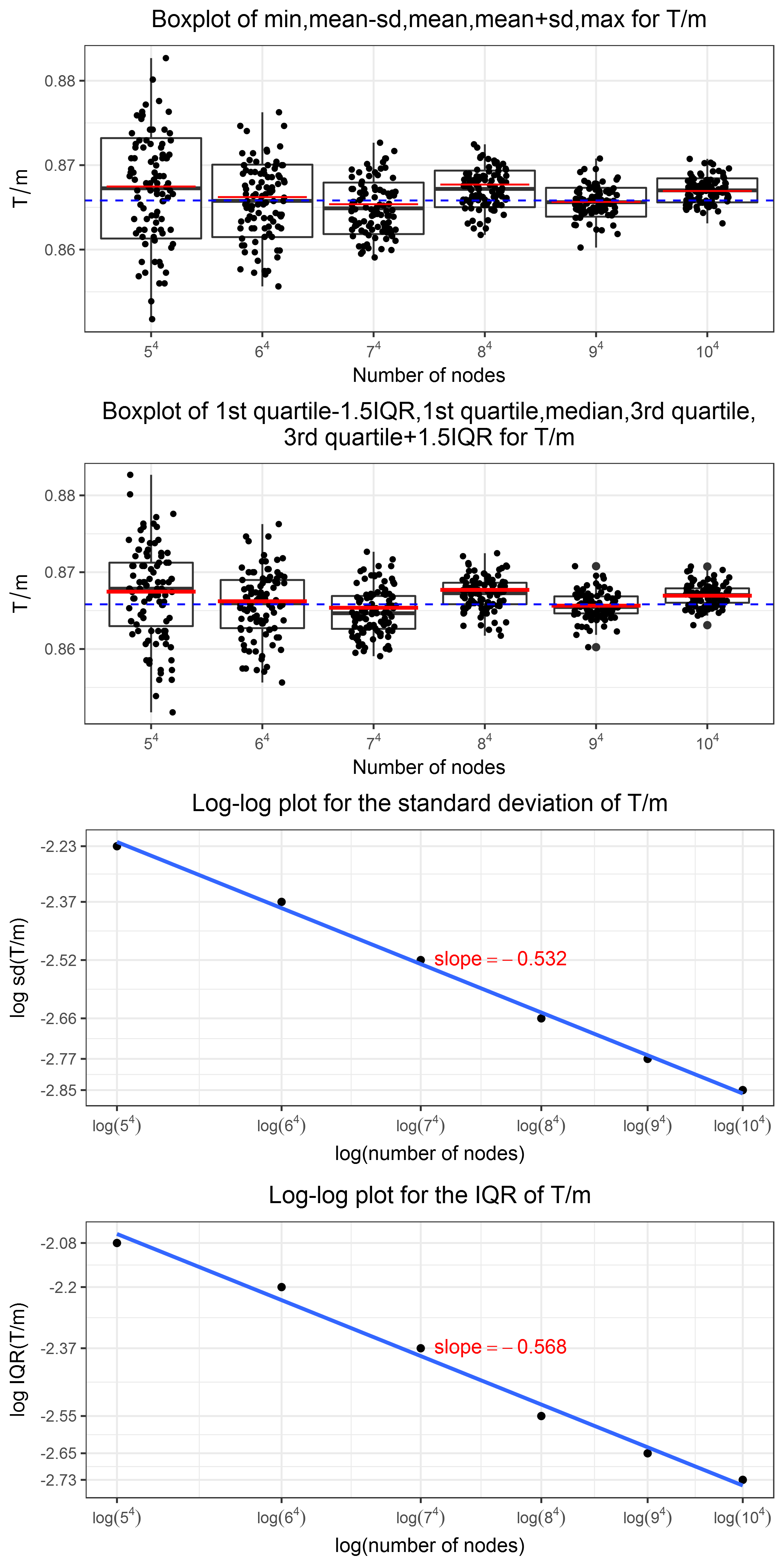}\caption{The boxplot and log-log plot of standard deviation and IQR for $T/m$
under alternative policy}
\end{figure}

\part{Summary}

We model the default contagion process in a large heterogeneous financial
network viewed by a regulator outside the network whose goal is to
minimize the number of final defaulted nodes with the minimum amount
of interventions. The regulator has only partial information in that
the connections of the nodes are unknown in the beginning but revealed
as the contagion process involves. The partial information setting
aligns with the reality better than most of the existing literature. 

Our work extend the previous literature, in particular \citet{Amini2013,Amini2015,Amini2017}
in that we provide analytical asymptotic results of the optimal intervention
policy for the regulator and the fraction of final defaulted nodes
for a heterogeneous network with a given arbitrary degree sequence
and arbitrary initial equity levels. Our results of the optimal intervention
policies generates insights in the perspectives of both random network
processes and regulations. The optimal intervention policy first depends
on the intervention cost: the lower the cost is, the more interventions
we implement. We only need to consider intervening on a bank if it
has distance to default of one when affected, i.e. the bank is very
close to default. Moreover, we observe that the optimal intervention
policy is monotonic with respect to the characteristics of the network:
the larger the out degree is, the smaller the in degree is, and the
higher the sum of initial equity levels and number of interventions
received, the earlier we should begin to intervene on the node. Moreover,
we should keep intervening on a node once we have intervened on it.
In other words, we do not allow a node that has received interventions
to default. We also quantify the improvements made by the optimal
intervention policy in terms of the features of the network. Our simulation
results show a good alignment with the theoretical calculations. The
numerical studies also provide evidence that although the optimal
intervention policy and the fraction of final default nodes are in
asymptotic sense, they provide good approximations for large but realistic
size of networks.

\appendix
\bibliographystyle{unsrtnat}
\bibliography{SysRisk}

\part{\label{part:Appendix}Appendix}

\section{\label{sec:Wormald's-theorem}Wormald's theorem}
\begin{lem}
\label{lem:(Wormald-1999)}(\citet{Wormald1999})Let $a\geq2$ be
a fixed integer and $\left(\left(Y_{t}^{l}\right)_{1\leq l\leq a}\right)_{t\geq0}$
denote a sequence of real valued random variables indexed by $n$
with its natural filtration $(\mathcal{F}_{t})_{t\geq0}$. Assume
that there is a constant $C_{0}>0$ such that $\left|Y_{t}^{l}\right|\leq C_{0}n$
for $\forall\,n,\,t\geq0$ and $1\leq l\leq a$. Let $f_{l}:\,\mathbb{R}^{a+1}\rightarrow\mathbb{R}$
be functions and $U\subseteq\mathbb{R}^{a+1}$ be some bounded connected
open set containing the closure of 
\[
\left\{ (0,z^{1},\ldots,z^{a}):\,\mathbb{P}(Y_{0}^{l}=z^{l}n,\,1\leq l\leq a)\neq0\text{ for some }n\right\} .
\]
Define the stopping time $T_{U}=\inf\left\{ t\geq1:\:(\frac{t}{n},\frac{Y_{t}^{1}}{n},\ldots,\frac{Y_{t}^{a}}{n})\notin U\right\} $.
Assume the following three conditions are satisfied: 
\begin{enumerate}
\item (Boundedness) For some function $\rho_{1}=\rho_{1}(n)\geq1$ and $\forall t<T_{U}$
and $1\leq l\leq a$, 
\[
|Y_{t+1}^{l}-Y_{t}^{l}|\leq\rho_{1}.
\]
\item (Trend) For some function $\rho_{2}=\rho_{2}(n)=o(1)$ and $\forall t<T_{U}$
and $1\leq l\leq a$, 
\[
\left|\mathbb{E}\left(Y_{t+1}^{l}-Y_{t}^{l}\mid\mathcal{F}_{t}\right)-f_{l}(\frac{t}{n},\frac{Y_{t}^{1}}{n},\ldots,\frac{Y_{t}^{a}}{n})\right|\leq\rho_{2}.
\]
\item (Lipschitz continuity) The functions $(f_{l})_{1\leq l\leq a}$ are
continuous and satisfies a Lipschitz condition on 
\[
U\cap\{(t,z^{1},\ldots,z^{a}):t\geq0\}
\]
with the same Lipschitz constant for each $l$. 
\end{enumerate}
Then the following holds: 
\begin{enumerate}
\item \label{enu:conclusion 1}For $(0,\hat{z}^{1},\ldots,\hat{z}^{a})\in U$
the system of differential equations 
\[
\frac{dz^{l}}{ds}=f_{l}(s,z^{1},\ldots,z^{a}),\qquad1\leq l\leq a
\]
has a unique solution in $U$ for $z^{l}:\mathbb{R}\rightarrow\mathbb{R}$
passing through 
\[
z_{0}^{l}=\hat{z}^{l},\qquad1\leq l\leq a
\]
and which extends to points arbitrarily close to the boundary of $U$. 
\item Let $\rho>\rho_{2}$ and $\rho=o(1)$. For a sufficiently large constant
C, with probability $1-O\left(\frac{\rho_{1}}{\rho}\exp\left(-\frac{n\rho^{3}}{\rho_{1}^{3}}\right)\right)$,
it holds that 
\[
\sup_{0\leq t\leq n\sigma}\left(\frac{Y_{t}^{l}}{n}-z_{\frac{t}{n}}^{l}\right)=O(\rho)
\]
where $z_{s}^{l}$ is the solution in (\ref{enu:conclusion 1}) with
\[
z_{0}^{l}=\frac{Y_{0}^{l}}{n}
\]
and 
\begin{equation}
\sigma=\sigma(n)=\sup\left\{ s\geq0:\,d_{\infty}\left(((z_{s}^{l})_{1\leq l\leq a}),\partial U\right)\geq C\rho\right\} ,\label{eq:suptime}
\end{equation}
 where $d_{\infty}(u,v)=\max_{1\leq i\leq j}|u_{i}-v_{i}|$ for $u=(u_{1},\ldots,u_{j})\in\mathbb{R}^{j}$
and $v=(v_{1},\ldots,v_{j})\in\mathbb{R}^{j}$.
\end{enumerate}
\end{lem}

\section{\label{sec:Extended-PMP}Extended Pontryagin maximum principle}
\begin{lem}
\label{lem:PontryaginMPM}(Extended Pontryagin maximum principle,
\citet{Chachuat2009}) Consider the optimal control problem to minimize
the cost functional including a terminal term
\[
\mathcal{J}(u,t_{f})\coloneqq\int_{t_{0}}^{t_{f}}\ell(t,x_{t},u_{t})dt+\phi(t_{f},x_{t_{f}}),
\]
with fixed initial time $t_{0}$ and free terminal time $t_{f}$,
subject to the dynamical system 
\[
\dot{x}_{t}=f(t,x_{t},u_{t});\quad x_{t_{0}}=x_{0},
\]
 where the vector function $x\in\hat{C^{1}}[t_{0},T]^{n_{x}}$ represents
the state variables characterizing the behavior of the system at any
time instant $t$, and some general terminal constraints
\begin{equation}
\psi_{k}(t_{f},x_{t_{f}})=0,\quad k=1,\ldots,n_{\psi}.\label{eq:Terminal}
\end{equation}
The admissible controls shall be taken in the class of piecewise continuous
functions
\[
u\in\mathcal{U}[t_{0},T]\coloneqq\{u\in\hat{C}[t_{0},T]^{n_{u}}:u_{t}\in U\text{ for }t_{0}\leq t\leq t_{f}\},
\]
 with $t_{f}\in[t_{0},T]$, where $T>t_{0}$ and the nonempty, possibly
closed and nonconvex set $U$ denotes the control region. 

Suppose $\ell$ and $f$ are continuous and have continuous first
partial derivatives with respect to $(t,x,u)$ on $[t_{0},T]\times\mathbb{R}^{n_{x}}\times\mathbb{R}^{n_{u}}$,
and also $\phi$ and $\psi\coloneqq(\psi_{k})_{k=1,\ldots,n_{\psi}}$
are continuous and have continuous first partial derivatives with
respect to $(t,x)$ on $[t_{0},T]\times\mathbb{R}^{n_{x}}$. Suppose
that the terminal constraints (\ref{eq:Terminal}) satisfy the constraint
qualification 
\[
\text{rank}\left(\frac{\partial\psi}{\partial x}(t_{f}^{*},x_{t_{f}^{*}}^{*})\right)=n_{\psi}
\]
where $\frac{\partial\psi}{\partial x}(t_{f}^{*},x_{t_{f}^{*}}^{*})$
denotes the Jacobian matrix of the partial derivatives of components
of $\psi$ with respect to $x$ evaluated at $(t_{f}^{*},x_{t_{f}^{*}}^{*})$.
Define the Hamiltonian function 
\begin{equation}
\mathcal{H}(t,x,u,\mathring{\lambda},\lambda)=\mathring{\lambda}\ell(t,x,u)+\lambda^{T}f(t,x,u).\label{eq:Ham_MPM}
\end{equation}

Let $(u^{*},t_{f}^{*})\in\hat{C}[t_{0},T]^{n_{u}}\times[t_{0},T)$
denote a minimizer for the problem, and $x^{*}\in\hat{C^{1}}[t_{0},T]$
the optimal state, then there exists a $n_{x}$ dimensional piecewise
continuously differentiable vector function $\lambda_{t}^{*}$ and
$\mathring{\lambda}^{*}\in\{0,1\}$ ($(\mathring{\lambda}^{*},\lambda_{t}^{*})$
are called adjoint variables) and a Lagrange multiplier vector $v^{*}\in\mathbb{R}^{n_{\psi}}$
such that $(\mathring{\lambda}^{*},\lambda_{t}^{*})\neq0$ for every
$t\in[t_{0},t_{f}^{*}]$ and the following conditions hold:
\begin{enumerate}
\item \label{enu:Cond1_MPM}The function $\mathcal{H}(t,x_{t}^{*},w,\mathring{\lambda}^{*},\lambda_{t}^{*})$
attains its minimum on $U$ at $w=u_{t}^{*}$ for every $t\in[t_{0},t_{f}^{*}]$,
i.e. 
\[
\mathcal{H}(t,x_{t}^{*},w,\mathring{\lambda}^{*},\lambda_{t}^{*})\geq\mathcal{H}(t,x_{t}^{*},u_{t}^{*},\mathring{\lambda}^{*},\lambda_{t}^{*}),\quad\forall w\in U.
\]
\item \label{enu:Cond2_MPM}$(x_{t}^{*},u_{t}^{*},\mathring{\lambda}^{*},\lambda_{t}^{*})$
verifies the equations
\begin{alignat}{1}
\frac{d}{dt}x_{t}^{*} & =f(t,x_{t}^{*},u_{t}^{*})\nonumber \\
\frac{d}{dt}\lambda_{t}^{*} & =-\frac{\partial}{\partial x}\mathcal{H}(t,x_{t}^{*},u_{t}^{*},\mathring{\lambda}^{*},\lambda_{t}^{*})\label{eq:Lambda_MPM}
\end{alignat}
at each instant $t$ of continuity of $u^{*}$ and $\mathring{\lambda}^{*}\in\{0,1\}$.
\item \label{enu:Cond3_MPM}$\mathcal{H}(t,x_{t}^{*},u_{t}^{*},\mathring{\lambda}^{*},\lambda_{t}^{*})=\mathcal{H}(t_{f}^{*},x_{t_{f}^{*}}^{*},u_{t_{f}^{*}}^{*},\mathring{\lambda}^{*},\lambda_{t_{f}^{*}}^{*})-\int_{t}^{t_{f}^{*}}\frac{\partial}{\partial t}\mathcal{H}(\tau,x_{\tau}^{*},u_{\tau}^{*},\mathring{\lambda},\lambda_{\tau}^{*})d\tau$.
Therefore, if $\frac{\partial}{\partial t}\mathcal{H}=0$, i.e. $\mathcal{H}$
is autonomous, then $\mathcal{H}$ is a constant over time.
\item \label{enu:Cond4_MPM}(Transversal condition) Define $\Psi(t,x)\coloneqq\mathring{\lambda}^{*}\phi(t,x)+v^{*T}\psi(t,x)$,
then
\begin{eqnarray}
\lambda_{t_{f}^{*}}^{*} & = & \frac{\partial}{\partial x}\Psi(t_{f}^{*},x_{t_{f}^{*}}^{*})\nonumber \\
\mathcal{H}(t_{f}^{*},x_{t_{f}^{*}}^{*},u_{t_{f}^{*}}^{*},\mathring{\lambda}_{t_{f}^{*}}^{*},\lambda_{t_{f}^{*}}^{*}) & = & -\frac{\partial}{\partial t}\Psi(t_{f}^{*},x_{t_{f}^{*}}^{*})
\end{eqnarray}
together with the terminal condition (\ref{eq:Terminal}) at $t=t_{f}^{*}$,
i.e. $\psi_{k}(t_{f}^{*},x_{t_{f}^{*}}^{*})=0$ for $k=1,\ldots,n_{\psi}$.
\item \label{enu:Cond5_MPM}The optimal control $u^{*}$ may or may not
be continuous; in the latter case we have a corner point. In particular,
the conditions that must hold at any corner point $\theta\in[t_{0},t_{f}^{*}]$
are
\begin{eqnarray}
x_{\theta^{-}}^{*} & = & x_{\theta^{+}}^{*},\nonumber \\
\lambda_{\theta^{-}}^{*} & = & \lambda_{\theta^{+}}^{*},\nonumber \\
\mathcal{H}(\theta^{-},x_{\theta}^{*},u_{\theta^{-}}^{*},\mathring{\lambda}^{*},\lambda_{\theta}^{*}) & = & \mathcal{H}(\theta^{+},x_{\theta}^{*},u_{\theta^{+}}^{*},\mathring{\lambda}^{*},\lambda_{\theta}^{*}).
\end{eqnarray}
\end{enumerate}
\end{lem}
\begin{proof}
See theorem 3.33 and theorem 3.34 in \citet{Chachuat2009}.
\end{proof}

\end{document}